\documentclass[11pt]{amsart}

\usepackage[margin=1in]{geometry}

\usepackage[boxsize=1em,centertableaux]{ytableau}

\usepackage[T1]{fontenc}
\usepackage{graphicx}
\usepackage{amssymb,amsmath,amsthm,mathrsfs,amsfonts,mathtools}
\usepackage{bm}
\usepackage{booktabs}
\usepackage{tikz-cd}
\usepackage{physics}
\usepackage{braket}
\usepackage{dsfont}
\usepackage{microtype}
\usepackage[margin=1in]{geometry}
\usepackage{caption}
\usepackage{xcolor}
\usepackage{hyperref}
\hypersetup{
  colorlinks=true,
  citecolor=blue,
  filecolor=blue,
  linkcolor=blue,
  urlcolor=blue
}
\IfFormatAtLeastTF{2026-06-01}{}{\usepackage{aliascnt}}
\usepackage[capitalise,noabbrev]{cleveref}

\newcommand{\id}{\operatorname{id}}
\newcommand{\Ran}{\operatorname{Ran}}
\newcommand{\supp}{\operatorname{supp}}
\newcommand{\Ad}{\operatorname{Ad}}
\newcommand{\CC}{\mathbb{C}}
\newcommand{\cO}{\mathcal{O}}

\newcommand{\hashtag}{\#}
\renewcommand{\Tr}{\operatorname{Tr}}
\renewcommand{\norm}[1]{\left\lVert #1\right\rVert}
\DeclareMathOperator{\spec}{spec}

\newtheorem{thm}{Theorem}[section]
\IfFormatAtLeastTF{2026-06-01}{\newtheorem{prop}[thm]{Proposition}
  \newtheorem{lem}[thm]{Lemma}
  \newtheorem{cor}[thm]{Corollary}
  \theoremstyle{definition}
  \newtheorem{defn}[thm]{Definition}
  \theoremstyle{remark}
  \newtheorem{rem}[thm]{Remark}
}{\newaliascnt{prop}{thm}
  \newtheorem{prop}[prop]{Proposition}
  \aliascntresetthe{prop}
  \newaliascnt{lem}{thm}
  \newtheorem{lem}[lem]{Lemma}
  \aliascntresetthe{lem}
  \newaliascnt{cor}{thm}
  \newtheorem{cor}[cor]{Corollary}
  \aliascntresetthe{cor}
  \theoremstyle{definition}
  \newaliascnt{defn}{thm}
  \newtheorem{defn}[defn]{Definition}
  \aliascntresetthe{defn}
  \theoremstyle{remark}
  \newaliascnt{rem}{thm}
  \newtheorem{rem}[rem]{Remark}
  \aliascntresetthe{rem}
}

\Crefname{thm}{Theorem}{Theorems}
\Crefname{prop}{Proposition}{Propositions}
\Crefname{lem}{Lemma}{Lemmas}
\Crefname{cor}{Corollary}{Corollaries}
\Crefname{defn}{Definition}{Definitions}
\Crefname{rem}{Remark}{Remarks}

\numberwithin{equation}{section}
\allowdisplaybreaks

\newcommand{\note}[1]{\overset{\text{#1}}}

\newcommand{\proofstep}[1]{\par\addvspace{\medskipamount}\noindent{\bfseries\boldmath #1\par}\nobreak\smallskip\noindent\ignorespaces}

\crefname{proposition}{Proposition}{Propositions}
\graphicspath{{figures/}}

\title{Quantum Markov State Models for Metastable Dynamics}
\date{\today}
\author{Hao-En Li}
\address{Department of Mathematics, University of California, Berkeley, CA 94720, USA}
\email{haoen2021@math.berkeley.edu}

\author{Lin Lin}
\address{Department of Computing and Mathematical Sciences, California Institute of Technology, Pasadena, CA 91125, USA}
\address{Applied Mathematics and Computational Research Division, Lawrence Berkeley National Laboratory, Berkeley, CA 94720, USA}
\address{Department of Mathematics, University of California, Berkeley, CA 94720, USA}
\email{lin@caltech.edu}

\author{Michael Ragone}
\address{Department of Mathematics, University of California, Berkeley, CA 94720, USA}
\email{micragone@berkeley.edu}

\begin{document}
\begin{abstract}
Open quantum systems can rapidly lose most microscopic information, leaving only a few degrees of freedom to govern their long-time dynamics. Classical Markov state models (MSMs) describe such metastable dynamics as transitions among a few representative phases and are widely used to reduce complex-system dynamics in condensed matter and chemical physics. In quantum systems, however, phase labels alone are insufficient when coherence persists between metastable states. Even when the slow modes are known, their spectral projection need not produce valid quantum states. We construct quantum Markov state models (QMSMs) that describe the surviving information and its evolution on a small physical state space containing classical sectors and quantum matrix blocks. We quantify metastability by assuming that the evolution channel $\mathcal C$ changes little when applied a second time, with sufficiently small defect $\eta=\|\mathcal C^2-\mathcal C\|_\diamond$, and that the number of slow degrees of freedom is bounded independently of the full system size. Under these assumptions, we construct compression and reconstruction channels whose composition recovers every reduced state exactly, while the reverse composition gives an exactly idempotent channel approximating $\mathcal C$, answering Kitaev's exact-rounding question \cite{Kitaev2025}. Our constructions give an optimal diamond-norm bound of $\mathcal{O}(\eta^{1/3})$ on all input states, as well as an improved bound of $\mathcal{O}(\eta^{1/2})$ on metastable states prepared by $\mathcal C$, with constants depending only on the slow dimension. The reduced transition channel can be iterated to predict the microscopic dynamics with controlled error. We illustrate the QMSM through a weakly driven dissipative spin chain supporting either a metastable logical qubit or long-lived classical phases.
\end{abstract}
\maketitle
\setcounter{tocdepth}{1}
\tableofcontents

\section{Introduction}\label{sec:introduction}

\subsection{Metastability and quantum Markov state models}\label{sec:intro-background}
Metastability is ubiquitous in nature, especially for complex systems, and involves a separation of dynamical timescales. A system first relaxes rapidly in many directions and then evolves much more slowly within a small family of long-lived configurations. In classical statistical physics, this intermediate regime is important for understanding phase-transition kinetics \cite{Binder1987,Debenedetti2001}, chemical reaction mechanisms \cite{Hanggi1990}, and protein-folding dynamics \cite{Noe2009,LindorffLarsen2011}. This separation of timescales naturally suggests a potential reduced description: one could identify the degrees of freedom that survive the fast relaxation and then use them to model the effective slow evolution.

For classical stochastic dynamics, Markov state models (MSMs) replace the microscopic state space involving a huge number of degrees of freedom by a set of representative phases and describe their dynamics using a transition matrix \cite{Noe2007,HusicPande2018}. In molecular-dynamics simulations, especially for macromolecules like proteins, such models allow long-time kinetics to be inferred from much shorter molecular trajectories. Successful applications of MSMs include the reconstruction of folding pathways \cite{Noe2009}, ligand-binding kinetics \cite{PlattnerNoe2015}, and protein--protein association \cite{Plattner2017} (see also \cite{Noe2014}). Rigorous mathematical approaches to classical metastability are also well-established to quantify the lifetimes of metastable states and the transitions among them \cite{BovierDenHollander2015}.

Open quantum systems can exhibit a similar separation of timescales, for instance, in driven many-body matter \cite{Letscher2017,Sieberer2025} and the preservation of quantum information \cite{Dennis2002,Lescanne2020}. In this setting, classical phase labels may be insufficient: coherence can persist between metastable states, and the slow degrees of freedom may encode an effective qubit \cite{macieszczak2016towards,LidarChuangWhaley1998,LidarWhaley2003}. A quantum Markov state model (QMSM) should therefore generalize classical MSMs by
identifying both the classical and quantum information that survives fast relaxation. Parallel to the classical case, a QMSM is expected to consist of a reduced state space, a reduced channel acting on it, and maps connecting the full and reduced state spaces. We use \emph{compression} for the channel mapping full states to reduced states, and \emph{reconstruction} for the channel mapping reduced states to full states. We seek rigorous constructions and error bounds relating the reduced dynamics to the physical microscopic evolution. Figure~\ref{fig:msm-qmsm-overview} illustrates the classical and quantum reductions.
\begin{figure}[t]
  \centering
  \includegraphics[width=\textwidth]{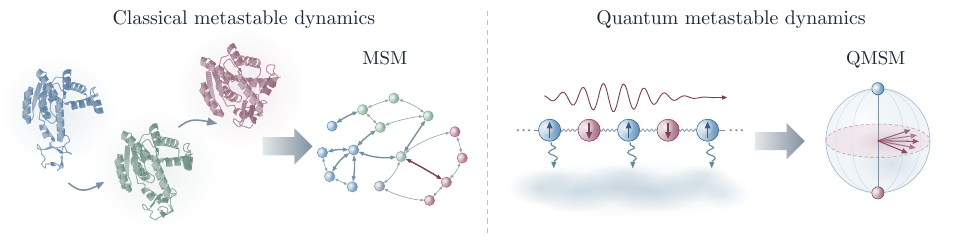}
  \captionsetup{font=small,justification=raggedright,singlelinecheck=false}
  \caption{Classical and quantum Markov state models. Left: metastability in molecular dynamics is represented by transitions among a few classical states. Right: a driven dissipative spin chain can be reduced to a metastable logical qubit. The Bloch-sphere fan depicts dephasing. }
  \label{fig:msm-qmsm-overview}
\end{figure}

To construct such a model, we start from the propagator in the metastable regime. We quantify the metastability of an evolution quantum channel by assuming that it is \emph{almost idempotent} and that its slow spectral subspace has dimension bounded independently of the full system size. Under these assumptions, we construct an \emph{exactly idempotent} channel onto the reduced state space, and use it as the compression to define a QMSM. The reduced state space can contain classical sectors, corresponding to $1\times 1$ blocks, and quantum matrix blocks, and iterating a reduced transition channel gives quantitative predictions for the full dynamics, just like the classical MSM. The decomposition into classical phases and quantum memories is physically natural. For instance, quantum coding theory uses noiseless subsystems and decoherence-free subspaces to store protected logical states \cite{kribs2005unified}, and analogous structures can coexist with classical metastable phases in open quantum dynamics \cite{macieszczak2016towards}.
From this perspective, one may view the reduced state space as the logical state space. 
The difficulty here is to construct a QMSM realizing this structure when the original channel is only almost idempotent, with physical compression and reconstruction maps and controlled error in the reduced dynamics. This is summarized in the informal \cref{thm:QMSM_inf} below.
\begin{thm}[QMSM theorem, informal]\label{thm:QMSM_inf}
  For any quantum channel $\mathcal C$ with sufficiently small idempotency defect and a bounded number of slow degrees of freedom, we can construct a QMSM with a compression channel $\mathcal Q$, a reconstruction channel $\mathcal R$, and a reduced transition channel $\mathcal M = \mathcal Q\mathcal C\mathcal R$ on the reduced state space $\mathscr D_\star$, such that:
  \begin{itemize}
    \item The transition channel $\mathcal M$ defines a valid effective quantum Markov process on the reduced quantum state space $\mathscr D_\star$;
    \item The QMSM identifies the structure of the reduced degrees of freedom, whose components may include both \emph{classical phases} and \emph{quantum memories};
    \item The QMSM provides quantitative error bounds relating the reduced dynamics to the original physical evolution, uniformly in the full system size.
  \end{itemize}  
\end{thm}

\subsection{Relation to previous work}\label{sec:intro-related}

Quantum metastability is still a nascent field. To our knowledge, there is no general framework for constructing physical QMSMs from almost idempotent channels. Previous work includes spectral descriptions of slow dynamics \cite{macieszczak2016towards,MacieszczakRoseLesanovskyGarrahan2021}, physical compression and reconstruction maps \cite{GrigolettoSarletteTicozziViola2026,Kitaev2025}, and structural results for long-lived states \cite{BergamaschiChenVazirani2025,YinSuraceLucas2025}. Table~\ref{tab:qmsm-comparison} compares their assumptions and conclusions with the requirements for constructing a physical QMSM.

Here \emph{physicality} means that the constructed effective dynamics can be described by a quantum channel. \emph{Generality} means applicability to arbitrary channels with sufficiently small idempotency defect, without a prescribed perturbative reference dynamics. Our results additionally require bounded slow rank.  \emph{Reduced dynamics} means the method can be used to derive an effective evolution with quantitative approximation bounds.  \emph{Not addressed} means the corresponding QMSM question is outside the scope of the cited results.

\begin{table}[h]
  \centering
  \footnotesize
  \setlength{\tabcolsep}{4pt}
  \renewcommand{\arraystretch}{1.15}
  \begin{tabular}{lccc}
    \toprule
     & Physicality & Generality & Reduced dynamics \\
    \midrule
    \cite{macieszczak2016towards,MacieszczakRoseLesanovskyGarrahan2021} & Not guaranteed & Yes & Yes \\
    \cite{GrigolettoSarletteTicozziViola2026} & Yes & Perturbative regime & Yes \\
    \cite{Kitaev2025} & Yes & Yes & Not addressed \\
    \cite{BergamaschiChenVazirani2025,YinSuraceLucas2025} & Not addressed & State dependent & Not addressed \\
    This work & Yes & Yes  & Yes \\
    \bottomrule
  \end{tabular}
  \caption{\raggedright Physical reduction and effective dynamics in selected approaches to quantum metastability, under the criteria defined in the text.}
  \label{tab:qmsm-comparison}
\end{table}

The spectral approach starts from a separation between fast and slow eigenmodes of a Lindblad generator. Macieszczak and collaborators use the slow modes to characterize metastable subspaces \cite{macieszczak2016towards,RoseMacieszczakLesanovskyGarrahan2016}. Their theory of classical metastability yields a stochastic dynamics between phases, together with quantitative approximation bounds, when the metastable manifold is approximately simplex-like \cite{MacieszczakRoseLesanovskyGarrahan2021}. Brown, Macieszczak, and Jack study metastability at the level of quantum trajectories \cite{Brown2024}, while Jin, Qiu, and Ma develop a discrete-time formulation based on quantum channels \cite{Jin2024}. However, the slow spectral projection need not be completely positive and trace preserving (CPTP). Consequently, the compression and reconstruction maps derived from it are not guaranteed to define valid reduced quantum states or physical transitions between them, as required by the first item of \cref{thm:QMSM_inf}.

Grigoletto, Sarlette, Ticozzi, and Viola construct CPTP reduction and reconstruction maps and a reduced Lindblad generator \cite{GrigolettoSarletteTicozziViola2026}. Their reduction onto the center manifold is asymptotically exact and includes persistent oscillations. For analytic perturbations, they retain the unperturbed center manifold as the reduced space and give finite-time error bounds. Thus their construction uses a Lindblad generator and, in the perturbative setting, a reference dynamics. Our assumptions apply directly to an almost idempotent channel, and our theorem guarantees a (discrete-time) reduced transition channel.

Kitaev relates almost-idempotent quantum channels to approximate $C^*$-algebras and obtains an approximate encoding--decoding factorization through a finite-dimensional algebra \cite[Theorems~10.5 and~12.3]{Kitaev2025}. Both factors are quantum channels, and their composition on the reduced state space \emph{approximates} the identity. He also asks whether an almost-idempotent UCP map is close to an exactly idempotent UCP map on the same matrix algebra \cite[Section~1.2]{Kitaev2025}. We find that such a dimensionally uniform approximation is not generally possible (see a counterexample in Appendix~\ref{sec:dimension-obstruction}). With the additional bounded slow rank assumption, we make the recovery relation exact while preserving complete positivity, obtaining a nearby idempotent channel. We use this channel to define a QMSM and bound the error in its iterated dynamics. This exact recovery relation connects to operator-algebra quantum error correction (OAQEC), which protects an algebra of logical observables \cite{kribs2005unified,BenyKempfKribs2007}. Here the logical algebra is determined by the metastable dynamics.

 In an alternative direction, one can ask whether metastable states enjoy any structural properties. Bergamaschi, Chen, and Vazirani show that approximately stationary states of a family of detailed-balanced quasi-local Lindbladians are spatially Markovian, have mutual information area laws, and admit recovery maps which can correct local errors~\cite{BergamaschiChenVazirani2025}.  Yin, Surace, and Lucas formulate metastability through local energetic stability of pure states in closed many-body systems and derive prethermal lifetime bounds \cite{YinSuraceLucas2025}. 

\subsection*{Organization of the paper.}
Section~\ref{sec:main-results} states the main operator-algebraic theorems, from which the QMSM construction and reduced transition dynamics follow. Section~\ref{sec:kitaev} develops Kitaev's approximate encoding--decoding factorization and the Stinespring representation, the main tools used in our proofs, and provides an overview of the proof strategy. Section~\ref{sec:exactification} proves the main theorems through a controlled surgery on the Stinespring dilations of the encoder and decoder. Section~\ref{sec:spin-chains} presents a weakly driven dissipative spin-chain example illustrating the framework. The appendices contain the technical spin-chain calculations and estimates, as well as the sharpness analysis for the error bounds. These materials are supplementary to the main development and may be skipped by readers not interested in those details.

\subsection*{Acknowledgments}

This material is based upon work supported by the U.S. Department of Energy, Office of Science, Accelerated Research in Quantum Computing Centers, Quantum Utility through Advanced Computational Quantum Algorithms, grant no. DE-SC0025572 (H.L., L.L.), and by the Challenge Institute for Quantum Computation (CIQC) funded by National Science Foundation (NSF) through grant number OMA-2016245 (L.L.). L.L. is a Simons Investigator in Mathematics.
 H.L. thanks the Department of Computing and Mathematical Sciences at Caltech, where part of this work was done, for its hospitality.  The authors thank Garnet Chan, Zherui Chen, Zhiyan Ding, Srivatsav Kunnawalkam Elayavalli, Marc Rieffel, and Yilun Yang for insightful discussions. 

\subsection*{AI statement}
We used OpenAI GPT-5.5, GPT-5.6 Sol, and GPT-6 Astra to explore candidate proof strategies, examples, and counterexamples, and to assist with manuscript preparation. We checked the AI-assisted material and take full responsibility for the content of this work.

\section{Notations and preliminaries}
\subsection{Notations}
Let $M_d=\mathcal B(\CC^d)$, where $\mathcal B(\mathcal H)$ denotes the algebra of all the bounded linear operators on the Hilbert space $\mathcal H$, 
 and let the set of density matrices (or state space) be $\mathscr D_d=\{\rho\in M_d:\rho\geq0,\ \Tr\rho=1\}$. Let $\mathcal A\subset M_d$ denote a unital $\ast$-subalgebra of the matrix algebra $M_d$, which is a finite-dimensional $C^\ast$-algebra. 
A quantum channel is a completely positive, trace-preserving (CPTP) linear map between finite-dimensional $C^*$-algebras. 
We denote by $\Tr_{\mathcal A}$ the trace on a finite-dimensional $C^*$-algebra $\mathcal A$, and define the Hilbert--Schmidt inner product on $\mathcal A$ by $\langle X,Y\rangle_{\rm HS }:=\Tr_{\mathcal A}(X^\ast Y)$, with the Frobenius norm induced by it denoted by $\norm{A}_{\rm HS} := \sqrt{\Tr(A^\ast A)}$. 
We use $T_\ast:\mathcal A_2 \to\mathcal A_1$ for the Hilbert--Schmidt adjoint of a linear map $T :\mathcal A_1 \to\mathcal A_2$, satisfying $\langle T(X),Y\rangle_{\mathcal A_2}=\langle X,T_\ast(Y)\rangle_{\mathcal A_1}$ for all $X\in\mathcal  A_1,Y\in \mathcal A_2$. We say $T:\mathcal A_1\to\mathcal A_2$ is unital if $T(I_{\mathcal A_1})=I_{\mathcal A_2}$. A map is unital and completely positive (UCP) if and only if its adjoint $T_\ast$ is CPTP. Physically, the formalism of quantum states and quantum channels corresponds to the Schr\"odinger picture, and the Heisenberg picture corresponds to the formalism of observables and their evolution under the UCP maps. If a map $T:\mathcal A\to \mathcal A$ is UCP and idempotent, then the range $\Ran T$ is a $C^\ast$-algebra with the Choi--Effros product $X\circ Y := T(XY)$ \cite{ChoiEffros1977} (see also \cite[Exercise 4.5]{Paulsen2003}).

For any abstract finite-dimensional $C^*$-algebra $\mathcal A 
\cong \bigoplus_jM_{n_j}$, the trace is the sum of the ordinary traces on its blocks 
\begin{equation}\label{eq:trace-sum}
\Tr_{\mathcal A}:=\sum_j\Tr_{M_{n_j}},\quad \text{where $\Tr$ is the ordinary trace on each block $M_{n_j}$.}
\end{equation} up to the isomorphism. Channels on states are adjoint to unital completely positive (UCP) maps on observables, and $\|T_\ast\|_\diamond=\|T\|_{\mathrm{cb}}$ \cite{Watrous2018,Paulsen2003}, where
\begin{equation}
\|T\|_{\mathrm{cb}}:=\sup_{n\geq1}\|\id_{M_n}\otimes T\|,\quad \norm{T_\ast}_{\diamond}:= \sup_{n\ge 1} \|\id_{M_n}\otimes T_\ast\|_{1\to 1}.
\end{equation}
Here $\norm{
  \cdot 
}$ denotes the operator norm and also denotes the induced operator norm for maps between finite-dimensional $C^*$-algebras, when the context is clear. $\norm{T}_{a\to b} := \sup_{X\ne 0 }\frac{\norm{T(X)}_b}{\norm{X}_a}$ denotes the induced operator norm for maps between $C^*$-algebras by the norms $\norm{\cdot}_a$ and $\norm{\cdot}_b$, and $\norm{X}_1 := \Tr[(X^\ast X)^{1/2}]$ denotes the trace norm. We say $E:M_d\to \mathcal A$ is a conditional expectation onto the unital $\ast$-subalgebra $\mathcal A$ if $E$ is a UCP map and satisfies the bimodular property $E(XAY) = XE(A)Y$ for all $X,Y\in M_d$ and $A\in\mathcal A$, and in particular $E(A)=A$ for all $A\in\mathcal A$ and $E$ is idempotent \cite{carlen2025inequalities}.
We sometimes denote $T_n:=\id_{M_n}\otimes T$ for a linear map $T$ between finite-dimensional $C^*$-algebras.
We write $f=\mathcal{O}(g)$ if $\norm{f}\leq Cg$ for some constant $C$ independent of the varying parameters in the stated regime, with $g\geq0$ and the norm understood from context. The notation $\mathcal{O}_\bullet(g)$ allows the implicit constant to depend on the parameters indicated by $\bullet$, but not on the remaining varying parameters. 

\subsection{Quantum metastability, almost idempotency and Riesz projections}\label{sec:riesz-slow}

We start with our motivating example of the quantum metastability, and briefly explain the spectral projection idea for the reduced description as considered in \cite{macieszczak2016towards,macieszczak2016towards} (see also \cite[Section 1.2]{Kitaev2025})
 before introducing the general framework in \cref{sec:main-results}. 
In the {Schr\"odinger} picture, a quantum Markov semigroup has the form $\mathcal C_t=e^{t\mathcal L_\ast}$, with the generator (or Lindbladian) $\mathcal L_\ast$. The Heisenberg-picture semigroup $T_t=\mathcal C_{t,\ast}=e^{t\mathcal L}$ is given by the dual. At a time $t_0$ in the metastable regime, fast components have decayed while the slow components have only changed little.
This motivates us to use the \emph{almost idempotency} of a quantum channel to describe the quantitative metastability
\begin{equation}\label{eq:state-defect}
\mathcal C=\mathcal C_{t_0},\quad
\Phi=\mathcal C_\ast=T_{t_0},\quad
\eta:=\|\mathcal C^2-\mathcal C\|_\diamond
=\|\Phi^2-\Phi\|_{\mathrm{cb}}\ll1.
\end{equation}
Intuitively, a metastable evolution channel $\mathcal C$ that satisfies $\eta$-almost idempotency changes very little when applied a second time. Figure~\ref{fig:metastable-spectrum} illustrates the corresponding separation between slow and fast eigenmodes.

\begin{figure}[h]
  \centering
  \includegraphics[width=0.7\textwidth]{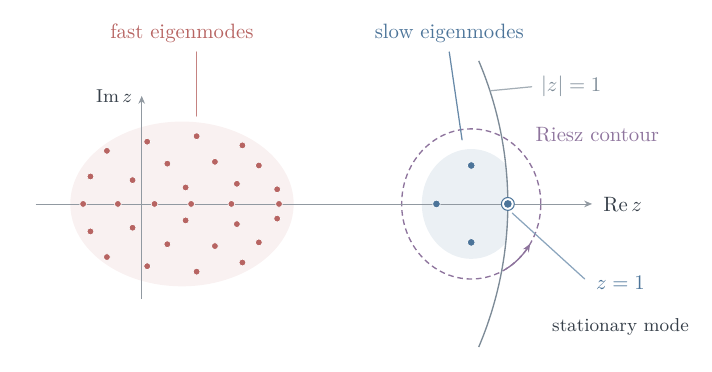}
  \caption{Schematic spectrum of a metastable evolution $\mathcal C $. }
    \label{fig:metastable-spectrum}
\end{figure}

To obtain a reduced description of the metastable behavior, we additionally assume that the slow subspace for the channel is of constant dimension independent of the full system size.  This is a natural condition both in classical and quantum metastable systems, as this ensures that only a small number of degrees of freedom govern the long-time dynamics after the fast relaxation.  Formally, let $\mathcal P_\Phi$ denote the projection onto the slow eigenspace of $\Phi$, and the assumption is $\rank \mathcal P_\Phi \le r$. $\mathcal P_\Phi$ is defined via the Riesz projection associated with 
  the cluster of $\mathrm{spec}(\Phi)$ near $1$, which corresponds to the slow eigenvalues of $\Phi$. Specifically, the Riesz projection in the Heisenberg picture is given by
\begin{equation}\label{eq:riesz-contour-main}
\mathcal P_\Phi=\frac{1}{2\pi\mathrm{i}}
\oint_{|z-1|=1/2}(z\,\id-\Phi)^{-1}\dd z,\quad (\text{assume $\eta<1/4$}).
\end{equation}
The condition $\eta <1/4$ separates the fast spectrum from this contour. The following proposition gives an error bound for approximating $\Phi$ using $\mathcal P_\Phi$.
\begin{prop}[Riesz projection of an almost-idempotent map]
\label{prop:riesz}
Let $\Phi:M_d\to M_d$ be UCP and suppose
 $
  \norm{\Phi^2-\Phi}_{\mathrm{cb}}<\frac14.
$ 
Then the Riesz projection $\mathcal P_\Phi$ of $\Phi$ onto the spectral component near $1$ is well-defined,
and
\[
  \norm{\mathcal{P}_\Phi-\Phi}_{\mathrm{cb}}
  =\mathcal{O}\bigl(\norm{\Phi^2-\Phi}_{\mathrm{cb}}\bigr).
\]
\end{prop}
The proof is deferred to Appendix~\ref{app:riesz-proof}.

\section{Main results}\label{sec:main-results}

To construct a QMSM, we seek an \emph{exactly idempotent quantum channel} $\mathcal Q$ approximating the metastable evolution $\mathcal C$. Equivalently (by taking adjoints), we seek a UCP map $E=\mathcal Q_\ast$ that satisfies $E^2=E$ and approximates $\Phi=\mathcal C_\ast$. We call such a map a \emph{UCP idempotent}. The slow spectral projection in \cref{prop:riesz} provides an idempotent close to $\Phi$, but need not be completely positive, as the examples in \cref{sec:spin-chains} demonstrate.

The range of a UCP idempotent can be identified with a finite-dimensional $C^*$-algebra,\footnote{In general, this uses the Choi--Effros product $x\circ y=E(xy)$ on $\Ran E$ \cite{ChoiEffros1977}. The range need not be closed under ordinary matrix multiplication. For a conditional expectation onto an ambient algebra, the two products coincide. We refer readers to \cref{sec:ucp-main} for details.} which decomposes into a direct sum of matrix blocks by the Wedderburn structure theorem \cite[Chapter~4]{carlen2025inequalities}. A reduced state assigns a classical label with a probability to each block and a quantum state within each block. Blocks larger than $1\times1$ can therefore serve as quantum memories. We first state the QMSM construction in the Schr\"odinger picture (\cref{thm:qmsm-master}) in \cref{sec:construction_qmsm}. In the later sections, we develop the construction mainly in the Heisenberg picture, where adjoint duality makes cb-norm estimates equivalent to diamond-norm estimates for channels on states.

\subsection{The construction of QMSMs} \label{sec:construction_qmsm}

A QMSM compresses the initial state (physical state) by $\mathcal Q$, evolves the reduced state (logical state) by the transition channel $\mathcal M$, and reconstructs a physical state by $\mathcal R$, as in \cref{eq:qmsm-compression-diagram}.
\begin{equation}\label{eq:qmsm-compression-diagram}
\begin{tikzcd}[column sep=5em,row sep=2.5em]
M_d \arrow[r,"\mathcal C"] & M_d \arrow[d,"\mathcal Q"] \\
\Ran\mathcal Q \arrow[u,hook,"\mathcal R"] \arrow[r,"\mathcal M"'] & \Ran\mathcal Q
\end{tikzcd}
\end{equation}

\begin{defn}[Quantum Markov state model]\label{def:qmsm}
Let $\mathcal C:M_d\to M_d$ be a channel and $\mathcal Q:M_d\to M_d$ a CPTP idempotent. We define the \emph{quantum Markov state model} associated with $(\mathcal C,\mathcal Q)$ to have compression $\mathcal Q:M_d\to\Ran\mathcal Q$, reconstruction given by the inclusion $\mathcal R:\Ran\mathcal Q\hookrightarrow M_d$, and reduced state space and transition channel
\begin{equation}\label{eq:qmsm-definition}
\mathscr D_\star=\mathcal Q(\mathscr D_d),\quad
\mathcal M=\mathcal Q\mathcal C\mathcal R.
\end{equation}
\end{defn}

The compression map $\mathcal Q:M_d\to\Ran\mathcal Q$ is the \emph{corestriction} of the ambient idempotent $M_d\to M_d$, obtained by changing only its codomain from $M_d$ to $\Ran\mathcal Q$. We use the same symbol for both maps. Reconstruction is inclusion because reduced states in $\mathscr D_\star \subset \Ran \mathcal Q$ can already be represented as physical density matrices, so $\mathcal R(\sigma)=\sigma$. Since $\mathcal Q^2=\mathcal Q$, we have $\mathcal Q\mathcal R=\id_{\Ran\mathcal Q}$. For an initial state $\rho\in\mathscr D_d$, we reconstruct a physical state after $k$ reduced steps by
\[
\rho^{(k)}=\mathcal R\mathcal M^k\mathcal Q(\rho),\quad k=0,1,\ldots.
\]

We show that, for any metastable channel $\mathcal C$ that is $\eta$-almost idempotent ($\eta\ll 1$), one can construct a QMSM that has all the desired properties for a reduced description of the metastable dynamics $\mathcal C$. 
We construct $\mathcal Q$ in two ways. The \emph{global} construction {$\mathcal Q_{\rm g}$} controls the diamond-norm error on arbitrary input states with order $\eta^{1/3}$. The \emph{metastable} construction {$\mathcal Q_{\rm m}$} controls the error {as measured by states prepared by $\mathcal C$} with the improved order $\sqrt\eta$. 

\begin{thm}[QMSM theorem, rigorous version of \cref{thm:QMSM_inf}]\label{thm:qmsm-master}
For every constant $r\in \mathbb N_{\ge 1}$, there are $\eta_r>0$ and $C_r<\infty$, independent of $d$, with the following property. Let $\mathcal C:M_d\to M_d$ be a channel with
\[
\eta=\|\mathcal C^2-\mathcal C\|_\diamond<\eta_r,\quad
r_0:=\rank\mathcal P_{\mathcal C}\leq r.
\]
Then the following QMSMs exist.
\begin{enumerate}
\item There are CPTP idempotents $\mathcal Q_{\rm g},\mathcal Q_{\rm m}$ such that 
\begin{equation}\label{eq:qmsm-channel-errors}
\begin{aligned}
\rank\mathcal Q_{\rm g}=r_0,\quad&
\|\mathcal Q_{\rm g}-\mathcal C\|_\diamond\leq C_r\eta^{1/3},\\
\rank\mathcal Q_{\rm m}\leq r_0+1,\quad&
\|\mathcal Q_{\rm m}\mathcal C-\mathcal C\|_\diamond\leq C_r\sqrt\eta.
\end{aligned}
\end{equation}
\item For each $\nu\in\{\mathrm g,\mathrm m\}$, there exists a finite-dimensional $C^*$-algebra $\mathcal A_\nu$ whose state space 
$\mathscr D(\mathcal A_\mu)$ 
is CPTP-isomorphic to $\mathscr D_\star^\nu=\mathcal Q_\nu(\mathscr D_d)$, i.e. the identifying channel on their linear spans has a CPTP inverse. In block coordinates,
\begin{equation}\label{eq:qmsm-block-states}
\mathscr D(\mathcal A_\nu)
=\left\{\bigoplus_jp_j\rho_j:\ p_j\geq0,\ \sum_jp_j=1,\ \rho_j\in\mathscr D_{n_{\nu,j}}\right\},
\quad
\mathcal A_\nu\cong\bigoplus_jM_{n_{\nu,j}},
\end{equation}
where $\sum_jn_{\nu,j}^2=\rank\mathcal Q_\nu$. The probabilities $p_j$ describe classical sectors, while blocks of dimension greater than one can restore quantum information.
\item The reduced channels $\mathcal M_\nu=\mathcal Q_\nu\mathcal C\mathcal R_\nu$ satisfy
\begin{equation}\label{eq:qmsm-iteration-errors}
\begin{aligned}
\|\mathcal R_{\mathrm g}\mathcal M_{\rm g}^k\mathcal Q_{\rm g}-\mathcal C^k\|_\diamond&
 \leq 2C_r k\eta^{1/3},~ k\geq1,  \\
\|\mathcal R_{\mathrm m}\mathcal M_{\rm m}^k\mathcal Q_{\rm m}\mathcal C-\mathcal C^{k+1}\|_\diamond&
 \leq C_r(k+1)\sqrt\eta,~ k\geq0.
\end{aligned}
\end{equation}
\end{enumerate}
\end{thm}

We next state the main UCP theorems in \cref{sec:ucp-main}, and \cref{thm:qmsm-master} follows from the UCP results and the reduced-transition corollary in \cref{sec:ucp-main}. The proof of \cref{thm:qmsm-master} will be given at the end of \cref{sec:reduced-dynamics}.

\subsection{UCP theorems and the metastable seminorm}\label{sec:ucp-main}
We now state the operator-algebraic results proved in this paper. Throughout, $\Phi=\mathcal C_\ast$ and $E_\nu=\mathcal Q_{\nu,\ast}$, for $\nu\in\{\mathrm g,\mathrm m\}$, relate the Heisenberg and Schrödinger pictures.  

\begin{thm}[Main theorem, cb-norm version]
\label{thm:main}
For every $r\geq1$, there are constants $\eta_r>0$ and $C_r<\infty$ such that the following holds.  If $\Phi:M_d\to M_d$ is UCP and
\[
  \norm{\Phi^2-\Phi}_{\mathrm{cb}}\leq\eta<\eta_r,
  \quad
  \rank \mathcal{P}_\Phi\leq r,
\]
then there is a UCP idempotent $E_{\rm g}:M_d\to M_d$ satisfying
\[
  \norm{E_{\rm g}-\Phi}_{\mathrm{cb}}\lesssim_r\eta^{1/3}.
\]
\end{thm}
The proof of \cref{thm:main} is given in \cref{sec:cb-proof}.
\begin{rem}[Sharpness and obstruction]\label{rem:sharpness-slow-rank}
The exponent $1/3$ is sharp. Without bounded slow rank, no such dimension-independent approximation modulus tending to zero exists. See Appendix~\ref{app:sharpness} for details. 
\end{rem}
This gives a quantitative answer to Kitaev's exact-rounding question \cite[Section~1.2]{Kitaev2025} under the fixed slow-rank condition: an almost-idempotent UCP map $\Phi$ is close to an exactly idempotent UCP map $E_{\rm g}$, with the bound $C_r\eta^{1/3}$ vanishing with the idempotency defect $\eta$ and the constant independent of the ambient dimension $d$. The resulting projection is UCP. In particular, it retains complete positivity, which the spectral-rounding construction in \cref{sec:riesz-slow} fails to guarantee. \Cref{thm:main} does not require $\Ran E_{\rm g}$ to be closed under ordinary matrix multiplication on $M_d$. The Choi--Effros product $x\circ y=E_{\rm g}(xy)$ supplies its reduced $C^*$-algebra structure \cite{ChoiEffros1977}.

There are two aspects in \cref{thm:main} that can be improved by replacing the global cb-norm estimate with the measurement that only concerns the metastable states. (1) The error bound can be tightened from cube-root dependence on $\eta$ to square-root dependence on $\eta$. (2)  {The range of the idempotent UCP can be ensured to be a subalgebra with respect to the \emph{ambient} multiplication inherited from $M_d$, removing the need for the Choi--Effros product structure.} Motivated by the measurement on the metastable states \cref{eq:qmsm-channel-errors}, we introduce the metastable seminorm \cref{def:metastable-seminorm} by the duality argument.

\begin{defn}[Metastable seminorm]\label{def:metastable-seminorm}
Let $\Phi:M_d\to M_d$ be UCP and let $\mathscr D_{\star,\Phi}=\Phi_\ast(\mathscr D_d)$ be the metastable states prepared by $\Phi_\ast$. For $Y\in M_d$, set the seminorm weighted by all the metastable states as
\begin{equation}\label{eq:metastable-observable-norm}
\|Y\|_{2,\mathscr D_{\star,\Phi}}
:=\sup_{\sigma\in\mathscr D_{\star,\Phi}}\sqrt{\Tr(\sigma Y^*Y)}
=\|\Phi(Y^*Y)\|^{1/2}.
\end{equation}
Passing to the Heisenberg picture, 
for a linear map $T:M_d\to M_d$, write $\Phi_m=\id_{M_m}\otimes\Phi$ and $T_m=\id_{M_m}\otimes T$. Its complete weighted seminorm is defined as
 \begin{equation}\label{eq:weighted-seminorm}
\|T\|_{\Phi,2}^{\mathrm{cb}}
:=\sup_{m\geq1}\ \sup_{\|X\|\leq1}
\|T_m(X)\|_{2,\mathscr D_{\star,\Phi_m}}
=\sup_{m\geq1}\ \sup_{\|X\|\leq1}
\|\Phi_m(T_m(X)^*T_m(X))\|^{1/2}.
\end{equation}
\end{defn}

 {We call this the} \emph{metastable seminorm} throughout. We have the following hierarchy of norms:
\begin{prop}\label{prop:weighted-hierarchy} Let $T:M_d\to M_d$ be a linear map. Then
\[
\norm{\Phi T}_{\mathrm{cb}}
\leq\norm{T}_{\Phi,2}^{\mathrm{cb}}
\leq\norm T_{\mathrm{cb}}.
\]
\end{prop}
\begin{proof}
  Put $Y=T_m(X)$.  Since $\Phi_m$ is UCP, Kadison--Schwarz inequality \cite{Paulsen2003,carlen2025inequalities}
  and positivity give
\[
\Phi_m(Y)^*\Phi_m(Y)\leq\Phi_m(Y^*Y)
\leq\norm{Y}^2\Phi_m(I)=\norm{Y}^2I.
\]
Using $\norm{A^*A}=\norm A^2$, taking square roots, and then taking suprema over $m\geq1$ and $\norm X\leq1$ yields the desired proposition.
\end{proof}
With \cref{prop:weighted-hierarchy}, together with the duality argument, we see that the metastable seminorm controls the diamond-norm distance measured on metastable states, which will be discussed in detail in \cref{sec:reduced-dynamics}.

We next give the main theorem in the metastable-seminorm version.
\begin{thm}[Main theorem, metastable-seminorm version]\label{thm:weighted-ambient}
For every $r\geq1$ there are constants $\eta_r>0$ and $C_r<\infty$, independent of $d$, such that, if
\[
\eta=\norm{\Phi^2-\Phi}_{\mathrm{cb}}<\eta_r,
\quad \rank \mathcal{P}_\Phi\leq r,
\]
there are a projection $q\in M_d$, an ambient algebra $\mathcal N_0\subset qM_dq$ whose unit is $q$ and dimension $\rank \mathcal{P}_\Phi$, and a conditional expectation $E_{\rm m}:M_d\to\mathcal N$ close to $\Phi$, satisfying 
\begin{equation}\label{eq:weighted-main}
\mathcal N=\mathcal N_0\oplus\CC(I-q),
\quad
\norm{E_{\rm m}-\Phi}_{\Phi,2}^{\mathrm{cb}}\lesssim_r\sqrt\eta,
\quad
\norm{\Phi(I-q)}\lesssim_r\eta.
\end{equation}
The scalar summand is omitted if $q=I$. Thus $\mathcal N$ is an ambient algebra with $\dim \mathcal N\leq\rank \mathcal{P}_\Phi+1 \le  r+1$. The inequality $\norm{\Phi(I-q)}\lesssim_r\eta$ implies that every metastable state has only $\cO_r(\eta)$ weight on the scalar summand $\mathbb C(I-q)$ of $\mathcal N$.
\end{thm}
Here the projection $q$ selects a physical subspace $q\CC^d$ in which the reduced model represents the classical and quantum information. Specifically, this subspace supports a faithful representation of the logical observable algebra as a physical subalgebra $\mathcal N_0\subset qM_dq$, as described in \cref{sec:surgery-picture}. For a metastable state $\sigma=\Phi_\ast(\rho)$ prepared from any initial density matrix $\rho$, we bound the probability of the $I-q$ outcome in the projective measurement $\{q,I-q\}$ by
\[
\Tr\!\left(\sigma(I-q)\right)
=\Tr\!\left(\rho\Phi(I-q)\right)
\leq\norm{\Phi(I-q)}
\lesssim_r\eta.
\]
Thus this measurement finds the prepared system in $q\CC^d$ with probability at least $1-\mathcal O_r(\eta)$.

When $q\ne I$, the sector $\CC(I-q)$ records the total probability in the complementary subspace, and ensures that the reduced range is a unital ambient algebra. By contrast, \cref{thm:main} does not require this additional sector, since the idempotent range of $E_g$ need not be closed under ordinary multiplication in general.

The examples in Appendix~\ref{app:sharpness} also show that the metastable square-root exponent is optimal for ambient expectations, and that global cb approximation by an ambient conditional expectation is generally not possible.
The proof is given in \cref{sec:weighted-proof}.  

\subsection{Reduced dynamics}\label{sec:reduced-dynamics}

We construct the reduced dynamics first in the Heisenberg picture with the UCP idempotent $E$. Let $\mathcal A_E=\Ran E$ be the reduced algebra (equipped with the Choi--Effros product if necessary).We let $\mathcal Q = E_\ast:M_d\to M_d$, and as before, we also denote the corestriction $\mathcal Q:M_d\to\Ran\mathcal Q$. The structure classification for $\mathcal A_E$ then gives the classification of the classical and quantum degrees of freedom in $\Ran \mathcal Q$ in \cref{thm:qmsm-master} (2). Let $\mathcal R:\Ran\mathcal Q  \hookrightarrow M_d$ be inclusion. Then $\mathcal M = \mathcal Q\mathcal C\mathcal R$ gives the reduced channel acting on the state space $\mathscr D_\star \subset \mathcal Q$. Using the error bounds for UCP idempotents, we have the corresponding error bound for the reduced transition:
\begin{cor}[Reduced transition]\label{cor:reduced-transition}
 With the notations in \cref{thm:qmsm-master}, we have 
\begin{itemize}
  \item With $\delta=\norm{\mathcal Q-\mathcal C}_\diamond=\norm{E-\Phi}_{\mathrm{cb}}$, for every $k\geq1$,
\begin{equation}\label{eq:transition-global}
\norm{\mathcal R\mathcal M^k\mathcal Q-\mathcal C^k}_\diamond
\leq2k\delta.
\end{equation}
\item With $\eta=\norm{\mathcal C^2-\mathcal C}_\diamond$ and $\epsilon=\norm{E-\Phi}_{\Phi,2}^{\mathrm{cb}}$ as defined in \cref{eq:weighted-seminorm}, for every $k\geq0$,
\begin{equation}\label{eq:transition-prepared}
\norm{\mathcal R\mathcal M^k\mathcal Q\mathcal C-\mathcal C^{k+1}}_\diamond
\leq(k+1)(\epsilon+\eta).
\end{equation}
\end{itemize}
\Cref{eq:transition-global} applies to arbitrary initial states, while \cref{eq:transition-prepared} includes one propagation step $\mathcal C$ and gives the estimate for metastable states.
\end{cor}
\begin{proof}
The reduced-space identification gives $\mathcal Q\mathcal R=\id_{\Ran\mathcal Q}$. For $k\geq1$, with $\mathcal Q$ interpreted as a map $M_d\to M_d$ in the products on the right,
\[
\mathcal R\mathcal M^k\mathcal Q=(\mathcal Q\mathcal C\mathcal Q)^k,\quad
\mathcal Q\mathcal C\mathcal Q-\mathcal C
=\mathcal Q(\mathcal C-\mathcal Q)\mathcal Q+(\mathcal Q-\mathcal C).
\]
The second expression has diamond norm at most $2\delta$. Telescoping powers of channels yields $\norm{(\mathcal Q\mathcal C\mathcal Q)^k-\mathcal C^k}_\diamond\leq2k\delta$, proving \cref{eq:transition-global}. By cb--diamond duality and \cref{prop:weighted-hierarchy},
\[
\norm{(\mathcal Q-\mathcal C)\mathcal C}_\diamond
=\norm{\Phi(E-\Phi)}_{\mathrm{cb}}
\leq\norm{E-\Phi}_{\Phi,2}^{\mathrm{cb}}=\epsilon.
\]
Also $\mathcal Q\mathcal C-\mathcal C=(\mathcal Q-\mathcal C)\mathcal C+(\mathcal C^2-\mathcal C)$, so
$
\norm{\mathcal Q\mathcal C-\mathcal C}_\diamond\leq\epsilon+\eta
$ by the triangle inequality. 
Finally, $\mathcal R\mathcal M^k\mathcal Q\mathcal C=(\mathcal Q\mathcal C)^{k+1}$ when writing only with the ambient maps for $k\geq0$, so another telescoping argument proves \cref{eq:transition-prepared}.
\end{proof}
\begin{proof}[Proof of \cref{thm:qmsm-master}]
Apply \cref{thm:main,thm:weighted-ambient} to $\Phi=\mathcal C_\ast$. Their ambient adjoints $\mathcal Q_\nu=E_{\nu,\ast}$ are CPTP idempotents. In the global construction, $E_{\rm g}=\Delta_0\Upsilon_0$ and $\Upsilon_0\Delta_0=\id_{\mathcal A}$ give $\Ran E_{\rm g}=\Delta_0(\mathcal A)$, so $\rank E_{\rm g}=\dim\mathcal A=r$. The metastable construction has rank at most $r+1$. Cb--diamond duality gives the first bound in \eqref{eq:qmsm-channel-errors}; \cref{prop:weighted-hierarchy} gives the second:
\[
\|\mathcal Q_{\rm m}\mathcal C-\mathcal C\|_\diamond
\leq\|(E_{\rm m}-\Phi)_\ast\Phi_\ast\|_\diamond+\eta
\leq\|E_{\rm m}-\Phi\|_{\Phi,2}^{\mathrm{cb}}+\eta.
\]
For each $\nu\in\{\mathrm g,\mathrm m\}$, write $E=E_\nu$ and $\mathcal Q=\mathcal Q_\nu=E_\ast$ as maps $M_d\to M_d$. The range $\mathcal A_E=\Ran E$, with the Choi--Effros product when necessary \cite{ChoiEffros1977}, is a finite-dimensional $C^*$-algebra. The structure theorem \cite[Theorem~4.28]{carlen2025inequalities} gives a unital $*$-isomorphism $\theta$ and thus $\theta$ is UCP: 
\[
\theta_\nu:\mathcal A_E\longrightarrow\mathcal A_\nu:=\bigoplus_{j=1}^{s_\nu} M_{n_{\nu,j}},\quad \text{where~}\sum_{j=1}^{s_\nu} n_{\nu,j}^2=\dim\mathcal A_E=\rank\mathcal Q.
\]
We now identify the structures of the states in $\mathscr D_{\star}^\nu = \mathcal Q(\mathscr D_d)$, as desired in \cref{thm:qmsm-master}~(2), 
 using the canonical block trace $\Tr_{\mathcal A_\nu}=\sum_j\Tr_{n_{\nu,j}}$ on the abstract $\mathcal A_\nu$. 
A logical block density matrix $\sigma \in \mathcal A_\nu$ specifies a physical state $\mathfrak J_\nu(\sigma)\in M_d$ by evaluating a physical observable $X \in M_d$ through its logical image $E(X)\in \mathcal A_E$:
\[
\Tr_d(\mathfrak J_\nu(\sigma)X)=\Tr_{\mathcal A_\nu}(\sigma\theta_\nu(E(X))).
\]
Conversely, define a channel $\mathfrak D_\nu:M_d\to\mathcal A_\nu$ on all physical matrices by evaluating each logical observable $a\in\mathcal A_\nu$ through its physical representative $\theta_\nu^{-1}(a)\in\mathcal A_E\overset{\iota_\nu}\hookrightarrow M_d$, where $\iota_\nu:\mathcal A_E\hookrightarrow M_d$ is inclusion. For a physical state $\rho$, its block density matrix is specified by
\[
\Tr_{\mathcal A_\nu}(\mathfrak D_\nu(\rho)a)=\Tr_d(\rho(\iota_\nu \theta_\nu^{-1})(a)).
\]
These maps are CPTP because they are the adjoints of the UCP maps $\theta_\nu E$ and $\iota_\nu\theta_\nu^{-1}$, respectively. For $a\in\mathcal A_\nu$ and $X\in M_d$, the two defining formulas give
\[
\begin{aligned}
\Tr_{\mathcal A_\nu}(\mathfrak D_\nu\mathfrak J_\nu(\sigma)a)&=\Tr_d(\mathfrak J_\nu(\sigma)\theta_\nu^{-1}(a)) 
 =\Tr_{\mathcal A_\nu}(\sigma\theta_\nu(E(\theta_\nu^{-1}(a))))=\Tr_{\mathcal A_\nu}(\sigma a),
\end{aligned}
\]
using $E|_{\mathcal A_E}=\id$, and
\[
\begin{aligned}
\Tr_d(\mathfrak J_\nu \mathfrak D_\nu(\rho)X)&=\Tr_{\mathcal A_\nu}(\mathfrak D_\nu(\rho)\theta_\nu(E(X))) 
 =\Tr_d(\rho E(X))=\Tr_d(\mathcal Q(\rho)X),
\end{aligned}
\]
using $\mathcal Q=E_\ast$ (as maps from $M_d$ to $M_d$). Arbitrariness of $a$ and $X$ therefore gives $\mathfrak D_\nu\mathfrak J_\nu(\sigma)=\sigma$ and $\mathfrak J_\nu \mathfrak D_\nu(\rho)=\mathcal Q(\rho)$.
Thus $\mathfrak J_\nu$ is a CPTP isomorphism from the block-matrix space $\mathcal A_\nu$ onto $\Ran\mathcal Q$, with inverse given by $\mathfrak D_\nu|_{\Ran\mathcal Q}$, since $Q$ is idempotent and thus acts on $\Ran \mathcal Q$ as the identity. 
Restricting to positive matrices of trace one gives $\mathfrak J_\nu(\mathscr D(\mathcal A_\nu))=\mathcal Q(\mathscr D_d)=\mathscr D_\star^\nu$. 
This proves \eqref{eq:qmsm-block-states}. \Cref{cor:reduced-transition} gives \eqref{eq:qmsm-iteration-errors}, after enlarging $C_r$ and taking $\eta_r\leq1$.
\end{proof}

\begin{rem}\label{rem:faithful-expectation}
An ambient range algebra also follows when the UCP idempotent preserves a faithful state: suppose $E_\ast(\rho)=\rho$ with $\mathscr D\ni\rho>0$. For $x\in\Ran E$, Schwarz inequality gives $E(x^*x)-x^*x\geq0$, and
$
\Tr\!\left(\rho[E(x^*x)-x^*x]\right)=0.
$ 
Faithfulness forces equality, also for $xx^*$, so every $x\in\Ran E$ lies in the multiplicative domain of $E$ \cite[Theorem~3.18]{Paulsen2003}. Thus $\Ran E$ is closed under ordinary multiplication and $E(axb)=aE(x)b$ for $a,b\in\Ran E$. Hence $E$ is a $\rho$-preserving conditional expectation. 
For a prescribed ambient subalgebra $\mathcal N$, Takesaki's theorem characterizes the existence of a $\rho$-preserving conditional expectation by modular invariance:
$\rho^{\mathrm{i}t}\mathcal N\rho^{-\mathrm{i}t}=\mathcal N$ for every $t\in \mathbb R$ \cite{Takesaki1972}. 

The metastable-seminorm construction in \cref{thm:weighted-ambient} bypasses these assumptions and directly constructs $\mathcal N_0\subset qM_dq$ and the ambient expectation onto $\mathcal N_0\oplus\CC(I-q)$.

\end{rem}

\section{Approximate encoding--decoding factorization and Stinespring representations}
\label{sec:kitaev}

\subsection{Kitaev's approximate encoding--decoding factorization}\label{sec:kitaev-factorization}

The proof relies on Kitaev's approximate physical factorization in \cref{thm:kitaev-factorization} through a small logical algebra $\mathcal A$ for an almost-idempotent UCP map $\Phi$ \cite{Kitaev2025}. $\Delta$ is the decoding map $ \mathcal A\to M_d$ and $\Upsilon$ is the encoding map $M_d \to \mathcal A$ for observables in the Heisenberg picture. The approximate factorization means that 
\begin{equation}
  \Phi \approx \Delta \Upsilon, \quad \Upsilon\Delta \approx \id_{\mathcal A}.
\end{equation}
We aim to repair the recovery relation while keeping both $\Upsilon_0$ and $\Delta_0$ UCP: 
\begin{equation}\label{eq:retain_exact}
\Upsilon \Delta \approx \id_{\mathcal A} \to  \Upsilon_0 \Delta_0 = \id_{\mathcal A}. 
\end{equation}
Once this surgery is complete, $E:=\Delta_0\Upsilon_0$ is an exactly idempotent UCP map. Moreover, $\mathcal A_E = \Ran E$ equipped with the Choi--Effros product is canonically $\ast$-isomorphic to the logical algebra $\mathcal A$, so Kitaev's logical algebra is identified with the one discussed in \cref{sec:reduced-dynamics}. More precisely, 
\begin{prop}[Logical algebra of the repaired idempotent]\label{prop:repair} Let $\Upsilon_0:M_d\to \mathcal A,\Delta_0:\mathcal A\to M_d$ be UCP maps satisfying 
  $\Upsilon_0\Delta_0 = \id_{\mathcal A}$, and $E =\Delta_0\Upsilon_0:M_d\to M_d$. Then $\mathcal A_E = \Ran E$ equipped with the Choi--Effros product is $\ast$-isomorphic to the logical algebra $\mathcal A$ via $\Delta_0:\mathcal A\to \mathcal A_E$.
\end{prop}
\begin{proof}
We identify the logical algebra $\mathcal A$ with the range algebra $\mathcal A_{E }=\Ran  E$ from \cref{sec:ucp-main} after surgery. We show that $\Delta_0:\mathcal A\to\mathcal A_{E }$ is a unital $*$-isomorphism for the Choi--Effros product, with inverse $\Upsilon_0|_{\mathcal A_{E }}$. Indeed, $E =\Delta_0\Upsilon_0$ and $\Upsilon_0\Delta_0=\id_{\mathcal A}$ give $E \Delta_0=\Delta_0$, $\Ran E =\Delta_0(\mathcal A)$, and $\Delta_0 : \mathcal A\to \Ran E$. Applying Kadison--Schwarz to both maps gives
\[
a^*a = [\Upsilon_0 (\Delta_0 (a))]^\ast \Upsilon_0(\Delta_0(a) )
\leq\Upsilon_0(\Delta_0(a)^*\Delta_0(a))\note{$\Upsilon_0$ is positive}\leq\Upsilon_0\Delta_0(a^*a)=a^*a.
\]
The same equality for $aa^*$ places $\Delta_0(a)$ in the multiplicative domain of $\Upsilon_0$ \cite[Theorem~3.18]{Paulsen2003}. Hence, for $a,b\in\mathcal A$, by the definition of the Choi--Effros product,
\[
\Delta_0(a)\circ\Delta_0(b):=E (\Delta_0(a)\Delta_0(b)) = \Delta_0 [(\Upsilon_0\Delta_0)(a) (\Upsilon_0\Delta_0)(b)]=\Delta_0(ab).
\]
Thus $\Delta_0$ is a unital $\ast$-homomorphism. Together with bijectivity, this proves \cref{prop:repair}. 
\end{proof}

\begin{rem}[Ambient identification in the metastable construction]\label{rem:metastable-logical-algebra}
In the metastable construction of \cref{thm:weighted-ambient}, the faithful representation $\mathfrak w:\mathcal A\to\mathcal N_0\subset qM_dq$ identifies the logical algebra directly with a subalgebra equipped with ordinary matrix multiplication on $M_d$. The full range of $E_{\rm m}$ is the physical subalgebra $\mathcal N_0\oplus \CC(I-q)$, so no Choi--Effros product is needed.
\end{rem}

 We keep the encoding--decoding terminology for these factorization maps and their Stinespring dilations. We first state the approximate factorization result as the following consequence of \cite[Theorem~12.3]{Kitaev2025}.
\begin{thm}[Approximate factorization]
\label{thm:kitaev-factorization}
There is a universal constant $c>0$ such that, whenever $\Phi:M_d\to M_d$ is UCP and
\[
  \norm{\Phi^2-\Phi}_{\mathrm{cb}}\leq\eta<\eta_0,
\]
there are a finite-dimensional $C^*$-algebra $\mathcal A$ with $\dim\mathcal A=\rank \mathcal{P}_\Phi$ and UCP maps
\[
  \Delta:\mathcal A\to M_d,
  \quad
  \Upsilon:M_d\to\mathcal A
\]
such that
\begin{equation}
\label{eq:kitaev-factorization}
  \norm{\Delta\Upsilon-\Phi}_{\mathrm{cb}}\leq c\eta,
  \quad
  \norm{\Upsilon\Delta-\id_{\mathcal A}}_{\mathrm{cb}}\leq c\eta.
\end{equation}
\end{thm}

\begin{proof}
In the notation of \cite{Kitaev2025},
\[
  \widetilde\Phi=\mathcal{P}_\Phi,
  \quad
  \mathcal A_{\mathrm{app}}
  =\Ran\widetilde\Phi
  =\Ran \mathcal{P}_\Phi.
\]
Theorem 10.5 of \cite{Kitaev2025} gives a genuine $C^*$-algebra $\mathcal B$ and an extended $\mathcal{O}(\eta)$-isomorphism $v:\mathcal B\to\mathcal A_{\mathrm{app}}$.  Here ``$\mathcal{O}(\eta)$-isomorphism'' includes bijectivity \cite[Definition~2.2 and Section~10]{Kitaev2025}
, so $v$ is a linear isomorphism and
\[
  \dim\mathcal B=\dim\mathcal A_{\mathrm{app}}=\rank \mathcal{P}_\Phi.
\]
More explicitly, let $\iota:\mathcal A_{\mathrm{app}}\hookrightarrow M_d$ be the inclusion and regard $\mathcal{P}_\Phi$ as a map into $\mathcal A_{\mathrm{app}}$.  \cite{Kitaev2025} first defines
\[
  \widetilde\Delta=\iota v:\mathcal B\to M_d,
  \quad
  \widetilde\Upsilon=v^{-1}\mathcal{P}_\Phi:M_d\to\mathcal B.
\]
Since $\mathcal{P}_\Phi|_{\mathcal A_{\mathrm{app}}}=\id_{\mathcal A_{\mathrm{app}}}$,
\[
  \widetilde\Delta\widetilde\Upsilon=\mathcal{P}_\Phi,
  \quad
  \widetilde\Upsilon\widetilde\Delta=\id_{\mathcal B}.
\]
These maps are constructed from $v$ and $v^{-1}$ and need not be UCP. Eqs.~(12.31)--(12.32) and (12.40)--(12.42) replace them with UCP maps having the same domains and codomains,
\[
  \widetilde\Delta:\mathcal B\to M_d
  \longmapsto\Delta:\mathcal B\to M_d,
  \quad
  \widetilde\Upsilon:M_d\to\mathcal B
  \longmapsto\Upsilon:M_d\to\mathcal B.
\]
The linear isomorphism $v$ and the algebra $\mathcal B$ are not changed.  Thus the original linear bijection $v:\mathcal B\to\mathcal A_{\mathrm{app}}$ and the dimension equality above remain valid.  The theorem statement simply renames $\mathcal B$ as $\mathcal A$.
 
  After enlarging $c$ if necessary, estimate (12.24) of \cite{Kitaev2025} gives
\[
  \norm{\Upsilon_n\!\left(\Delta_n(X)\Delta_n(Y)\right)-XY}
  \leq c\eta \norm{X} \norm{Y}
\]
for every $n$ and $X,Y\in M_n\otimes\mathcal A$.  Set $Y=I_n\otimes I_{\mathcal A}$.  Since $\Delta$ is unital, $\Delta_n(Y)=I_n\otimes I_d$, and hence
\[
  \norm{(\Upsilon\Delta)_n(X)-X}
  \leq c\eta   \norm{X}.
\]
Taking the supremum over $n$ and nonzero $X$ gives the second estimate in \cref{eq:kitaev-factorization}.
\end{proof}

\subsection{Stinespring representations and matrix blocks}\label{sec:stinespring-picture}
Stinespring representations express a completely positive map through an isometry into a larger Hilbert space and a $\ast$-homomorphism on the larger space \cite[Theorem~4.1]{Paulsen2003}.
 See also \cite[Section~11.1]{Kitaev2025}.
In particular, for UCP maps between matrix algebras, the unital $\ast$-homomorphism can be taken as an amplification $x\mapsto x\otimes I_{\mathcal E}$ \cite[Exercise~4.11]{Paulsen2003}. Specifically, 
if $\Phi:M_{d_1}\to M_{d_2}$ is UCP, then there exists a Hilbert space $\mathcal E$ and an isometry $V$ such that 
\begin{equation}\label{eq:matrix-stinespring}
\Phi(x)=V^*(x\otimes I_{\mathcal E})V,\quad
V^*V=I_{d_2},\quad V:\CC^{d_2}\to\CC^{d_1}\otimes\mathcal E.
\end{equation}
  Writing $V\xi=\sum_{k=1}^{\ell}V_k\xi\otimes e_k$ in an orthonormal environmental basis gives the Kraus representation
\[
\Phi(x)=\sum_{k=1}^{\ell}V_k^*xV_k,\quad
V_k:\CC^{d_2}\to\CC^{d_1},\quad \sum_kV_k^*V_k=I_{d_2}.
\]
The adjoint channel acts as the 
$\Phi_\ast(\rho)=\Tr_{\mathcal E}(V\rho V^*)$ where $\Tr_{\mathcal E}: M_{d_1}\otimes M_{\mathcal E}\to M_{d_1}$ is the partial trace over the environment, which can be easily seen from the following computation:
\[
\langle \Phi(x), y\rangle_{\rm HS} = \Tr(V^\ast(x\otimes I_{\mathcal E})V y ) = \Tr((x\otimes I_{\mathcal E}) (VyV^\ast)) =\Tr(x \Tr_{\mathcal E} (VyV^\ast)) = \langle x, \Tr_{\mathcal E} (VyV^\ast)\rangle_{\rm HS}.
\]

The structure theorem for finite-dimensional $C^*$-algebras gives a $*$-isomorphism $\mathcal A\cong\bigoplus_{j=1}^sM_{n_j}$ \cite[Theorem~4.28]{carlen2025inequalities}. We use this identification to write $a=(a_j)$ and introduce the coordinate map and its trace adjoint:
\begin{equation}\label{eq:block-coordinates}
\begin{aligned}
\operatorname{pr}_j:\bigoplus_iM_{n_i}\to M_{n_j},\quad&
\operatorname{pr}_j(a)=a_j,\\
\operatorname{pr}_{j,\ast}:M_{n_j}\to\bigoplus_iM_{n_i},\quad&
\operatorname{pr}_{j,\ast}(b)=(0,\ldots,0,b,0,\ldots,0).
\end{aligned}
\end{equation}
Here the traces are defined in \cref{eq:trace-sum}. In particular,
$\langle a,b\rangle_{\mathcal A}=\sum_j\Tr(a_j^*b_j)$.
For a UCP map from $\mathcal A$ to $M_d$, a Stinespring representation has the block form
\begin{equation}\label{eq:stinespring-representation}
\Phi(a)=W^*\pi(a)W,\quad
\pi(a)=\bigoplus_j(a_j\otimes I_{\mathcal E_j}),\quad
W:\CC^d\to\mathcal G:=\bigoplus_j(\CC^{n_j}\otimes\mathcal E_j).
\end{equation}
This representation keeps every logical matrix factor intact. The multiplicity spaces $\mathcal E_j$ carry the environmental freedom. We can easily see \cref{eq:stinespring-representation} by identifying $A\cong\bigoplus_j M_{n_j}$ and applying \cref{eq:matrix-stinespring} block-wise.

\subsection{The physical picture of the surgery}\label{sec:surgery-picture}

Before presenting the rigorous proof in \cref{sec:exactification}, we outline its strategy. We start from the approximate encoding--decoding factorization in \cref{sec:kitaev-factorization} and repair the Stinespring representations of the UCP maps, introduced in \cref{sec:stinespring-picture}, to restore the \emph{exact recovery} relation. 

 Since the encoding map $\Upsilon: M_d \to \mathcal A \cong \bigoplus_j M_{n_j}$ sends the physical input to each logical block $M_{n_j}$, we could write $\Upsilon_j=\operatorname{pr}_j  \Upsilon$ in block coordinates. The Stinespring representations of $\Delta$ and $\Upsilon_j$ are given by
\begin{equation}\label{eq:encoding-decoding-dilations}
\Delta(a)=W^*\pi(a)W,\quad
\Upsilon_j(X)=V_j^*(X\otimes I_{\mathcal Y_j})V_j,\quad
Z_j=(W\otimes I_{\mathcal Y_j})V_j.
\end{equation}
Here $V_j:\CC^{n_j}\to\CC^d\otimes\mathcal Y_j$ is an isometry, and $Z_j$ dilates the composite $\Upsilon_j\Delta$. We call $\mathcal E_j$ the \emph{decoding environment} (the multiplicity space of $\pi$) and $\mathcal Y_j$ the \emph{encoding environment}. Correspondingly, we call the Stinespring dilations $W$ and $V_j$ the \emph{decoding dilation} and \emph{encoding dilation}, respectively.

If $\delta=\|\Upsilon\Delta-\id_{\mathcal A}\|_{\mathrm{cb}}$ is small, then each $\Upsilon_j\Delta$ is close to $\operatorname{pr}_j$. By \cref{lem:product-dilation}, $Z_j$ is therefore $\mathcal{O}_{\mathcal A}(\sqrt\delta)$-close to a product isometry $J_j\xi=\xi\otimes\zeta_j$ in the $j$th representation block. The joint environmental unit vector $\zeta_j\in\mathcal E_j\otimes\mathcal Y_j$ is independent of the logical input $\xi$. Intuitively, this means that approximate recovery nearly decouples the logical information from both environments. 

Figure~\ref{fig:surgery-cartoon} summarizes the coupled repair of the decoding and encoding dilations.
\begin{figure}[h]
  \centering
  \includegraphics[width=\textwidth]{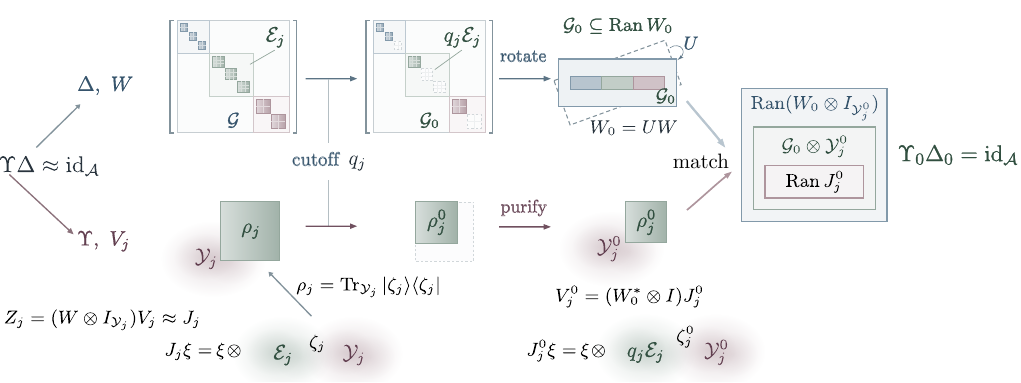}
  \caption{\raggedright A schematic illustration of the surgery on the decoding and encoding dilations. A cutoff by $\{q_j\}$ selects the repaired decoding space $\mathcal G_0$ and nearby environmental states $\rho_j^0$. Rotation ensures $\mathcal G_0\subseteq\operatorname{Ran}W_0$, and purification yields the repaired product isometries $J_j^0\xi=\xi\otimes\zeta_j^0$. The nested range inclusions then allow the pullback that defines $V_j^0$, matching the two repaired dilations and yielding $\Upsilon_0\Delta_0=\mathrm{id}_{\mathcal A}$.}
  \label{fig:surgery-cartoon}
\end{figure}

Although the product dilation $J_j$ keeps the exact logical input $\xi\in \CC^{n_j}$, simply replacing $Z_j$ by $J_j$ does not yield the exact recovery relation in \cref{eq:retain_exact}. Indeed, 
it need not factor through the original decoding isometry $W$ since its range need not lie in $\Ran W\otimes\mathcal Y_j$. To address this mismatch, we must repair the encoding and decoding maps together. With $P=WW^*$, the leakage measuring the failure of \(\operatorname{Ran}J_j\) to lie in \(\operatorname{Ran}W\otimes\mathcal Y_j\) is controlled by
\[
\|((I-P)\otimes I)J_j\|
\leq\|J_j-Z_j\|\lesssim_{\mathcal A}\sqrt\delta.
\]
For the global cb-norm estimate, we first apply \cref{lem:state-cutoff} to the decoding-environment states $\rho_j:=\Tr_{\mathcal Y_j}\dyad{\zeta_j}$ by tracing out $\mathcal Y_j$ for $\zeta_j\in \mathcal E_j\otimes \mathcal Y_j$. At a cutoff threshold $\tau$, this gives projections $q_j$ on $\mathcal E_j$ and nearby states $\rho_j^0$ supported on $q_j\mathcal E_j$. This identifies a subspace $\mathcal{G}_0:=\bigoplus_j(\CC^{n_j}\otimes q_j\mathcal E_j)\subset\mathcal G$ that preserves every logical matrix factor while restricting the decoding environments. Its leakage obeys $\|(I-P)|_{\mathcal{G}_0}\|=\mathcal{O}_{\mathcal A}(\sqrt\tau)$, and the change in each environmental state is $\|\rho_j^0-\rho_j\|_1=\mathcal{O}_{\mathcal A}(\delta/\tau)$. 
Next, \cref{lem:subspace-rotation} gives a small unitary rotation $U$ that replaces the decoding isometry by $W_0=UW$ and ensures 
 $\mathcal{G}_0\subseteq\Ran W_0$. Finally, we purify $\rho_j^0$ using a new encoding environment $\mathcal Y_j^0$. This gives matching product isometries $J_j^0\xi=\xi\otimes\zeta_j^0$ with $\zeta_j^0\in q_j\mathcal E_j\otimes\mathcal Y_j^0$, so the pullback $V_j^0=(W_0^*\otimes I_{\mathcal Y_j^0})J_j^0$  defines the repaired encoding isometry. This yields the exactly idempotent UCP map $E_{\rm g}$
\begin{equation}\label{eq:surgery-scheme}
(W,\{V_j\})\xrightarrow{\text{cutoff, rotation}}(W_0,\{V_j^0\}),
\quad \Upsilon_0\Delta_0=\id_{\mathcal A},
\quad E_{\rm g}:=\Delta_0\Upsilon_0.
\end{equation}
 The errors in the UCP maps are controlled by two contributions: $\mathcal{O}_{\mathcal A}(\delta/\tau)$ from changing the decoding-environment states and $\mathcal{O}_{\mathcal A}(\sqrt\tau)$ from rotating the decoder. Choosing $\tau\sim\delta^{2/3}$ balances them and gives the cube-root estimate. The surgery modifies the Stinespring isometries. The stability estimate in
  \cref{lem:stinespring-continuity} then transfers these perturbation bounds to cb-norm bounds for the UCP maps. 

As discussed in \cref{sec:ucp-main}, we seek to improve two aspects of \cref{thm:main}: (1) the cube-root error bound, and (2) the possible need for the Choi--Effros product on the range of the resulting projection. The surgery picture explains how the metastable seminorm in \cref{def:metastable-seminorm} yields both a square-root error bound and a range closed under ordinary matrix multiplication:

\begin{figure}[h]
  \centering
  \includegraphics[width=0.8\textwidth]{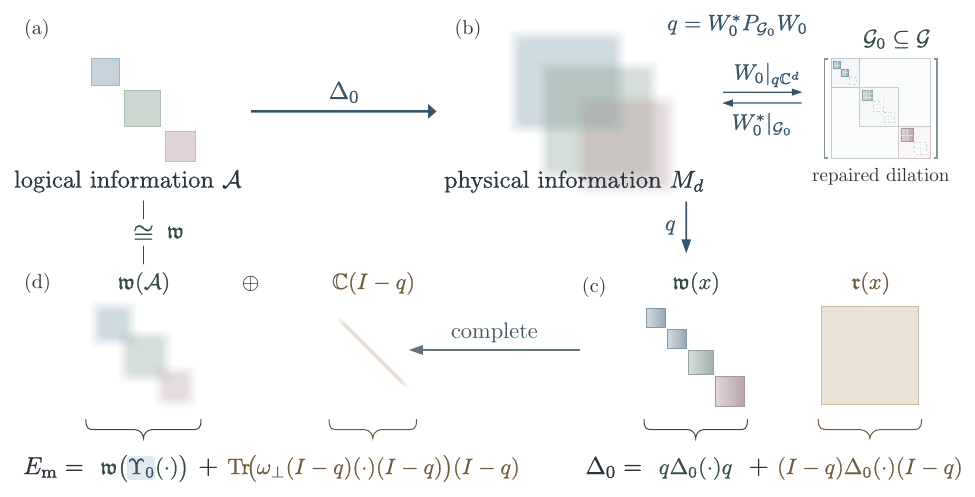}
  \caption{\raggedright A schematic illustration of the
  ambient completion in the metastable seminorm, read clockwise from (a) to (d). The projection $q$ identifies the physical subspace that corresponds to the cutoff dilation space $\mathcal G_0 \subset \mathcal G$. 
  The decomposition determined by $q$
  gives $\Delta_0=\mathfrak w+\mathfrak r$. Replacing the latter by a scalar component yields the conditional expectation $E_{\rm m}$ onto the ambient subalgebra $\mathfrak w(\mathcal A)\oplus\CC(I-q)$. }
  \label{fig:seminorm-completion}
\end{figure}

\begin{itemize}
\item[(1)] We keep the cutoff threshold $\tau$ fixed. The decoding-environment states then change by $\mathcal O_{\mathcal A}(\delta)$, and purifications in the original encoding environments $\mathcal Y_j$ can be chosen with $\|\zeta_j^0-\zeta_j\|=\mathcal O_{\mathcal A}(\sqrt\delta)$. Although $U-I$ need not be globally small, its action on the relevant product vectors $J_j^0\xi$ is of the same square-root order. 
\item[(2)] We pull back \(\mathcal G_0\) through the repaired encoding isometry \(W_0\), obtaining the physical sector \(q\CC^d=W_0^*\mathcal G_0\), which supports a faithful representation \(\mathfrak w(\mathcal A)\subset qM_dq\) of the logical algebra \(\mathcal A\). Thus \(\mathfrak w(\mathcal A)\) is realized as an ordinary matrix subalgebra with the ambient product. Since the encoder assigns only \(\|\Upsilon(I-q)\|=\mathcal O_{\mathcal A}(\delta)\) weight to the complementary sector, its action there can be replaced by a scalar component $\CC(I-q)$ at a cost of only \(\mathcal O_{\mathcal A}(\sqrt\delta)\) in the metastable seminorm. Equivalently, in the Schrödinger picture, the two-outcome measurement \(\{q,I-q\}\) finds the system in the sector carrying the logical algebra with probability at least \(1-\mathcal O_r(\eta)\), 
as discussed in \cref{thm:weighted-ambient}. The resulting conditional expectation has range
$  
\mathfrak w(\mathcal A)\oplus\CC(I-q)\subset M_d,
 $ 
an ordinary physical matrix subalgebra requiring no Choi--Effros construction. This completion preserves the improved error bound and yields the stronger algebraic conclusion. Figure~\ref{fig:seminorm-completion} illustrates this ambient completion.
\end{itemize}

\section{Surgery on the approximate idempotent factorization }
\label{sec:exactification}

We keep the physical-picture roles from \cref{sec:surgery-picture} in mind. The following proposition formalizes the surgery procedure for obtaining an exact idempotent factorization from an approximate one, with error bound in the ordinary cb-norm or the metastable seminorm. 
\begin{prop}[Surgery with cb-norm error]
\label{prop:exactification}
Let $\mathcal A$ be a finite-dimensional $C^*$-algebra.  There are constants $d_{\mathcal A},K_{\mathcal A}>0$ that only depend on the structure of $\mathcal A$, specifically, only the dimension and the largest direct summand in the decomposition of $\mathcal A$, with the following property.  If $\Delta:\mathcal A\to M_d$ and $\Upsilon:M_d\to\mathcal A$ are UCP and
\[
  \delta=\norm{\Upsilon\Delta-\id_{\mathcal A}}_{\mathrm{cb}}<d_{\mathcal A},
\]
then there are UCP maps $\Delta_0,\Upsilon_0$ such that
\begin{equation}
\label{eq:exact-retract}
  \Upsilon_0\Delta_0=\id_{\mathcal A}
\end{equation}
and
\begin{equation}
\label{eq:exactification-error}
  \norm{\Delta_0-\Delta}_{\mathrm{cb}}
  +\norm{\Upsilon_0-\Upsilon}_{\mathrm{cb}}
  \leq K_{\mathcal A}\delta^{1/3}.
\end{equation}
\end{prop}

\begin{prop}[Surgery with metastable-seminorm error]\label{prop:weighted-exactification}
Under the same setup as in \cref{prop:exactification}, 
there are UCP maps $\Delta_0:\mathcal A\to M_d$, $\Upsilon_0:M_d\to\mathcal A$, a projection $q\in M_d$, a faithful representation $\mathfrak w:\mathcal A\to qM_dq$ with unit $q$, and a UCP conditional expectation $E_{\rm m}$ onto the subalgebra $\mathfrak w(\mathcal A)\oplus\CC (I-q)$ of $M_d$, satisfying
\[
\Upsilon_0\Delta_0=\Upsilon_0\mathfrak w=\id_{\mathcal A},\quad
\Upsilon_0(I-q)=0,\quad
q\Delta_0(a)=\mathfrak w(a)=\Delta_0(a)q\quad(a\in\mathcal A),
\]
\begin{equation}\label{eq:weighted-retract-bounds}
\norm{\Upsilon_0-\Upsilon}_{\mathrm{cb}}\leq K_{\mathcal A}\sqrt\delta,
\quad
\norm{\Upsilon(I-q)}\leq K_{\mathcal A}\delta.
\end{equation}
Writing $\mathcal{T}=E_{\rm m}-\Delta\Upsilon$, $\mathcal{T}_m=\id_{M_m}\otimes\mathcal{T}$, and $\Upsilon_m=\id_{M_m}\otimes\Upsilon$, one also has, for every $m$ and $X\in M_m\otimes M_d$,
\begin{equation}\label{eq:weighted-retract-error}
\norm{\Upsilon_m(\mathcal{T}_m(X)^*\mathcal{T}_m(X))}^{1/2}
\leq K_{\mathcal A}\sqrt\delta\,\norm X.
\end{equation}
The scalar summand is omitted if $I-q=0$.
\end{prop}

We will first establish some auxiliary lemmas in \cref{sec:weighted-proof}. The main task in this section is to prove the surgery results \cref{prop:exactification,prop:weighted-exactification}, which will be addressed in \cref{sec:exactification-proof,sec:weighted-exactification} respectively. 
The proofs of the main theorems \cref{thm:main,thm:weighted-ambient} 
are applications of the surgery results above, which we will present in \cref{sec:cb-proof} and \cref{sec:weighted-proof} respectively.
\subsection{Auxiliary lemmas}\label{sec:auxiliary-lemmas}

We state some abstract lemmas that will be used in the proof of \cref{prop:exactification}. We first show that the approximate recovery property implies that the isometry $Z_j$ dilating $\Gamma_j\Delta$ is close to a product isometry $J_j:\xi \mapsto \xi \otimes \zeta_j$ on the $j$th logical block.

\begin{lem}[Product dilations on blocks]
\label{lem:product-dilation}
Let $\mathcal A=\bigoplus_iM_{n_i}$, let $\pi(x)=\bigoplus_i(x_i\otimes I_{\mathcal E_i})$, and write $\operatorname{pr}_j(x)=x_j$. For each $j$, suppose that the isometry
\[
Z_j:\CC^{n_j}\longrightarrow
\left(\bigoplus_i\CC^{n_i}\otimes\mathcal E_i\right)\otimes\mathcal Y_j
\]
satisfies
\[
\norm{Z_j^*(\pi(\,\cdot\,)\otimes I_{\mathcal Y_j})Z_j-\operatorname{pr}_j}_{\mathrm{cb}}
\leq\delta<1.
\]
Then there are unit vectors $\zeta_j\in\mathcal E_j\otimes\mathcal Y_j$ such that, for $J_j\xi=\xi\otimes\zeta_j$ in the $j$th summand,
\begin{equation}\label{eq:product-dilation}
\norm{Z_j-J_j}\leq\sqrt{2n_j\delta}.
\end{equation}
\end{lem}

In order to prove the product-dilation lemma on blocks, we first establish a simple version \cref{lem:product-dilation-leakage} of the result for a single block, with possible leakage. 
\begin{lem} 
\label{lem:product-dilation-leakage}
Let $V:\CC^d\to(\CC^d\otimes\mathcal E_1)\oplus\mathcal E_2$ be an isometry, written as $V\xi=V_1\xi\oplus V_2\xi$. Suppose that the completely positive map $T:M_d\to M_d$ is close to the identity:
\[
 T(x)=V_1^*(x\otimes I_{\mathcal E_1})V_1,\quad
\norm{T-\id_{M_d}}_{\mathrm{cb}}\leq\delta<1.
\]
Then there exists a unit vector $\zeta\in\mathcal E_1$ such that
\[
J_\zeta\xi=(\xi\otimes\zeta)\oplus0,\quad
\norm{V-J_\zeta}\leq\sqrt{2d\delta}.
\]
\end{lem}

\begin{proof}
Set
\[
\Omega=\frac1{\sqrt d}\sum_{a=1}^d e_a\otimes e_a,\quad
v=(I\otimes V)\Omega,\quad
v_1=(I\otimes V_1)\Omega,\quad
\sigma=\Tr_{\mathcal E_1}\ket{v_1}\!\bra{v_1}.
\]
Here $\{e_a\}_{a=1}^d$ is the standard basis of $\CC^d$. Now, let $\mathsf E_{ab}$ denote the elementary matrix with $1$ in the $(a,b)$ entry and zeros elsewhere.
The dual is $T_\ast(\omega)=\Tr_{\mathcal E_1}(V_1\omega V_1^*)$, so 
\[
\begin{aligned}
\sigma &= \frac{1}{d }\sum_{a,b}\Tr_{\mathcal E_1}\left(\mathsf E_{ab}\otimes V_1 \mathsf E_{ab} V_1^*  \right)= \frac{1}{d }\sum_{a,b}\mathsf E_{ab} \otimes T_\ast (\mathsf E_{ab})\\
& = (\id_{M_d}\otimes T_\ast)\Big(\sum_{a,b} \mathsf E_{ab} \otimes \mathsf E_{ab}\Big) = (T_\ast)_d(\dyad{\Omega}).
\end{aligned}
\]
The cb--diamond duality gives
\begin{equation}\label{eq:product-dilation-1}
\norm{\sigma-\ket{\Omega}\!\bra{\Omega}}_1
\leq\norm{T_\ast-\id_{M_d}}_\diamond 
=\norm{T-\id_{M_d}}_{\mathrm{cb}}\leq\delta.
\end{equation}
Set $\zeta'=(\bra{\Omega}\otimes I_{\mathcal E_1})v_1$. By the definition of partial trace and 
\cref{eq:product-dilation-1},
\[
\norm{\zeta'}^2
=\langle v_1,(\ket{\Omega}\!\bra{\Omega}\otimes I_{\mathcal E_1})v_1\rangle = \Tr[(\dyad{\Omega}\otimes I_{\mathcal E_1})\dyad{v_1}]
=\langle\Omega,\sigma\Omega\rangle
\geq1-\norm{\sigma-\ket{\Omega}\!\bra{\Omega}}_1\geq1-\delta.
\]
Thus $\zeta=\zeta'/\norm{\zeta'}$ is well defined. Since $V$ is an isometry, the full vector $v$ has norm one, including its leakage component. Moreover, $\langle v,(\Omega\otimes\zeta)\oplus0\rangle= \langle v_1, \Omega\otimes \zeta\rangle  = \langle \zeta',\zeta\rangle =\norm{\zeta'}$, and hence
\[
 \norm{v-((\Omega\otimes\zeta)\oplus0)}^2
=2-2\norm{\zeta'}
\leq2(1-\sqrt{1-\delta})\leq2\delta.
\]
Finally, orthogonality of the $e_a$ gives
\[
\norm{V-J_\zeta}^2
\leq\norm{V-J_\zeta}_{\mathrm{HS}}^2
=\sum_{a=1}^d\norm{(V-J_\zeta)e_a}^2
=d\norm{(I\otimes(V-J_\zeta))\Omega}^2
\leq2d\delta.
\]
\end{proof}

\begin{proof}[Proof of \cref{lem:product-dilation}]
  Recall from \cref{eq:block-coordinates} that $\operatorname{pr}_{j,\ast}: M_{n_j} \to\mathcal A$ is the canonical inclusion into the $j$th block of $\mathcal A$. 
 With the $j$th representation block placed first, we have
\[
\pi(0\oplus\cdots\oplus x\oplus\cdots\oplus0)\otimes I_{\mathcal Y_j} = \pi(\operatorname{pr}_{j,\ast} (x)) \otimes I_{\mathcal Y_j} 
=(x\otimes I_{\mathcal E_j\otimes\mathcal Y_j})\oplus0.
\]
Since $\operatorname{pr}_{j,\ast}$ is a completely-isometric map, the assumption implies
\[
\norm{Z_j^*\bigl((\,\cdot\,\otimes I_{\mathcal E_j\otimes\mathcal Y_j})\oplus0\bigr)Z_j-\id_{M_{n_j}}}_{\mathrm{cb}} = 
\norm{(Z_j^\ast (\pi(\cdot)\otimes I_{\mathcal Y_j})-\mathrm{pr}_j)\circ \operatorname{pr}_{j,\ast}}_{\mathrm{cb}}
\leq\delta.
\]
Apply \cref{lem:product-dilation-leakage} with its parameters chosen as
 $
d=n_j,  V=Z_j, 
\mathcal E_1=\mathcal E_j\otimes\mathcal Y_j, 
\mathcal E_2=\bigoplus_{i\ne j}\CC^{n_i}\otimes\mathcal E_i\otimes\mathcal Y_j
 $. With this decomposition we write $Z_j = V_1 \oplus V_2$ and thus \cref{lem:product-dilation-leakage} applies to the map $T_j(x ) = V_1^\ast (x \otimes I_{\mathcal E_j\otimes\mathcal Y_j}) V_1$.
Undoing the block permutation gives the isometry  $J_j$, whose only nonzero component $\xi\otimes\zeta_j \in \mathcal E_j\otimes\mathcal Y_j$ is in the $j$th block, with the error bound
 $
\|Z_j-J_j\|\le\sqrt{2n_j\delta}.
 $
\end{proof}
Next we give a generic cutoff approach to density matrices. This is a Markov's inequality-type estimate. 
In the application, $\rho$ describes the environmental state used by the product encodings $\{J_j\}$, 
and a positive operator $D$ measures leakage out of the decoder's physical dilation range. The lemma replaces $\rho$ by a nearby state $\rho_0$ supported on a subspace projection $q$ where leakage $D$ is uniformly bounded.  

\begin{lem} 
\label{lem:state-cutoff} 
Let $D\in M_d$, $D\ge 0$, let $\rho \in \mathscr D$ be a density matrix.  For every $\tau>\Tr(\rho D)$, there are a projection $q$ and a density matrix $\rho_0 \in \mathscr D$, supported on $q$, such that
\begin{equation}
\label{eq:state-cutoff}
  qDq\leq\tau q,
  \quad
  \left(1-\frac{\Tr(\rho D)}{\tau}\right)\rho_0\leq\rho,
  \quad
  \norm{\rho-\rho_0}_1\leq\frac{2\Tr(\rho D)}{\tau}.
\end{equation}
\end{lem}

\begin{proof}
Set
\[
  \mathsf{H}=\rho^{1/2}(D-\tau I)\rho^{1/2},
  \quad
  p=1_{(-\infty,0]}(\mathsf{H}),
  \quad
  \widetilde p =\rho^{1/2}p\rho^{1/2},
\] where $1_{(-\infty,0]}$ refers to the spectral projector of the self-adjoint operator $\mathsf{H}$ for the interval $(-\infty,0]$.
For every $v = \rho^{1/2}p \widetilde v\in\Ran(\rho^{1/2}p)$,
\[
  \langle v,(D-\tau I)v\rangle = 
  \Tr(\rho^{1/2}p\ket {\widetilde{v}}\!\bra{ \widetilde{v} }p\rho^{1/2}(D-\tau I)) = \Tr(p\mathsf{H}p\ket{\widetilde{v}}\!\bra{\widetilde{v}}) 
  \leq0.
\]
Let $q$ be the support projection of $\widetilde p$, so that the closure of $\Ran(\rho^{1/2}p)$ is $\Ran(q)$. The above inequality grants $qDq\leq\tau q$.  If
\[
  \beta=1-\Tr\widetilde p=\Tr((I-p)\rho),
\]
then positivity of $(I-p)\mathsf{H}(I-p)$ gives
\[
  \tau\beta = \Tr\!\left((I-p)\rho^{1/2}\tau I\rho^{1/2}(I-p)\right) 
  \leq\Tr\!\left((I-p)\rho^{1/2}D\rho^{1/2}(I-p)\right)
  \leq\Tr(\rho D).
\]
Since $0\leq\widetilde p\leq\rho$, the state $\rho_0=\widetilde p/(1-\beta)$ satisfies
\begin{equation}\left(1-\frac{\Tr(\rho D)}{\tau}\right)\rho_0
    \leq(1-\beta)\rho_0=\widetilde p\leq\rho, \quad 
   \norm{\rho-\rho_0}_1 \leq\norm{\rho-\widetilde p}_1+\norm{\widetilde p-\rho_0}_1
  =2\beta\leq\frac{2\Tr(\rho D)}{\tau}.
\end{equation}
\end{proof}

The following estimate can be viewed as a version of Hayashi's pinching inequality \cite[Appendix~B, Lemma~2]{OgawaHayashi2001}. The error bound depends only on the logical dimension $n$, not on the size $k$ of the environment.

\begin{lem} 
\label{lem:positive-partial-trace}
For every positive $A\in M_n\otimes M_k$,
\begin{equation}
\label{eq:positive-partial-trace}
  A\leq nI_n\otimes\Tr_{\CC^n}A.
\end{equation}
\end{lem}

\begin{proof}
Let $A=\ket{\psi}\!\bra{\psi}$ and take a Schmidt decomposition
\[
\psi=\sum_{\ell=1}^{r_\psi}\sqrt{\lambda_\ell}\,e_\ell\otimes f_\ell,
\quad
r_\psi\leq n,\quad \lambda_\ell \ge 0 .
\]
Then $ \Tr_{\CC^n}A=\sum_\ell\lambda_\ell\ket{f_\ell}\!\bra{f_\ell}$.  For every $v\in\CC^n\otimes\CC^k$, Cauchy--Schwarz gives
\[
\langle v,Av\rangle
\leq r_\psi\sum_\ell\lambda_\ell
|\langle e_\ell\otimes f_\ell,v\rangle|^2
\leq r_\psi\langle v,(I_n\otimes \Tr_{\CC^n}A)v\rangle
\leq n\langle v,(I_n\otimes \Tr_{\CC^n}A)v\rangle.
\]
Thus $A\leq nI_n\otimes\Tr_{\CC^n}A$.  Summing a rank-one spectral decomposition 
proves \cref{eq:positive-partial-trace} for general $A \ge 0 $.
\end{proof}
In our application, the cutoff product sector is only approximately contained in the physical dilation range. A unitary rotation makes this inclusion exact while preserving the Stinespring form, hence complete positivity and unitality. The required finite-dimensional direct rotation follows from singular-value decomposition (SVD) constructions and subspace-rotation estimates (see, e.g., \cite[Theorems~5.2 and~5.5]{StewartSun1990} and \cite{DavisKahan1970,Kato1995}).
\begin{lem} 
\label{lem:subspace-rotation}
Let $P$ be an orthogonal projection on a finite-dimensional Hilbert space and let $\mathcal{G}_0$ be a subspace such that
\[
  \norm{(I-P)|_{\mathcal{G}_0}}\leq\varepsilon<1.
\]
Then there is a unitary $U$ such that
\begin{equation}
\label{eq:subspace-rotation}
  \mathcal{G}_0\subseteq U\Ran P,
  \quad
  \norm{U-I}\leq2\varepsilon.
\end{equation}
\end{lem}

\begin{proof}
If $v\in\mathcal{G}_0$ and $Pv=0$, then $\norm{(I-P)v}=\norm v$, so $v=0$.  Thus $P|_{\mathcal{G}_0}$ is injective and $\mathcal{G}_0'=P\mathcal{G}_0\subseteq\Ran P$ has the same dimension as $\mathcal{G}_0$.  Moreover, $P_{\mathcal{G}_0'}v=Pv$ for $v\in\mathcal{G}_0$.

Take an SVD of $P|_{\mathcal{G}_0}:\mathcal{G}_0\to\mathcal{G}_0'$. Since the singular values are all at most $1$, this gives orthonormal bases $\mathbf{s}_k\in\mathcal{G}_0$ and $\mathbf{r}_k\in\mathcal{G}_0'$, with $0\leq\theta_k<\pi/2$, such that
\[
  P\mathbf{s}_k=\cos\theta_k\,\mathbf{r}_k,
  \quad
  \sin\theta_k=\norm{(I-P)\mathbf{s}_k}\leq\varepsilon.
\]
For $\theta_k>0$, set
\[
  \mathbf{t}_k=\frac{\mathbf{s}_k-\cos\theta_k\,\mathbf{r}_k}{\sin\theta_k}\in\ker P,
  \quad
  \mathbf{s}_k=\cos\theta_k\,\mathbf{r}_k+\sin\theta_k\,\mathbf{t}_k.
\]
The vectors $\mathbf{t}_k$ are orthonormal because
\[
  \langle(I-P)\mathbf{s}_i,(I-P)\mathbf{s}_j\rangle
  =\delta_{ij}\sin^2\theta_i.
\]
Each pair $\mathbf{r}_k\in \Ran P,\mathbf{t}_k\in \ker P$ forms an orthonormal basis for a plane, and the above ensures that these planes are all mutually orthogonal. Define
\[
  U\mathbf{r}_k=\cos\theta_k\,\mathbf{r}_k+\sin\theta_k\,\mathbf{t}_k=\mathbf{s}_k,
  \quad
  U\mathbf{t}_k=-\sin\theta_k\,\mathbf{r}_k+\cos\theta_k\,\mathbf{t}_k
\]
on each nontrivial plane, and let $U$ be the identity on the orthogonal complement.  For $\theta_k=0$, $U$ fixes $\mathbf{r}_k=\mathbf{s}_k$.  In the ordered basis $(\mathbf{r}_k,\mathbf{t}_k)$,
\[
  U_k-I_2=
  \begin{pmatrix}
    \cos\theta_k-1&-\sin\theta_k\\
    \sin\theta_k&\cos\theta_k-1
  \end{pmatrix},
  \quad
  (U_k-I_2)^*(U_k-I_2)=4\sin^2(\theta_k/2)I_2.
\]
Therefore
\[
  \norm{U-I}
  =2\max_k\sin(\theta_k/2)
  \leq2\max_k\sin\theta_k
  \leq2\varepsilon.
\]
Since $U\mathcal{G}_0'=\mathcal{G}_0$ and $\mathcal{G}_0'\subseteq\Ran P$, \cref{eq:subspace-rotation} follows.
\end{proof}
With these lemmas in hand, we are ready to do the surgery to the Stinespring dilations to obtain the exact idempotent UCP.  The next estimate is the Schrödinger-picture form of the continuity theorem for Stinespring dilations \cite[Theorem~1]{KretschmannSchlingemannWerner2008}.
It supplies the operational error bound: nearby isometries produce nearby channels after the environment is discarded, uniformly over reference systems. We will use it for the repaired decoder and encoder.
\begin{lem}[Continuity under Stinespring perturbations]
\label{lem:stinespring-continuity}
Let $V_1,V_2:\mathcal H\to\mathcal U\otimes\mathcal W$ be isometries and define
\[
  \mathscr{C}_{V_1}(\omega)=\Tr_{\mathcal W}(V_1\omega V_1^*),
  \quad
  \mathscr{C}_{V_2}(\omega)=\Tr_{\mathcal W}(V_2\omega V_2^*). 
\]
Then
\begin{equation}
\label{eq:stinespring-continuity}
  \norm{\mathscr{C}_{V_1}-\mathscr{C}_{V_2}}_\diamond
  \leq2\norm{V_1-V_2}. 
\end{equation}
\end{lem}

\begin{proof}
For any reference space $\mathcal H'$ and density matrix $\omega$ on $\mathcal H'\otimes\mathcal H$, trace-norm contractivity under partial trace \cite[Eq.~(1.183)]{Watrous2018} and $\norm{AXB}_1\leq\norm A\norm X_1\norm B$ give
\begin{align*}
&~~~~\norm{\Tr_{\mathcal W}\!\left[(I_{\mathcal H'}\otimes V_1)\omega(I_{\mathcal H'}\otimes V_1^*)
-(I_{\mathcal H'}\otimes V_2)\omega(I_{\mathcal H'}\otimes V_2^*)\right]}_1\\
& \leq
\bigl(\norm{V_1-V_2}\norm{V_1}+\norm{V_2}\norm{V_1-V_2}\bigr)\norm{\omega}_1
=2\norm{V_1-V_2}. 
\end{align*}
Taking the supremum over $\mathcal H'$ and $\omega$ proves \cref{eq:stinespring-continuity}.
\end{proof}

\subsection{Proof of \cref{prop:exactification}}\label{sec:exactification-proof}

\begin{proof}[Proof of \cref{prop:exactification}] 
If $\delta=0$, take $\Delta_0=\Delta$ and $\Upsilon_0=\Upsilon$.  Henceforth assume $\delta>0$ and write
\[
  \mathcal A=\bigoplus_{j=1}^s M_{n_j},
  \quad
  n_*=\max_j n_j,
  \quad
  \dim\mathcal A=\sum_{j=1}^s n_j^2.
\]

\proofstep{Step 0: Stinespring setup (see \cref{sec:stinespring-picture,sec:surgery-picture})}
Recall from \cref{sec:stinespring-picture} that we write the Stinespring dilation of $\Delta$ as
\begin{equation}
\label{eq:Delta-stinespring}
  \Delta(x)=W^*\pi(x)W,
  \quad
\pi(x)=\bigoplus_{j=1}^s(x_j\otimes I_{\mathcal E_j}),\quad \text{on~}
  \mathcal G=\bigoplus_{j=1}^s(\CC^{n_j}\otimes\mathcal E_j),\end{equation}
where $W:\CC^d\to\mathcal G$ is an isometry.  Put $P=WW^* \in \mathcal B(\mathcal G)$. {The $j$th component $\Upsilon_j:M_d\to M_{n_j}$ given by $\Upsilon_j = \operatorname{pr}_j\circ \Upsilon$ is also UCP and so admits a similar decomposition:}
\begin{equation}\label{eq:Upsilon-stinespring}
 {\Upsilon_j(a) = V_j^* (
    a\otimes I_{\mathcal Y_j}) V_j \quad \text{on } \mathbb{C}^d\otimes \mathcal Y_j,}
\end{equation}  {where $V_j:\CC^{n_j}\to\CC^d\otimes\mathcal Y_j$ is an isometry.}
{We may then construct an isometry $Z_j: \CC^{n_j}\to \mathcal{G}\otimes \mathcal Y_j$ by}
\begin{equation}
\label{eq:composite-isometry}
  Z_j=(W\otimes I_{\mathcal Y_j})V_j,
\end{equation}{which then enjoys the property}
 $
  Z_j^*(\pi(\,\cdot\,)\otimes I_{\mathcal Y_j})Z_j
  =\Upsilon_j\Delta
  =\operatorname{pr}_j \Upsilon\Delta
 $. Therefore 
 \begin{equation}
  \norm{Z_j^*(\pi(\,\cdot\,)\otimes I_{\mathcal Y_j})Z_j-\operatorname{pr}_j}_{\mathrm{cb}}
  \leq\norm{\Upsilon\Delta-\id_{\mathcal A}}_{\mathrm{cb}}
  =\delta.
 \end{equation}

\proofstep{Step 1: product dilation.}
The composite isometry $Z_j$ records an encoding by $V_j$ followed by a decoding by $W$, with the encoding environment $\mathcal Y_j$ and decoding environments $\mathcal E_i$. Approximate recovery says that the logical factor is nearly unchanged by this round trip. In an ideal product dilation, the logical input $\xi \in \CC^{n_j}$ survives intact and the environment is in a fixed state $\zeta_j$ independent of $\xi$. The product-dilation \cref{lem:product-dilation} makes this interpretation quantitative. \Cref{lem:product-dilation} applied to $Z_j$ gives the unit vector
\[
  \zeta_j\in\mathcal E_j\otimes\mathcal Y_j,
  \quad
  J_j\xi=\xi\otimes\zeta_j\quad\text{in the $j$th summand},
\]
such that
\begin{equation}
\label{eq:Zj-Jj}
  \norm{Z_j-J_j}
  \leq\sqrt{2n_j\delta}
  \leq\sqrt{2n_*\delta}.
\end{equation}
The product form has the exact logical action $J_j^*(\pi(x)\otimes I)J_j=x_j$ (see also later in \cref{eq:exact-retract-1}). However, it may not be realizable through the original decoder: $\Ran J_j$ need not lie in $\Ran W\otimes\mathcal Y_j$. Pulling it back by $W^*\otimes I$ need not produce an isometry or define a unital Stinespring map. The cutoff and rotation below create exact containment in a modified physical range, after which this pullback becomes legitimate.

 We have
 $
  (P\otimes I_{\mathcal Y_j})Z_j=Z_j.
 $ 
Thus
\[
  \Ran Z_j\subseteq\Ran(P\otimes I_{\mathcal Y_j})
  \quad\text{and}\quad
  ((I-P)\otimes I_{\mathcal Y_j})Z_j=0.
\]
Consequently, \cref{eq:Zj-Jj} gives the leakage bound, which quantifies how much of the $J_j$ fails to lie within the range of $P\otimes I_{\mathcal Y_j}$.
\begin{equation}
\label{eq:J-leakage}
\begin{aligned}
  \norm{((I-P)\otimes I_{\mathcal Y_j})J_j}
  =\norm{((I-P)\otimes I_{\mathcal Y_j})(J_j-Z_j)} \leq\norm{J_j-Z_j}
  \le \sqrt{2n_\ast \delta}.
\end{aligned}
\end{equation}
\proofstep{Step 2: cutoff.}
   We let $Q_j$ be the projection onto $\CC^{n_j}\otimes\mathcal E_j$ in $\mathcal G$, and define $A_j=Q_j(I-P)Q_j \in \mathcal B(\mathcal G)$, which is the leakage of the $j$-th representation block.  For $v_j\in\CC^{n_j}\otimes  \mathcal E_j$
\[
  \langle v_j,(I-P)v_j\rangle = \langle v_j, Q_j (I-P)Q_j v_j\rangle
  =\langle v_j,A_jv_j\rangle  
\]
Since $I-P$ is a projection, squaring \cref{eq:J-leakage} gives
\[
  J_j^*((I-P)\otimes I_{\mathcal Y_j})J_j
  \leq2n_*\delta I_{n_j}.
\]
The range of $J_j$ lies in the $j$th block, so $I-P$ may be replaced by $A_j$.  Using
$J_jJ_j^*=I_{n_j}\otimes\ket{\zeta_j}\!\bra{\zeta_j}$ and taking the normalized trace yields
\begin{equation}\label{eq:J-leakage-3}
  2n_*\delta \ge 
\frac1{n_j}\Tr\!\left(J_j^*(A_j\otimes I_{\mathcal Y_j})J_j\right)
 =\frac1{n_j}\Tr\!\left((A_j\otimes I_{\mathcal Y_j})
(I_{n_j}\otimes\ket{\zeta_j}\!\bra{\zeta_j})\right)  
=:\Tr(\rho_jD_j).
\end{equation}
Here we denote
\begin{equation}
  \rho_j=\Tr_{\mathcal Y_j}\ket{\zeta_j}\!\bra{\zeta_j} \in \mathcal B(\mathcal E_j), 
  \quad
  D_j=\frac1{n_j}\Tr_{\CC^{n_j}}A_j \in \mathcal B(\mathcal E_j).
\end{equation}
$\Tr(\rho_j D_j)$ quantifies the weight $\rho_j $ outside $\Ran W$. 
  Applying \cref{lem:state-cutoff} to $(D_j,\rho_j)$ and using \cref{eq:J-leakage-3}, we obtain $q_j$ and $\rho_j^0$, supported on $q_j$, such that
\begin{equation}
\label{eq:rho-cutoff}
  q_jD_jq_j\leq\tau q_j,
  \quad
  \norm{\rho_j-\rho_j^0}_1
  \leq4n_*\frac{\delta}{\tau},\quad \forall \tau>2n_*\delta.
\end{equation}
These two estimates have distinct uses. The operator bound controls the geometry of the cutoff environmental support. We define $\mathcal{G}_0 := \bigoplus_{j=1}^s (\CC^{n_j}\otimes q_j\mathcal E_j)\subset \mathcal G$. 
The trace-norm bound controls the change in the decoding-environment state used to reconstruct the encoder: the cost of replacing $\rho_j$ by the $\rho_j^0$ supported on $q_j\mathcal E_j$ is uniformly upper bounded by $\mathcal O(\delta/\tau)$. 
Since $\Tr_{\CC^{n_j}}A_j=n_jD_j$, \cref{lem:positive-partial-trace} gives
\begin{equation}\label{eq:positive-partial-trace-2}
  A_j\leq n_j I_{n_j}\otimes \Tr_{\CC^{n_j}}A_j
  =
  n_j^2 I_{n_j}\otimes D_j.
\end{equation}
Let $v_j\in \CC^{n_j}\otimes q_j\mathcal E_j$. Then we have $q_jv_j=v_j$ and thus 
using \cref{eq:positive-partial-trace,eq:positive-partial-trace-2} gives $\langle v_j,(I-P)v_j\rangle =\langle v_j, A_jv_j\rangle \leq n_j^2\langle v_j,(I_{n_j}\otimes D_j)v_j\rangle \le n_j^2\tau \norm{v_j}^2$. 
Since $I-P\geq0$, Cauchy--Schwarz for the semidefinite form induced by $I-P$ gives
\begin{equation}
|\langle v_j,(I-P)v_k\rangle|
    \leq\langle v_j,(I-P)v_j\rangle^{1/2}
\langle v_k,(I-P)v_k\rangle^{1/2} \leq n_jn_k\tau\norm{v_j}\norm{v_k},
\end{equation}
Therefore, for any $v=\sum_jv_j\in\mathcal{G}_0$, by the Cauchy--Schwarz inequality,
\begin{align*}
  \langle v,(I-P)v\rangle
  &\leq\tau\left(\sum_jn_j\norm{v_j}\right)^2 \leq\tau\left(\sum_jn_j^2\right)\norm v^2
  =\tau\dim(\mathcal A)\norm v^2.
\end{align*}
Since $I-P$ is a projection, we have proved
\begin{equation}
\label{eq:slow-subspace-leakage}
  \norm{(I-P)|_{\mathcal{G}_0}}
  \leq\sqrt{\tau\dim\mathcal A}.
\end{equation}

\proofstep{Step 3: rotate the Stinespring space.}
If $\tau\dim\mathcal A<1$, \cref{lem:subspace-rotation,eq:slow-subspace-leakage} give a unitary $U$ on $\mathcal G$ such that
\begin{equation}
\label{eq:rotation-bound}
  \mathcal{G}_0\subseteq U\Ran P,
  \quad
  \norm{U-I}\leq2\sqrt{\tau\dim\mathcal A}.
\end{equation}
Set the rotated isometry (which is still an isometry because $U$ is unitary) and the repaired decoder
\begin{equation} 
\label{eq:Delta0}
  W_0=UW,
  \quad
  \Delta_0(x)=W_0^*\pi(x)W_0.
\end{equation}
Then $\Delta_0$ is UCP, because $W_0$ is an isometry and $\pi$ is a $\ast$-homomorphism which means that $\Delta_0$ has a Stinespring form, and the dual of \cref{lem:stinespring-continuity} gives
\begin{equation}
\label{eq:Delta-error}
  \norm{\Delta_0-\Delta}_{\mathrm{cb}}
  \leq2\norm{U-I}
  \leq4\sqrt{\tau\dim\mathcal A}.
\end{equation}
  The inclusion $\mathcal{G}_0\subseteq\Ran W_0$ ensures that any unit vector $\zeta_j^0\in q_j\mathcal E_j\otimes\mathcal Y_j^0$ can be pulled back by the isometry $W_0^*\otimes I_{\mathcal Y_j^0}$. Step 4 uses precisely these states to construct the product isometry and 
  to repair the encoder. 

\proofstep{Step 4: construction of $\Upsilon_0$ via purification and comparison.} We now purify $\rho_j^0$ to define the modified product isometry $J_j^0$ (and thus the modified $\Upsilon_0$). We recall that $\rho_j^0$ is supported on $q_j\mathcal E_j$, so 
we can take a spectral decomposition
\[
\rho_j^0=\sum_{\ell=1}^{r_j}p_\ell
\ket{u_\ell}\!\bra{u_\ell},
\quad
r_j=\rank\rho_j^0,
\quad
u_\ell\in q_j\mathcal E_j.
\]
Choose a new encoding environment $\mathcal Y_j^0=\CC^{r_j}$, serving as a purifying auxiliary space, with orthonormal basis $(f_\ell)_\ell$ and set the purification
\[
\zeta_j^0=\sum_{\ell=1}^{r_j}\sqrt{p_\ell}\,
u_\ell\otimes f_\ell \in q_j \mathcal E_j\otimes\mathcal Y_j^0. 
\]
Since $\sum_\ell p_\ell=1$, this is a unit vector, and orthonormality of the $f_\ell$ gives
 $ 
\Tr_{\mathcal Y_j^0}\ket{\zeta_j^0}\!\bra{\zeta_j^0}=\rho_j^0.
 $ 
Set the product isometry $J_j^0\xi=\xi\otimes\zeta_j^0$.  The definition of $\mathcal{G}_0$ gives $\Ran J_j^0\subseteq\mathcal{G}_0\otimes\mathcal Y_j^0$.  Moreover, the rotation constructed above satisfies
 $
\mathcal{G}_0\subseteq U\Ran P
=\Ran(UW)=\Ran W_0.
 $ 
Therefore
 $
  \Ran J_j^0
  \subseteq\mathcal{G}_0\otimes\mathcal Y_j^0
  \subseteq\Ran(W_0\otimes I_{\mathcal Y_j^0}).
 $ 
 Define
\begin{equation}
\label{eq:Upsilon0}
  (\Upsilon_0(a))_j
  =(V_j^0)^*(a\otimes I_{\mathcal Y_j^0})V_j^0,\quad V_j^0=(W_0\otimes I_{\mathcal Y_j^0})^*J_j^0.
\end{equation}
The new encoding isometry $V_j^0$ first appends the purified environmental state $\zeta_j^0$ to a logical input and then pulls this product vector back into the physical system using $W_0^*$. 
Since $V_j^0$ is an isometry, the Stinespring form  $\Upsilon_0:M_d\to\bigoplus_jM_{n_j}$ is also UCP. 

We now verify the exact-recovery property. 
 $\pi(x)$ acts as $x_j\otimes I_{\mathcal E_j}$ on the $j$th block, so
\begin{equation}\label{eq:exact-retract-1}
(\pi(x)\otimes I_{\mathcal Y_j^0})J_j^0\xi
=x_j\xi\otimes\zeta_j^0
=J_j^0x_j\xi,\quad \forall \xi\in\CC^{n_j},\quad \text{i.e.~} (\pi(x)\otimes I_{\mathcal Y_j^0})J_j^0=J_j^0x_j.
\end{equation}
Put $P_0=W_0W_0^*$.  Since $\Ran J_j^0\subseteq\mathcal{G}_0\otimes\mathcal Y_j^0$ and $\mathcal{G}_0\subseteq\Ran P_0$,
 $
(P_0\otimes I_{\mathcal Y_j^0})J_j^0=J_j^0. 
 $  
Thus
\begin{equation}
\label{eq:exact-left-inverse}
\begin{aligned}
  (\Upsilon_0\Delta_0(x))_j &= (V_j^0)^* (W_0^*\pi(x)W_0 \otimes I_{\mathcal Y_j^0})V_j^0  
 =(J_j^0)^*(P_0\pi(x)P_0\otimes I_{\mathcal Y_j^0})J_j^0 \\
 & = (J_j^0)^*(\pi(x)\otimes I_{\mathcal Y_j^0})J_j^0
  \note{\cref{eq:exact-retract-1}}=(J_j^0)^*J_j^0x_j=x_j.
\end{aligned}
\end{equation}
This proves \cref{eq:exact-retract}.  

We already have the estimate \cref{eq:Delta-error} for $\norm{\Delta_0-\Delta}_{\mathrm{cb}}$. 
It remains to compare $\Upsilon_0$ and $\Upsilon$.  Write $\Ad_W(X)=WXW^*$.  From \cref{eq:composite-isometry},
 $
(\Ad_W\circ\Upsilon_{j,\ast})(\omega)
=\Tr_{\mathcal Y_j}(Z_j\omega Z_j^*).
 $ 
Moreover,
\[
\Tr_{\mathcal Y_j}(J_j\omega J_j^*)
=\Tr_{\mathcal Y_j}\bigl((I\otimes\zeta_j)\omega(I\otimes\zeta_j)^*\bigr)
=\omega\otimes\rho_j.
\]
Applying \cref{lem:stinespring-continuity,eq:Zj-Jj} gives
\begin{equation}\label{eq:product-channel}
\begin{aligned}
\norm{\Ad_W\circ\Upsilon_{j,\ast}-(\omega\mapsto\omega\otimes\rho_j)}_\diamond
 =\norm{\omega\mapsto\Tr_{\mathcal Y_j}
(Z_j\omega Z_j^*-J_j\omega J_j^*)}_\diamond \leq2\norm{Z_j-J_j}
\leq2\sqrt{2n_*\delta}.
\end{aligned}
\end{equation}
For the modified channel, the corresponding identity is exact:
\begin{equation}\label{eq:modified-product-channel}
(\Ad_{W_0}\circ\Upsilon_{j,\ast}^0)(\omega)=\omega\otimes\rho_j^0.
\end{equation}
Insert successively the maps $\omega\mapsto\omega\otimes\rho_j^0$, $\omega\mapsto\omega\otimes\rho_j$, and $\Ad_W\circ\Upsilon_{j,\ast}$.  The triangle inequality gives
\begin{equation}\label{eq:dual-error}
\begin{aligned}
\norm{\Upsilon_{j,\ast}^0-\Upsilon_{j,\ast}}_\diamond
&=\norm{\Ad_{W_0}\circ\Upsilon_{j,\ast}^0-\Ad_{W_0}\circ\Upsilon_{j,\ast}}_\diamond\\
&\leq\norm{\Ad_{W_0}\circ\Upsilon_{j,\ast}^0-(\omega\mapsto\omega\otimes\rho_j^0)}_\diamond +\norm{(\omega\mapsto\omega\otimes\rho_j^0)-(\omega\mapsto\omega\otimes\rho_j)}_\diamond\\
&\quad+\norm{(\omega\mapsto\omega\otimes\rho_j)-\Ad_W\circ\Upsilon_{j,\ast}}_\diamond +\norm{\Ad_W\circ\Upsilon_{j,\ast}-\Ad_{W_0}\circ\Upsilon_{j,\ast}}_\diamond\\
&\leq0+\norm{\rho_j^0-\rho_j}_1+2\sqrt{2n_*\delta}+2\norm{W_0-W}\\
& \le 4n_*\frac{\delta}{\tau}+2\sqrt{2n_*\delta}
+4\sqrt{\tau\dim\mathcal A}.
\end{aligned}
\end{equation}
The first term vanishes by \cref{eq:modified-product-channel}.  The second equals $\norm{\rho_j^0-\rho_j}_1$ by multiplicativity of the trace norm under tensor products and is bounded using \cref{eq:rho-cutoff}.  The third uses \cref{eq:product-channel}; the fourth uses \cref{lem:stinespring-continuity} followed by \cref{eq:rotation-bound}.
For a map $T:M_d\to\bigoplus_jM_{n_j}$,
 $
  \norm T_{\mathrm{cb}}=\max_j\norm{\operatorname{pr}_jT}_{\mathrm{cb}}.
 $ 
Applying cb--diamond duality componentwise to \cref{eq:dual-error} gives
\begin{equation}
\label{eq:Upsilon-error}
  \norm{\Upsilon_0-\Upsilon}_{\mathrm{cb}}
  \leq4n_*\frac{\delta}{\tau}
  +2\sqrt{2n_*\delta}
  +4\sqrt{\tau\dim\mathcal A}.
\end{equation}
The algebra dependence is explicit: only the structure of the reduced algebra $\mathcal A$ ($n_*$ and $\dim\mathcal A$) enters (but not the total dimension 
$d$ or the Stinespring multiplicities). Take
\begin{equation}\label{eq:dA}
  d_{\mathcal A}
  =\min\left\{1,(2n_*)^{-3},(\dim\mathcal A)^{-3/2}\right\}.
\end{equation}
For $0<\delta<d_{\mathcal A}$, set $\tau=\delta^{2/3}$.  Then
 $
  \tau>2n_*\delta,
  \sqrt{\tau\dim\mathcal A}<1$ by the choice of $d_{\mathcal A}$ in \cref{eq:dA}, 
so the cutoff and rotation steps apply.  Combining \cref{eq:Delta-error,eq:Upsilon-error} and using $\delta^{1/2}\leq\delta^{1/3}$ gives
\begin{align*}
\norm{\Delta_0-\Delta}_{\mathrm{cb}}
+\norm{\Upsilon_0-\Upsilon}_{\mathrm{cb}}
&\leq4n_*\delta^{1/3}+2\sqrt{2n_*}\,\delta^{1/2}
+8\sqrt{\dim\mathcal A}\,\delta^{1/3} \leq K_{\mathcal A}\delta^{1/3},
\end{align*}
where one may take
\begin{equation}\label{eq:KA}
  K_{\mathcal A}=4n_*+2\sqrt{2n_*}+8\sqrt{\dim\mathcal A}
\end{equation}
which only depends on the structure of $\mathcal A$. 
This proves \cref{eq:exactification-error}.
\end{proof}

\subsection{Proof of \cref{thm:main}}\label{sec:cb-proof}
With \cref{prop:exactification} at hand, the proof of \cref{thm:main} is very straightforward.
\begin{proof}[Proof of \cref{thm:main}]
Apply Kitaev's factorization theorem \cref{thm:kitaev-factorization}.  There are $\mathcal A,\Delta,\Upsilon$ such that
\begin{equation}
\label{eq:theorem-data}
  \dim\mathcal A\leq r,
  \quad
  \norm{\Delta\Upsilon-\Phi}_{\mathrm{cb}}\leq c\eta,
  \quad
  \delta:=\norm{\Upsilon\Delta-\id_{\mathcal A}}_{\mathrm{cb}}\leq c\eta.
\end{equation}
Only finitely many $C^*$-algebras occur, up to isomorphism, with dimension at most $r$.  Define 
\[
  d_r=\min_{\dim\mathcal A\leq r}d_{\mathcal A},
  \quad
  K_r=\max_{\dim\mathcal A\leq r}K_{\mathcal A},
\]
where the extrema range over isomorphism classes and thus $d_r$ and $K_r$ are finite and depend only on $r$.  Choose
 $
  \eta_r<\min\left\{\frac14,\eta_0,\frac{d_r}{c},1\right\}.
 $ 
Then \cref{prop:exactification} gives UCP maps $\Delta_0,\Upsilon_0$ satisfying
\[
  \Upsilon_0\Delta_0=\id_{\mathcal A},
  \quad
  \norm{\Delta_0-\Delta}_{\mathrm{cb}}
  +\norm{\Upsilon_0-\Upsilon}_{\mathrm{cb}}
  \leq K_r(c\eta)^{1/3}.
\]
Set $E_{\rm g}=\Delta_0\Upsilon_0$.  Then $ E_{\rm g }$ is UCP and
 $
  E_{\rm g }^2=\Delta_0(\Upsilon_0\Delta_0)\Upsilon_0=E_{\rm g }.
 $ 
Every UCP map has cb norm one \cite[Proposition~3.6]{Paulsen2003}.  Consequently,
\begin{align*}
\norm{ E_{\rm g }-\Phi}_{\mathrm{cb}}
&\leq\norm{\Delta_0\Upsilon_0-\Delta\Upsilon}_{\mathrm{cb}}
+\norm{\Delta\Upsilon-\Phi}_{\mathrm{cb}} \leq\norm{\Delta_0-\Delta}_{\mathrm{cb}}
+\norm{\Upsilon_0-\Upsilon}_{\mathrm{cb}}+c\eta\\
&\leq\bigl(K_rc^{1/3}+c\bigr)\eta^{1/3}.
\end{align*}
Thus one may take $C_r=K_rc^{1/3}+c$.
\end{proof}

\subsection{Proof of \cref{prop:weighted-exactification}}\label{sec:weighted-exactification}
The metastable-seminorm version of \cref{prop:exactification} uses the same cutoff and rotation construction. Most of the proof can be carried over from the cb-norm case. 

\begin{proof}\label{proof:weighted-exactification}
  We only need to modify the proof of \cref{prop:exactification} as follows, retaining its notation $W,\pi,P,V_j,Z_j,J_j,\rho_j,D_j$. The $q_j$, $\rho_j^0$, $\Delta_0$ and $\Upsilon_0$ are also constructed in the same way, but with slightly adjusted thresholds and error controls to accommodate the metastable-seminorm setting.

\proofstep{Step 1: product dilation.}
Reuse \cref{eq:Zj-Jj,eq:J-leakage}, which bound both the isometry error and the leakage by $\sqrt{2n_*\delta}$, with $n_*=\max_jn_j$.

\proofstep{Step 2: a fixed cutoff and nearby purifications.}
Use the same state-cutoff construction, now with the \emph{fixed} threshold $\tau$:
\[
\tau=\frac1{4\dim\mathcal A},\quad
d_{\mathcal A}=\min\{1,\tau/(4n_*)\}.
\]
For $\delta<d_{\mathcal A}$, we have $\Tr(\rho_jD_j)\leq2n_*\delta<\tau/2$.  Thus \cref{lem:state-cutoff} gives a state $\rho_j^0$ supported on $q_j$ with
\begin{equation}\label{eq:rho-cutoff-2}
q_jD_jq_j\leq\tau q_j,\quad
\left(1-\frac{2n_*\delta}{\tau}\right)\rho_j^0
\leq\rho_j.
\end{equation}
For $\mathcal{G}_0=\bigoplus_j(\CC^{n_j}\otimes q_j\mathcal E_j)$, \cref{eq:slow-subspace-leakage} then yields
\[
\norm{(I-P)|_{\mathcal{G}_0}}\leq\sqrt{\tau\dim\mathcal A}=\frac12.
\]
The fixed threshold gives enough geometric control to perform the rotation, although its global operator norm need not tend to zero with $\delta$. Meanwhile, the environmental state $\rho_j\to \rho_j^0$ changes by only $\mathcal{O}_{\mathcal A}(\delta)$. In order to make good use of the small preparation error, we directly compare the product environmental-state vectors $\zeta_j$ and $\zeta_j^0$ on which the rotation acts. 
  \Cref{prop:exactification} estimates the states $\rho_j$ and $\rho_j^0$ after tracing out $\mathcal Y_j$ (or $\mathcal Y_j^0$) and therefore allows any purification of $\rho_j^0$. Here we need $\zeta_j^0$ itself to be close to $\zeta_j$, so we keep the original $\mathcal Y_j$ and align the purifications within that common environment. The domination \cref{eq:rho-cutoff-2} implies $\supp\rho_j^0\subseteq\supp\rho_j$, so define
\[
\zeta_j^0=\bigl((\rho_j^0)^{1/2}\rho_j^{-1/2}\otimes I\bigr)\zeta_j,
\]
where $\rho_j^{-1/2}$ denotes the pseudo-inverse taken on $\supp\rho_j$.  Explicitly, $\rho_j^{-1/2}$ is defined to be zero on $\ker\rho_j$, so $\rho_j^{-1/2}\rho_j\rho_j^{-1/2}=P_j$, where $P_j$ projects onto $\supp\rho_j$ and thus acts as the identity on $\supp\rho_j^0$.  Since $\Tr_{\mathcal Y_j}\ket{\zeta_j}\!\bra{\zeta_j}=\rho_j$, partial trace gives
\[
\Tr_{\mathcal Y_j}\ket{\zeta_j^0}\!\bra{\zeta_j^0}
=(\rho_j^0)^{1/2}\rho_j^{-1/2}\rho_j\rho_j^{-1/2}(\rho_j^0)^{1/2}
=\rho_j^0.
\]
Thus $\zeta_j^0$ is a normalized purification.  Using the reduced state of $\zeta_j$ and cyclicity of the trace,
 $
\langle\zeta_j,\zeta_j^0\rangle
=\Tr\!\left(\rho_j(\rho_j^0)^{1/2}\rho_j^{-1/2}\right)
=\Tr\!\left(\rho_j^{1/2}(\rho_j^0)^{1/2}\right).
 $ 
The domination in \cref{eq:rho-cutoff-2} and operator monotonicity of the square root function \cite[Theorem 3.18]{carlen2025inequalities} imply $\rho_j^{1/2}\geq\sqrt{1-2n_*\delta/\tau}\,(\rho_j^0)^{1/2}$.  Taking the trace against the positive operator $(\rho_j^0)^{1/2}$ yields
\[
\langle\zeta_j,\zeta_j^0\rangle
\geq\sqrt{1-2n_*\delta/\tau}\,\Tr\rho_j^0
=\sqrt{1-2n_*\delta/\tau}.
\]
Finally, both vectors are normalized, so
\begin{equation}\label{eq:purification-distance}
\begin{aligned}
\norm{\zeta_j^0-\zeta_j}^2
 =2-2\operatorname{Re}\langle\zeta_j,\zeta_j^0\rangle \leq2(1-\sqrt{1-2n_*\delta/\tau})
\leq4n_*\delta/\tau, \quad \norm{J_j^0-J_j} \leq2\sqrt{n_*\delta/\tau}.
\end{aligned}
\end{equation}
\proofstep{Step 3: rotation on the dilations of encoded logical states.}
Reuse the vectors and orthogonal planes from \cref{lem:subspace-rotation}.  Since $U^*\mathbf{s}_\ell=\mathbf{r}_\ell$, $\langle\mathbf{s}_\ell,\mathbf{r}_\ell\rangle=\cos\theta_\ell$ and $0\leq\theta_\ell<\pi/2$,
 $
\norm{(U^*-I)\mathbf{s}_\ell}^2=\norm{\mathbf{r}_\ell-\mathbf{s}_\ell}^2=2(1-\cos\theta_\ell)\leq2\sin^2\theta_\ell.
 $ 
The planes are mutually orthogonal, so for every $v\in\mathcal{G}_0$,
\begin{equation}\label{eq:weighted-vector-rotation}
\begin{aligned}
\norm{(U^*-I)v}^2&=2\sum_\ell|\langle\mathbf{s}_\ell,v\rangle|^2(1-\cos\theta_\ell) \leq2\sum_\ell|\langle\mathbf{s}_\ell,v\rangle|^2\sin^2\theta_\ell=2\norm{(I-P)v}^2.
\end{aligned}
\end{equation}
Expanding in an orthonormal basis of $\mathcal Y_j$ gives the same inequality 
\begin{equation}\label{eq:weighted-vector-rotation-2}
\norm{((U^*-I)\otimes I_{\mathcal Y_j})v'}^2\leq2\norm{((I-P)\otimes I_{\mathcal Y_j})v'}^2
\end{equation}
for any $v'\in\mathcal{G}_0\otimes\mathcal Y_j$.  Set $W_0=UW$ and $J_j^0\xi=\xi\otimes\zeta_j^0$.  Since $\Ran J_j^0\subseteq\mathcal{G}_0\otimes\mathcal Y_j$, \cref{eq:weighted-vector-rotation-2,eq:purification-distance,eq:J-leakage} give
\begin{equation}\label{eq:weighted-rotation-error}
  \begin{aligned}
\norm{((U^*-I)\otimes I_{\mathcal Y_j})J_j^0}&\leq\sqrt2\norm{((I-P)\otimes I_{\mathcal Y_j})J_j^0} \leq\sqrt2\left(\norm{((I-P)\otimes I_{\mathcal Y_j})J_j}+\norm{J_j^0-J_j}\right)\\
&\leq\sqrt2\left(\sqrt{2n_*\delta}+2\sqrt{n_*\delta/\tau}\right)=:\epsilon_U.
\end{aligned}
\end{equation}
This is where the comparison improves. The bound measures the displacement of the dilation vectors $J_j^0\xi=\xi\otimes\zeta_j^0$ for encoded logical states, rather than that of an arbitrary vector in $\mathcal G$. 

\proofstep{Step 4: the exact factorization and ambient completion.}
Define $\Delta_0,\Upsilon_0$ by \cref{eq:Delta0,eq:Upsilon0}, now with the shared $\mathcal Y_j$. The same computation in   Equation \eqref{eq:exact-left-inverse} again gives $\Upsilon_0\Delta_0=\id_{\mathcal A}$, while by \cref{eq:Zj-Jj,eq:weighted-rotation-error,eq:purification-distance} we have
\begin{equation}\label{eq:V_j^0-V_j}
\begin{aligned}
V_j^0-V_j
&=(W^*\otimes  I_{\mathcal Y_j})\left[((U^*-I)\otimes I_{\mathcal Y_j})J_j^0+(J_j^0-J_j)+(J_j-Z_j)\right],\\
\norm{V_j^0-V_j}&\leq\epsilon_U+2\sqrt{n_*\delta/\tau}+\sqrt{2n_*\delta}=:\epsilon_V=\mathcal{O}_{\mathcal A}(\sqrt\delta).
\end{aligned}
\end{equation}
Thus \cref{lem:stinespring-continuity} gives
\begin{equation}\label{eq:weighted-recovery}
\norm{\Upsilon_0-\Upsilon}_{\mathrm{cb}}\leq2\epsilon_V.
\end{equation}
To estimate the metastable-seminorm error, use $(\pi(x)\otimes I)J_j^0=J_j^0x_j$ (\cref{eq:exact-retract-1}) and $(PU^*\otimes I)J_j^0=(U^*\otimes I)J_j^0$ (since $U^*\mathcal{G}_0\subseteq\Ran P$) to obtain
\begin{align*}
(\Delta(x)\otimes I)V_j^0-V_j^0x_j
&=(W^*\pi(x)W\otimes I)(W^*U^*\otimes I)J_j^0-(W^*U^*\otimes I)J_j^0x_j\\
&=(W^*\otimes I)([\pi(x),U^*]\otimes I)J_j^0,\\
\norm{(\Delta(x)\otimes I)V_j^0-V_j^0x_j}
&\leq\norm{(\pi(x)(U^*-I)\otimes I)J_j^0}+\norm{((U^*-I)\otimes I)J_j^0x_j}\\
&\leq\bigl(\norm{\pi(x)}+\norm{x_j}\bigr)\norm{((U^*-I)\otimes I)J_j^0} \leq2\epsilon_U\norm x.
\end{align*}
Here we used $\norm{\pi(x)}\le\norm{x},\norm{x_j}\leq\norm x$, and the rotation estimate \cref{eq:weighted-rotation-error}. 
Suppose we define the conditional expectation $E_{\rm g}:=\Delta_0\Upsilon_0$ just as in the cb-norm case \cref{sec:cb-proof}
and set $\mathcal T_{\rm g}:=E_{\rm g}-\Delta\Upsilon$. Combining this estimate with \cref{eq:weighted-recovery} gives
\[
\begin{aligned}
\norm{(\mathcal T_{\rm g}(X)\otimes I)V_j^0}
 \leq \norm{(\Upsilon_0(X)-\Upsilon(X))_j}
 +2\epsilon_U\norm{\Upsilon(X)} \leq(2\epsilon_U+2\epsilon_V)\norm X.
\end{aligned}
\]
Replacing $V_j^0$ by $V_j$ costs at most $2\epsilon_V\norm X$, because $\norm{\mathcal T_{\rm g}}_{\mathrm{cb}}\leq2$. These identities also hold after tensoring, i.e. for $\mathcal T_{{\rm g},m}=\id_{M_m}\otimes\mathcal T_{\rm g}$,
\begin{equation}\label{eq:weighted-global-retract-error}
\norm{\Upsilon_m(\mathcal T_{{\rm g},m}(X)^*\mathcal T_{{\rm g},m}(X))}^{1/2}
\leq(4\epsilon_V+2\epsilon_U)\norm X
=\mathcal O_{\mathcal A}(\sqrt\delta)\norm X.
\end{equation}
Thus the square-root estimate is already available for the $E_{\rm g}$ construction. 
However, as discussed in \cref{sec:surgery-picture}, the range $\Ran E_{\rm g}=\Delta_0(\mathcal A)$ is not necessarily a subalgebra of $M_d$ under its ordinary product. The repaired decoder acts on $x=\bigoplus_jx_j$ by $\Delta_0(x)=W_0^*\bigl[\bigoplus_j(x_j\otimes I_{\mathcal E_j})\bigr]W_0$. To isolate the part of its action coming from the cutoff representation on $\mathcal G_0$, consider the physical projection $q\in M_d$ associated with that subspace:

\begin{equation}\label{eq:physical-cutoff}
P_{\mathcal{G}_0}=\bigoplus_j(I_{n_j}\otimes q_j),
\quad q:=W_0^*P_{\mathcal{G}_0}W_0,
\quad \Ran q=W_0^*\mathcal{G}_0.
\end{equation}
The inclusion $\mathcal{G}_0\subseteq\Ran W_0=U\Ran W$ gives $q^2=W_0^*P_{\mathcal{G}_0}W_0W_0^\ast P_{\mathcal{G}_0}W_0=W_0^*P_{\mathcal{G}_0}W_0=q$ and thus $q$ is indeed a projection. Note that $\pi(x)$ acts as $\bigoplus_j x_j\otimes I_{\mathcal E_j}$ on $\mathcal G$, so it commutes with $P_{\mathcal{G}_0}$ (i.e. $
\mathcal{G}_0$ reduces $\pi(\mathcal A)$), and hence 
 $q\Delta_0(x)=W_0^*P_{\mathcal{G}_0}\pi(x)W_0=W_0^*\pi(x)P_{\mathcal{G}_0}W_0=\Delta_0(x)q$ for all $x\in\mathcal A$.
Since $V_j^0=(W_0^*\otimes I)J_j^0$ and $\Ran J_j^0\subseteq\mathcal{G}_0\otimes\mathcal Y_j$, \cref{eq:physical-cutoff} gives
 \begin{equation}\label{eq:V_j^0-cutoff}
(q\otimes I)V_j^0=V_j^0, \quad  ((I-q)\otimes I)V_j^0=0.
 \end{equation}
Write $\pi_{\mathcal{G}_0}(x):=\pi(x)|_{\mathcal{G}_0}=\bigoplus_j(x_j\otimes I_{q_j\mathcal E_j})$ for the restricted representation, and define
\[
\mathsf{L}:=W_0^*|_{\mathcal{G}_0},\quad q=\mathsf{L}\mathsf{L}^*,\quad
\mathfrak w(x):=\mathsf{L}\pi_{\mathcal{G}_0}(x)\mathsf{L}^*.
\]
The map $\mathsf{L}$ is an isometry, and each $q_j\neq0$ because \cref{lem:state-cutoff} supplies a density matrix $\rho_j^0$ supported on $q_j$, so $\mathfrak w$ is a faithful representation $\mathcal A\note{$\cong$}\to q M_d q$ with unit $q$. 
\begin{equation}\label{eq:weighted-representation}
\begin{aligned}
\mathfrak w(x)
 =\mathsf{L}\Bigl[\bigoplus_j(x_j\otimes I_{q_j\mathcal E_j})\Bigr]\mathsf{L}^*
=W_0^*\Bigl[\bigoplus_j(x_j\otimes q_j)\Bigr]W_0 =q\Delta_0(x)q=q\Delta_0(x)=\Delta_0(x)q.
\end{aligned}
\end{equation}
Using \cref{eq:weighted-representation} above together with \cref{eq:exact-left-inverse},
\[
(\Upsilon_0(\mathfrak w(x)))_j=(V_j^0)^*(q\Delta_0(x)q\otimes I)V_j^0=(V_j^0)^*(\Delta_0(x)\otimes I)V_j^0 =(\Upsilon_0(\Delta_0(x)))_j=x_j.
\]
Thus $\Upsilon_0\mathfrak w=\id_{\mathcal A}$. Now it becomes clearer how we modify $E_{\rm g}$ by writing $\Delta_0$ as the decomposition
\[
\Delta_0(x)=\mathfrak w(x)\oplus \mathfrak r(x),\quad
\mathfrak r(x):=(I-q)\Delta_0(x)(I-q).
\]
Note that $\mathfrak r(xy ) = \mathfrak r(x)\mathfrak r(y)$ does not hold in general, so $\Delta_0(\mathcal A) = \mathrm{Ran}(E_{\rm g})$ is in general not a subalgebra of $M_d$. 
 We replace this complementary action by an independent scalar component, while keeping $\mathfrak w(\Upsilon_0(X))$ on $q$.
Specifically, we choose any density matrix $\omega_\perp$ supported on $I-q$ when $I-q\neq0$, and set
\begin{equation}\label{eq:weighted-expectation}
E_{\rm m}(X)=\mathfrak w(\Upsilon_0(X))+\Tr(\omega_\perp (I-q)X(I-q))\,(I-q).
\end{equation}  
This map is UCP and fixes its range $\operatorname{Ran}(E_{\rm m}) = \mathfrak w(\mathcal A)\oplus\CC (I-q) \subset M_d$. $\operatorname{Ran}(E_{\rm m}) $ is an ambient subalgebra of $M_d$, so Schwarz equality on it implies the conditional-expectation property \cite[Theorem~3.18]{Paulsen2003}. 

The completion preserves the square-root order of the error. Indeed, $\widetilde{\mathcal T}:=E_{\rm m}-E_{\rm g}$ satisfies $\widetilde{\mathcal T}(X)=(I-q)\widetilde{\mathcal T}(X)(I-q)$ and $\norm{\widetilde{\mathcal T}}_{\mathrm{cb}}\leq2$. Writing $\widetilde{\mathcal T}_m=\id_{M_m}\otimes\widetilde{\mathcal T}$, we have at every amplification level
\begin{equation}\label{eq:weighted-completion-cost}
\norm{\Upsilon_m(\widetilde{\mathcal T}_m(X)^*\widetilde{\mathcal T}_m(X))}^{1/2}
\leq2\norm{\Upsilon(I-q)}^{1/2}\norm X.
\end{equation}
By \cref{eq:V_j^0-cutoff,eq:V_j^0-V_j},
\begin{equation}\label{eq:weighted-complement-weight}
\norm{\Upsilon(I-q)}=\max_j\norm{((I-q)\otimes I)V_j}^2\leq\epsilon_V^2.
\end{equation}
By the triangle inequality for $Y\mapsto\norm{\Upsilon_m(Y^*Y)}^{1/2}$, \cref{eq:weighted-global-retract-error,eq:weighted-completion-cost} give, for $\mathcal T=E_{\rm m}-\Delta\Upsilon=\mathcal T_{\rm g}+\widetilde{\mathcal T}$,
\begin{equation}\label{eq:weighted-factor-error}
\norm{\Upsilon_m(\mathcal T_m(X)^*\mathcal T_m(X))}^{1/2}
\leq(4\epsilon_V+2\epsilon_U+2\norm{\Upsilon(I-q)}^{1/2})\norm X
\leq(6\epsilon_V+2\epsilon_U)\norm X.
\end{equation}
Thus the scalar completion preserves the $\mathcal O_{\mathcal A}(\sqrt\delta)$ bound.
The estimate \cref{eq:weighted-factor-error} proves \cref{eq:weighted-retract-error} uniformly in $m$. Equation \eqref{eq:weighted-recovery} and \cref{eq:weighted-complement-weight} prove \cref{eq:weighted-retract-bounds}.  
Since $\tau^{-1}=4\dim\mathcal A$,
\[
\epsilon_V=(1+\sqrt2)(\sqrt2+4\sqrt{\dim\mathcal A})\sqrt{n_*\delta}.
\]
Thus one may take
 $
K_{\mathcal A}=(1+\sqrt2)^2n_*(\sqrt2+4\sqrt{\dim\mathcal A})^2.
 $ 
Indeed, $n_*,\dim\mathcal A\geq1$ imply $K_{\mathcal A}\geq(1+\sqrt2)^2(\sqrt2+4)^2>64$, and $\epsilon_U\leq\epsilon_V=\sqrt{K_{\mathcal A}\delta}$, so
\[
2\epsilon_V\leq6\epsilon_V+2\epsilon_U\leq8\sqrt{K_{\mathcal A}\delta}\leq K_{\mathcal A}\sqrt\delta,\quad \epsilon_V^2=K_{\mathcal A}\delta.
\]
These inequalities cover all three bounds.
\end{proof}

\subsection{Proof of \cref{thm:weighted-ambient}}\label{sec:weighted-proof}
\begin{proof}[Proof of \cref{thm:weighted-ambient}]
Let $c,\eta_0>0$ be the constants in \cref{thm:kitaev-factorization}, and let $d_{\mathcal A},K_{\mathcal A}$ be the constants in \cref{prop:weighted-exactification}. Just as in \cref{sec:cb-proof}, we may set
\[
d_r=\min_{\dim\mathcal A\leq r}d_{\mathcal A}>0,
\quad K_r=\max_{\dim\mathcal A\leq r}K_{\mathcal A}<\infty,
\quad
\eta_r=\frac12\min\left\{\frac14,\eta_0,\frac{d_r}{c}\right\},
\]
 For $\eta<\eta_r$, \cref{thm:kitaev-factorization} gives a $C^*$-algebra $\mathcal A$ and UCP maps $\Delta:\mathcal A\to M_d$, $\Upsilon:M_d\to\mathcal A$ with $\dim\mathcal A=\rank\mathcal P_\Phi\leq r$ and
\[
\epsilon_0:=\norm{\Delta\Upsilon-\Phi}_{\mathrm{cb}}\leq c\eta,
\quad
\delta:=\norm{\Upsilon\Delta-\id_{\mathcal A}}_{\mathrm{cb}}
\leq c\eta \le  d_{\mathcal A}.
\]
Thus \cref{prop:weighted-exactification} applies to these $\Delta,\Upsilon$.  In particular, $\mathfrak w:\mathcal A\to qM_dq$ is a faithful representation with unit $q$, and $E_{\rm m}$ is a UCP conditional expectation onto $\mathfrak w(\mathcal A)\oplus\CC(I-q)$. The bounds \cref{eq:weighted-retract-bounds,eq:weighted-retract-error} give
\begin{align*}
\norm{\Upsilon(I-q)} \leq K_r\delta,\quad \norm{\Upsilon_m(\mathcal{T}_m(X)^*\mathcal{T}_m(X))}^{1/2}
\leq K_r\sqrt\delta\,\norm X
\end{align*}
for every $m$ and $X\in M_m\otimes M_d$.
Since $E_{\rm m}$ and $\Delta\Upsilon$ are UCP, $\norm{\mathcal{T}}_{\mathrm{cb}}\leq2$.  For every amplified contraction $X$, the factorization error, complete contractivity of $\Delta$, and the preceding estimate yield
\begin{align*}
\norm{\Phi_m(\mathcal{T}_m(X)^*\mathcal{T}_m(X))} 
 \leq\norm{\Upsilon_m(\mathcal{T}_m(X)^*\mathcal{T}_m(X))}+4\epsilon_0
\leq K_r^2\delta+4\epsilon_0.
\end{align*}
Taking square roots and suprema gives $\norm{\mathcal{T}}_{\Phi,2}^{\mathrm{cb}}\leq\sqrt{K_r^2\delta+4\epsilon_0}$.  Since $E_{\rm m}-\Phi=\mathcal{T}+(\Delta\Upsilon-\Phi)$, the triangle inequality and \cref{prop:weighted-hierarchy} yield
\[
\norm{E_{\rm m}-\Phi}_{\Phi,2}^{\mathrm{cb}}
\leq\norm{\mathcal{T}}_{\Phi,2}^{\mathrm{cb}}+\norm{\Delta\Upsilon-\Phi}_{\mathrm{cb}}
\leq\sqrt{K_r^2\delta+4\epsilon_0}+\epsilon_0
\leq\left(\sqrt{c(K_r^2+4)}+c\right)\sqrt\eta,
\]
where the last step uses $\delta,\epsilon_0\leq c\eta$ and $\eta<1$.  Likewise,
\[
\norm{\Phi(I-q)}
\leq\norm{\Upsilon(I-q)}+\epsilon_0
\leq K_r\delta+\epsilon_0
\leq c(K_r+1)\eta.
\]
Set $\mathcal N_0=\mathfrak w(\mathcal A)$.  Faithfulness of $\mathfrak w$ gives $\dim\mathcal N_0=\dim\mathcal A=\rank\mathcal P_\Phi$, and both estimates hold with
\[
C_r=\max\left\{\sqrt{c(K_r^2+4)}+c,\ c(K_r+1)\right\}.
\] 
which is independent of $d$. 
\end{proof}

\section{QMSMs for weakly driven dissipative spin chains}\label{sec:spin-chains}
 
We use two variants of a weakly driven spin chain to test the global cb-norm and metastable-seminorm results for constructing QMSMs. 
Each example reduces a dissipative chain of \(n\) spins to either a metastable logical qubit or long-lived classical phases. The many-body dissipative dynamics is given by a continuous-time quantum Markov semigroup, and the reduced dynamics is obtained by iterating the transition channel in \cref{thm:qmsm-master}
(or the UCP transition map, by duality, see \cref{cor:reduced-transition}) at the metastable timescale $t_0$. In each variant (quantum or classical), the slow algebra $\mathcal N_\hashtag=\operatorname{Ran}E_\hashtag$ ($\hashtag\in \{\mathrm{q},\mathrm{c}\}$) is an ambient subalgebra of $M_{2^n}$, with a constant dimension independent of $n$. 
 The explicit conditional expectation $E_\hashtag$ therefore realizes the conclusions of both general theorems, while the Riesz range need not itself be an algebra and fails to capture the essential structure of the metastable information. 

In this example, the slow algebra $\mathcal N_\hashtag$ is exactly invariant when the weak drive is turned off, and we use this algebra to construct the reduced state space in the presence of the weak drive. This is inspired by the approach of \cite{GrigolettoSarletteTicozziViola2026} and makes it practical to evaluate $\mathcal K_\hashtag$ and verify the approximation bounds. However, the reduced UCP map and reduced dynamics are defined directly through the nonperturbative construction
 $
\mathcal K_\hashtag=E_\hashtag\Phi_\hashtag\iota_\hashtag,
 $ 
and do not rely on the reference decoherence-free manifold as in \cite{GrigolettoSarletteTicozziViola2026}. Figure~\ref{fig:spin-chain-models} summarizes the two models and their reduced channels. We first present the logical-qubit case, followed by a variant exhibiting long-lived classical phases. All detailed calculations and estimates are deferred to Appendix~\ref{app:spin-proofs}.
\begin{figure}[t]
  \centering
  \includegraphics[width=\textwidth]{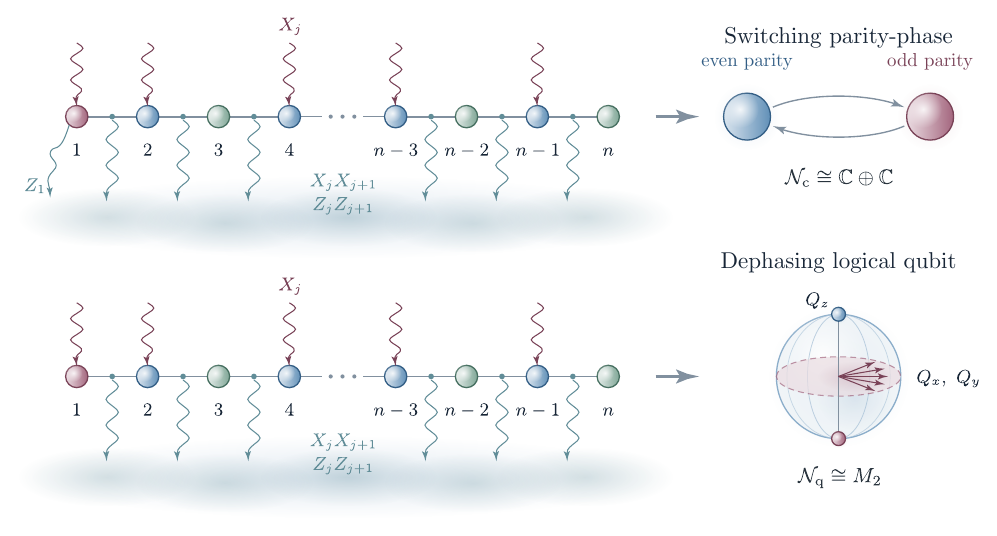}
  \captionsetup{font=small,justification=raggedright,singlelinecheck=false}
  \caption{
  A schematic illustration of the weakly driven dissipative spin chain 
 with classical or quantum metastable information. Both open chains have odd $n$, weak $X_j$ drive on $j\in\{1,2,4,\ldots,n-1\}$, and Lindbladian jumps $X_jX_{j+1}$ and $Z_jZ_{j+1}$ on every bond. Top: adding the boundary dissipator $Z_1$ gives switching between the parity phases $\rho_\pm=(I\pm Q_x)/2^n$, where $Q_x=Z^{\otimes n}$. Bottom: without this extra jump, the dynamics reduces to a logical qubit. $Q_z=X^{\otimes n}$ is conserved while $Q_x,Q_y$ dephase.}
  \label{fig:spin-chain-models}
\end{figure}

\hypertarget{sec:spin-2}{\subsection{The qubit model}\label{sec:spin-2}}

Let \(n=2n_{\mathrm b}+3\geq5\) be odd, where \(n_{\mathrm b}\) counts the driven bulk sites \(4,6,\ldots,n-1\). We suppress dependence on \(n\) in the notation and measure time in units of the common dissipative rate. With single-site Pauli operators \(X_j,Y_j,Z_j\), set
\begin{equation}\label{eq:drivingfield}
\mathcal I= \{1\} \cup \{2k\}_{k=1}^{n_{\mathrm b}+1} = \{1,2,4,6,\ldots,n-1\},\quad
H=\varepsilon\sum_{j\in\mathcal I}X_j.
\end{equation}
For a jump operator $L$, write the Heisenberg-picture dissipator as
 $
\mathcal D_L(A)=L^*AL-\tfrac12\{L^*L,A\}.
 $ 
The generator $\mathcal L_{\mathrm q}$ and its semigroup are
\[
\mathcal L_{\mathrm q}
=\mathrm{i}[H,\,\cdot\,]
+\sum_{j=1}^{n-1}\bigl(\mathcal D_{X_jX_{j+1}}+\mathcal D_{Z_jZ_{j+1}}\bigr),
\quad T_t^{(\mathrm q)}=e^{t\mathcal L_{\mathrm q}}.
\]
For these Hermitian Pauli jumps, $\mathcal D_P(A)=PAP-A$.
The dissipative jumps are \(X_jX_{j+1}\) and \(Z_jZ_{j+1}\). The sites outside \(\mathcal I\) are undriven, but remain coupled through these jumps. Choose
\begin{equation}\label{eq:qubit-parameters}
\varepsilon=n^{-2},\quad t_0=C_t\log n,\quad C_t\geq8,\quad
\Phi_{\mathrm q}=T_{t_0}^{(\mathrm q)}.
\end{equation}

\hypertarget{sec:spin-3}{\subsection{The reference qubit algebra and reconstruction}\label{sec:spin-3}}

For a Hermitian Pauli unitary \(P\), define
\begin{equation}\label{eq:qubit-expectation}
\operatorname{Ad}_P(A)=PAP,\quad 
\mathcal E_P(A)=\tfrac12(A+PAP),\quad
E_{\mathrm q}=\prod_{j=1}^{n-1}\mathcal E_{X_jX_{j+1}}\mathcal E_{Z_jZ_{j+1}}.
\end{equation}
We define 
\[
Q_x=Z^{\otimes n},\quad Q_z=X^{\otimes n},\quad Q_y=-\mathrm{i}Q_zQ_x,
\]
We choose the logical labels so that the conserved string \(X^{\otimes n}\) is \(Q_z\). Consequently, the reduced UCP map below has the standard \(Z\)-dephasing form, fixing \(Q_z\) and damping \(Q_x,Q_y\).
Odd \(n\) gives \(Q_xQ_z=-Q_zQ_x\), so the linear span of \(\{I,Q_x,Q_y,Q_z\}\) is also an ambient matrix algebra, i.e. 
\[
\mathcal N_{\rm q} = \operatorname{span}_{\mathbb C}\{I,Q_x,Q_y,Q_z\}\cong M_2.
\]
We note that $\mathcal E_P$ is the conditional expectation onto the commutant of $P$ by computing the action on the Pauli words.  
 With the normalized trace \(\tau=2^{-n}\operatorname{Tr}\), $E_{\mathrm q}$ is the conditional expectation onto \(\mathcal N_{\rm q}\)
 \begin{equation}\label{eq:condexpect_qubit}
E_{\mathrm q}(A)=\tau(A)I+\sum_{j=x,y,z}\tau(Q_jA)Q_j.
\end{equation}
Let \(\iota_{\mathrm q}:M_2\to M_{2^n}\) send \(\sigma_j\) to \(Q_j\) for $j =x,y,z$. We also write \(E_{\mathrm q}\) for its corestriction to the logical algebra, so \(E_{\mathrm q}\iota_{\mathrm q}=\mathrm{id}\). The reconstruction map which is the dual of \(E_{\mathrm q}\) is
\[
E_{\mathrm q,\ast}\!\left(\tfrac12(I+v_x\sigma_x+v_y\sigma_y+v_z\sigma_z)\right)
=\frac{I+v_xQ_x+v_yQ_y+v_zQ_z}{2^n}.
\]
In the zero-field case i.e. the driving field $\varepsilon = 0$ in \cref{eq:drivingfield}, the subalgebra $\mathcal N_{\rm q}$ is exactly preserved by the dissipative dynamics. And the metastability behavior arises when the driving field is weakly turned on. 
With the parameters \cref{eq:qubit-parameters}, we have the following quantitative estimates for the metastability of $e^{\mathcal L_{\mathrm q}}$.
 \begin{prop}\label{prop:spin-qubit}
For all sufficiently large odd \(n\), the zero-field expectation \(E_{\mathrm q}\) defined above satisfies
\[
\begin{aligned}
\operatorname{rank}\mathcal P_{\Phi_{\mathrm q}}&=4,\quad
\eta_{\mathrm q}:=\|\Phi_{\mathrm q}^2-\Phi_{\mathrm q}\|_{\mathrm{cb}}
\sim\frac{\log n}{n^3},\\
\delta_{\mathrm q}:=\|E_{\mathrm q}-\Phi_{\mathrm q}\|_{\mathrm{cb}}
&\sim n^{-1},\quad
\epsilon_{\mathrm q}:=\|E_{\mathrm q}-\Phi_{\mathrm q}\|_{\Phi_{\mathrm q},2}^{\mathrm{cb}}
=\mathcal{O}(n^{-3/2}).
\end{aligned}
\]
All constants are independent of \(n\) and reference-system dimension.
\end{prop}
 Thus the assumptions of \cref{thm:main,thm:weighted-ambient} hold with fixed slow rank, and the same explicit expectation realizes both guarantees, with \(q=I\). Here, we construct this expectation directly and do not identify it with the particular idempotent selected by either existence proof. The derivation is summarized in \cref{sec:spin-8} and proved in Appendix~\ref{app:spin-spectral} and~\ref{app:spin-weighted}.

\hypertarget{sec:spin-5}{\subsection{The qubit transition map}\label{sec:spin-5}}
At the metastable timescale $t_0$, define the compressed transition map and its dual
\[
\mathcal K_{\mathrm q}=E_{\mathrm q}\Phi_{\mathrm q}\iota_{\mathrm q},\quad \mathcal M_{\rm q} = 
\mathcal K_{\mathrm q,\ast}=\iota_{\mathrm q,\ast}\Phi_{\mathrm q,\ast}E_{\mathrm q,\ast}.
\]
For subsequent state-picture formulas here and in Appendix~\ref{app:spin-proofs}, we use the notation of the QMSM master theorem, \cref{thm:qmsm-master}. We use the same notation $\mathcal M_\hashtag$ in logical-register coordinates under the CPTP identification. This identification allows us to write $\mathcal M_{\mathrm q}$ directly on qubit density matrices.
The UCP map $\mathcal K_{\mathrm q}$ fixes \(I,Q_z\) and multiplies \(Q_x,Q_y\) by \(u=u_{\mathrm q}(t_0)\). Appendix~\ref{app:spin-local} computes this scalar using only a 4-by-4 and a 2-by-2 matrix exponential $
u_{\rm q}(t)=(e^{tM_{\rm q}})_{00}[(e^{tM_{\mathrm s}})_{00}]^{n_{\mathrm b}}.
 $ 
The equality of the two multipliers follows from a fixed symmetry. The global X-parity is a fixed unitary and therefore lies in the multiplicative domain \cite[Theorem~3.18]{Paulsen2003}:
\[
T_t^{(\mathrm q)}(Q_z)=Q_z,\quad Q_y=-\mathrm{i}Q_zQ_x,\quad
T_t^{(\mathrm q)}(Q_y)=-\mathrm{i}Q_zT_t^{(\mathrm q)}(Q_x).
\]
The expectation is a bimodule map over the subalgebra $\mathcal N_{\rm q}$, so
\[
E_{\mathrm q}T_t^{(\mathrm q)}(Q_x)=uQ_x
\quad\Longrightarrow\quad
E_{\mathrm q}T_t^{(\mathrm q)}(Q_y)=-\mathrm{i}Q_z(uQ_x)=uQ_y.
\]
For a logical density matrix,
\[
\mathcal M_{\mathrm q}(\rho)=\frac{1+u}{2}\rho+\frac{1-u}{2}\sigma_z\rho\sigma_z, \quad 
\mathcal M_{\mathrm q}^{\,k}
\begin{pmatrix}a&c\\\overline c&1-a\end{pmatrix}
=\begin{pmatrix}a&u^kc\\u^k\overline c&1-a\end{pmatrix} .
\]Physically, this has a standard dephasing-channel form. 
 In the reduced qubit, the logical populations are conserved, while the coherence decays. For sufficiently large $n$, $0<u<1$, and \(-\log u\sim\gamma_{\mathrm q}t_0\sim C_t(\log n)/n^3\), so the decay rate per unit metastable timescale satisfies \(\kappa_{\mathrm q}=-t_0^{-1}\log u\sim\gamma_{\mathrm q}\sim n^{-3}\).

In Appendix~\ref{app:spin-local}, we show that the Riesz projection $\mathcal P_{\rm q}$ associated with the semigroup $T_t^{(\rm q)}$ is not positive and that its range is not an algebra, by examining the action of $\mathcal P_{\rm q}$ on the eigenoperators.
\hypertarget{sec:spin-6}{\subsection{Full-state prediction from the QMSM process}\label{sec:spin-6}}

\Cref{cor:reduced-transition} applies with \(E=E_{\mathrm q}\) and \(\Phi=\Phi_{\mathrm q}\). Substitution of the errors in \cref{sec:spin-2} gives
\[
\begin{aligned}
\|\mathcal R_{\mathrm q}\mathcal M_{\mathrm q}^k\mathcal Q_{\mathrm q}-\mathcal C_{\mathrm q}^k\|_\diamond
&\leq2k\delta_{\mathrm q}=\mathcal{O}(k/n),\quad k\geq1,\\
\|\mathcal R_{\mathrm q}\mathcal M_{\mathrm q}^k\mathcal Q_{\mathrm q}\mathcal C_{\mathrm q}-\mathcal C_{\mathrm q}^{k+1}\|_\diamond
&\leq(k+1)(\epsilon_{\mathrm q}+\eta_{\mathrm q})
=\mathcal{O}((k+1)n^{-3/2}),\quad k\geq0.
\end{aligned}
\]
The second estimate is weighted by the preparation $\mathcal C_{\rm q}$ of the metastable states. We note that the explicit computation detailed in the Appendix gives a separate longer-time guarantee: for metastable states, the error is \(\mathcal{O}(n\varepsilon+(k+1)n\varepsilon^2+n e^{-t_0/2})\). For \(k\gamma_{\rm q}t_0\leq S\) with fixed \(S\), this is \(\mathcal{O}_S(n^{-1}+1/\log n)\). That is, it remains small throughout a fixed number of iterations of the reduced process instead of growing linearly with $k$ as in the estimate given by \cref{cor:reduced-transition}, using some model-specific information (Appendix~\ref{app:spin-lifetime}).

\hypertarget{sec:spin-7}{\subsection{A classical variation: two parity phases}\label{sec:spin-7}}

We also provide a dissipative spin-chain model, whose reduced slow dynamics is classical, i.e. the slow algebra is commutative. The model is the same as above, except that we 
add one boundary jump \(Z_1\):
\begin{equation}\label{eq:spin-classical-generator}
\begin{aligned}
\mathcal L_{\mathrm c}(A) =\mathcal L_{\rm q}(A)+\mathcal D_{Z_1}(A)
=\mathcal L_{\mathrm q}(A)+(Z_1AZ_1-A),\quad 
E_{\mathrm c} =\mathcal E_{Z_1}E_{\mathrm q}, \quad 
\mathcal N_{\mathrm c} =\operatorname{span}\{I,Q_x\}
\cong\mathbb C\oplus\mathbb C.
\end{aligned}
\end{equation}
Indeed,
\[
\begin{aligned}
Z_1Q_xZ_1&=Q_x,\quad
Z_1Q_yZ_1=-Q_y,\quad
Z_1Q_zZ_1=-Q_z,\\
\mathcal E_{Z_1}(Q_x)&=Q_x,\quad
\mathcal E_{Z_1}(Q_y)=\mathcal E_{Z_1}(Q_z)=0.
\end{aligned}
\]
Thus the extra dephasing removes \(Q_y,Q_z\). With \(P_\pm=(I\pm Q_x)/2\),
\begin{equation} \label{eq:condexpect_classical}
E_{\mathrm c}(A)=\sum_{s\in\{+,-\}}\frac{\operatorname{Tr}(P_sA)}{2^{n-1}}P_s
=\tau(A)I+\tau(Q_xA)Q_x,\quad
\rho_\pm=\frac{I\pm Q_x}{2^n}.
\end{equation}
This is the conditional expectation averaging within the two parity blocks. Writing \(u_{\mathrm c}\) for \(u_{\mathrm c}(t_0)\), the (classical) transition matrix and its iterations are
\[
K_{\mathrm c}^k=\frac12\begin{pmatrix}1+u_{\mathrm c}^k&1-u_{\mathrm c}^k\\1-u_{\mathrm c}^k&1+u_{\mathrm c}^k\end{pmatrix},\quad
p^{(k)}=K_{\mathrm c}^kp^{(0)},\quad  p  = \mqty(p_+\\p_-).
\]
Equivalently, for the ambient representative \(\rho=(I+mQ_x)/2^n\),
\[
\mathcal M_{\mathrm c}^{\,k}(\rho)=u_{\mathrm c}^k\rho+(1-u_{\mathrm c}^k)I/2^n.
\]
 The one-step switching probability is \((1-u_{\mathrm c})/2\). The relaxation rate is \(\kappa_{\mathrm c}=-t_0^{-1}\log u_{\mathrm c}\). The slow rank is two, \(\gamma_{\mathrm c}=(n_{\mathrm b}+2)\varepsilon^2+\mathcal{O}(n\varepsilon^4)\), and the same upper error bounds hold. The Riesz projection for $T_t^{\rm (c)}$ is also not positive, and the 
 range is also not an algebra (Appendix~\ref{app:spin-local}). 

\hypertarget{sec:spin-8}{\subsection{Proof strategy and key estimates}\label{sec:spin-8}}

We provide the main proof strategy and key estimates for the analysis of the weakly driven dissipative spin chain example, and the detailed derivations are given in Appendix \ref{app:spin-proofs}. 
The proof first identifies the slow spectral projection and controls the fast remainder. It then compares the explicit conditional expectations \(E_{\mathrm q},E_{\mathrm c}\) with the metastable propagators \(\Phi_{\mathrm q},\Phi_{\mathrm c}\) in the cb norm and metastable seminorm. For \(\hashtag\in\{\mathrm c,\mathrm q\}\), the resolvent argument gives
\[
\|\mathcal P_{\hashtag}-E_\hashtag\|_{\mathrm{cb}}\lesssim n\varepsilon,\quad
T_t^{(\hashtag)}=\mathcal F_\hashtag+e^{-\gamma_{\hashtag}t}\mathcal{S}_{\hashtag}+\mathcal Z_{\hashtag,t},\quad
\|\mathcal Z_{\hashtag,t}\|_{\mathrm{cb}}\lesssim n e^{-t/2}.
\] Appendix~\ref{app:spin-spectral} proves these assertions. The additional key estimate proved in Appendix~\ref{app:spin-weighted} is
 $
\|E_\hashtag-\mathcal P_{\hashtag}\|_{\Phi_{\hashtag},2}^{\mathrm{cb}}
\leq C(\sqrt n\varepsilon+e^{-t_0}).
 $ 
Using the triangle inequality 
and \(\|S\|_{\Phi,2}^{\mathrm{cb}}\leq\|S\|_{\mathrm{cb}}\), we have 
\[
\begin{aligned}
\delta_{\hashtag}
 \lesssim(n\varepsilon+\gamma_{\hashtag}t_0+n e^{-t_0/2})
=\mathcal{O}(n^{-1}),\quad 
\epsilon_{\hashtag}
 \lesssim  C(\sqrt n\varepsilon+\gamma_{\hashtag}t_0+n e^{-t_0/2})
=\mathcal{O}(n^{-3/2}).
\end{aligned}
\]
The decomposition also bounds \(\eta_{\hashtag}\) above by \(C\gamma_{\hashtag}t_0+Cn e^{-t_0/2}\). Testing a slow eigenoperator bounds it below by \(e^{-\gamma_{\hashtag}t_0}(1-e^{-\gamma_{\hashtag}t_0})\). Hence \(\eta_{\hashtag}\sim(\log n)/n^3\). So both theorem assumptions hold and both norm estimates are uniform in the many-body dimension $n$.

\newpage
\appendix
\section{Proof of Proposition~\ref{prop:riesz}}
\label{app:riesz-proof}

\begin{proof}[Proof of Proposition~\ref{prop:riesz}]
Set
 $
  \eta:=\norm{\Phi^2-\Phi}_{\mathrm{cb}}<\frac14.
 $ 
By spectral mapping, every $\lambda\in\spec(\Phi)$ satisfies
\[
  |\lambda(\lambda-1)|
  \leq \norm{\Phi^2-\Phi}_{\mathrm{cb}}
  =\eta<\frac14.
\]
Hence $\spec(\Phi)=\Sigma_0\sqcup\Sigma_1$ splits into two components near $0$ and $1$, respectively. See \cref{fig:spectral-region}.

\begin{figure}[h]
  \centering
  \includegraphics[width=0.5\textwidth]{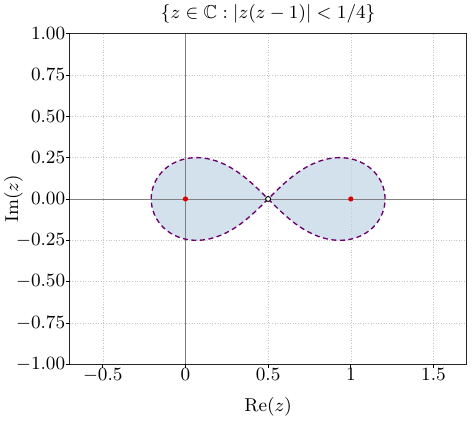}
  \caption{$\{|z(z-1)|<1/4\}$ has two components containing $0$ and $1$ respectively.}
  \label{fig:spectral-region}
\end{figure}

On the spectral region $\{|z(z-1)|<1/4\}$, define
\[
  \phi(z)
  =\frac12\left[
  1+(2z-1)\bigl(1+4(z^2-z)\bigr)^{-1/2}
  \right],
\]
where $(\cdot)^{-1/2}$ denotes the holomorphic branch defined by the binomial series near $1$. Since
 $
  (2z-1)^2=1+4(z^2-z),
 $ 
the function $\phi$ equals $0$ on $\Sigma_0$ and $1$ on $\Sigma_1$. Therefore, by the holomorphic functional calculus,
\[
  \mathcal P_\Phi:=\phi(\Phi)
  =\frac{1}{2\pi\mathrm{i}}
  \oint_{\mathcal C}(z\id-\Phi)^{-1}\,\mathrm dz,
\]
for any positively oriented contour $\mathcal C$ enclosing $\Sigma_1$ but not $\Sigma_0$ \cite{Kato1995}. Thus $\mathcal P_\Phi$ is precisely the Riesz projection onto the spectral component near $1$.
Using
 $
  (2\Phi-\id)^2=\id+4(\Phi^2-\Phi),
 $ 
and the fact that the two factors commute, we may write
 $
  \mathscr J_\Phi
  :=(2\Phi-\id)
  \bigl(\id+4(\Phi^2-\Phi)\bigr)^{-1/2},
 $ 
so that $\mathscr J_\Phi^2=\id$ and
 $
  \mathcal P_\Phi=\frac12(\id+\mathscr J_\Phi).
 $ 
In particular, $\mathcal P_\Phi^2=\mathcal P_\Phi$.

It remains to estimate its distance from $\Phi$. By
\cite[Proposition~3.6]{Paulsen2003},
 $
  \norm{\Phi}_{\mathrm{cb}}
  =\norm{\Phi(I)}=1,
 $ 
and hence $\norm{2\Phi-\id}_{\mathrm{cb}}\leq3$. Expanding the inverse square root at the identity,
 $
  \bigl(\id+4(\Phi^2-\Phi)\bigr)^{-1/2}
  =
  \id-2(\Phi^2-\Phi)+\mathcal O(\eta^2),
 $ 
we obtain
\begin{align*}
\norm{\mathcal P_\Phi-\Phi}_{\mathrm{cb}}
&=
\frac12
\norm{(2\Phi-\id)
\left[
\bigl(\id+4(\Phi^2-\Phi)\bigr)^{-1/2}-\id
\right]}_{\mathrm{cb}}\\
&\leq
3\norm{\Phi^2-\Phi}_{\mathrm{cb}}
+\mathcal O(\eta^2)
=\mathcal O(\eta).
\end{align*}
\end{proof}
\hypertarget{app:spin-proofs}{\section{Technical proofs for the spin-chain examples}\label{app:spin-proofs}}

The appendix has two purposes. Appendix~\ref{app:spin-local} derives in detail the reduced qubit channel and the classical transition matrix directly from the coefficient matrices \(M_{\mathrm q},M_{\mathrm c},M_{\mathrm s}\). Appendix~\ref{app:spin-spectral} 
 verifies the assumptions of the main theorems and bounds the difference between the conditional expectation and the metastable propagator, as well as the error for the iterations of the QMSM reduced dynamics and the full semigroup evolution. 

 We write 
 the normalized trace 
 \(\tau=2^{-n}\operatorname{Tr}\), where $n$ is the chain length. 

\hypertarget{app:spin-local}{\subsection{Local blocks and the reduced transition maps}\label{app:spin-local}}

We first separate the full generator into pieces. Put \(\mathcal J= \{2k\}_{k=2}^{n_{\rm b}+1} =  \{4,6,\ldots,n-1\}\) and
\[
\begin{aligned}
\mathcal D_X =\sum_{\ell=1}^{n-1}
\mathcal D_{X_\ell X_{\ell+1}},\quad \mathcal L_{0,\mathrm q}
 =\mathrm{i}[\varepsilon(X_1+X_2),\,\cdot\,]
+\sum_{\ell=1}^2
\mathcal D_{Z_\ell Z_{\ell+1}},\quad 
\mathcal L_{0,\mathrm c}
 =\mathcal L_{0,\mathrm q}+ \mathcal D_{Z_1}.
\end{aligned}
\]
For \(j\in\mathcal J\), define the local generator
 $
\mathcal L_j
=\mathrm{i}[\varepsilon X_j,\,\cdot\,]+
\mathcal D_{Z_{j-1}Z_{j}} + \mathcal D_{Z_{j}Z_{j+1}}.
 $ 
The full generator is the sum of the boundary and bulk-site pieces:
\[
\mathcal L_\hashtag
=\mathcal D_X+\mathcal L_{0,\hashtag}
+\sum_{j\in\mathcal J}\mathcal L_j,
\quad \hashtag\in\{\mathrm c,\mathrm q\}.
\]
Let \(\mathcal I=\{1,2\}\cup\mathcal J\) which is the set of sites where the weak driving field acts. For \(\mathbf s=(s_j)_{j\in\mathcal I}\in\{0,1\}^{\mathcal I}\), define the following Pauli-word observables and the subspace they span:
\[
P_{\mathbf s}
:=
\left(\prod_{j\in\mathcal I}Z_j^{\,1-s_j}Y_j^{\,s_j}\right)
\left(\prod_{k\in\mathcal I^c}Z_k\right),
\quad
\mathcal V_x
:=
\operatorname{span}_{\mathbb C}\{P_{\mathbf s}\}.
\]
An X-field only interchanges \(Z\) and \(Y\) at its driven site, and every dephasing term acts diagonally on Pauli words. Hence \(\mathcal V_x\) is invariant under every summand in $\mathcal L_{\hashtag}$. Moreover, every \(P_{\mathbf s}\) commutes with every \(X_\ell X_{\ell+1}\), so
 $ 
\left.\mathcal D_X\right|_{\mathcal V_x}=0.
 $
Technically, this exact invariance is what allows the reduced transition parameter to be computed from the following blocks of coefficient matrices.

\hypertarget{the-coefficient-matrices}{\subsubsection{The coefficient matrices}\label{the-coefficient-matrices}}

We first identify the tensor-product structure of $\mathcal V_x$. We recall that $\mathcal I = \{1,2\}\cup \mathcal J$. 
The two boundary sites $\{1,2\}$ contribute the four-dimensional local space
 $
\operatorname{span}\{Z_1Z_2, Y_1Z_2, Z_1Y_2, Y_1Y_2\}\simeq \mathbb C^4,
 $ 
while each $j\in\mathcal J$ contributes the two-dimensional space
$\operatorname{span}\{Z_j,Y_j\}\simeq\mathbb C^2$.
Hence we may identify 
\[
\mathcal V_x\simeq \mathbb C^4\otimes(\mathbb C^2)^{\otimes n_{\mathrm b}},
\quad
\dim\mathcal V_x=4\cdot 2^{n_{\mathrm b}}.
\]
We now compute the actions of the local generators. 
For $j\in\mathcal J$, let $P_{\mathbf s}$ be a Pauli word and let
$\mathbf s^{(j\to \rm b)}$ denote $\mathbf s$ with $s_j$ replaced by $\rm b \in \{0,1\}$.
Using $\mathrm{i}[X,Z]=2Y$, $\mathrm{i}[X,Y]=-2Z$, one finds
\[
\mathcal L_j(P_{\mathbf s})  
=
\begin{cases}
2\varepsilon P_{\mathbf s^{(j\to1)}}, & s_j=0,\\[2mm]
-2\varepsilon P_{\mathbf s^{(j\to0)}}-4P_{\mathbf s}, & s_j=1.
\end{cases}
\]
Thus, in the ordered basis $(Z_j,Y_j)$, the local generator $\mathcal L_j$ is represented by the matrix 
 $
 M_{\mathrm s}=
\begin{pmatrix}
0&-2\varepsilon\\
2\varepsilon&-4
\end{pmatrix}.
 $ 
For the boundary sites, let \((B_0,B_1,B_2,B_{12})\) denote the Pauli words with 
 ordered basis 
$ (Z_1Z_2, Y_1Z_2, Z_1Y_2, Y_1Y_2) $ 
on the first two sites, with all spectator coordinates held fixed. Directly,
\[
\begin{aligned}
\mathcal L_{0,\mathrm q}(B_0)&=2\varepsilon B_1+2\varepsilon B_2,\quad 
&\mathcal L_{0,\mathrm q}(B_1) =-2\varepsilon B_0-2B_1+2\varepsilon B_{12},\\
\mathcal L_{0,\mathrm q}(B_2)&=-2\varepsilon B_0-4B_2+2\varepsilon B_{12},\quad &
\mathcal L_{0,\mathrm q}(B_{12})=-2\varepsilon B_1-2\varepsilon B_2-2B_{12},
\end{aligned}
\]
and therefore $\mathcal L_{0,\rm q}$ acts on the local basis $(B_0,B_1,B_2,B_{12})$ as 
\[
M_{\mathrm q}
=
\begin{pmatrix}
0&-2\varepsilon&-2\varepsilon&0\\
2\varepsilon&-2&0&-2\varepsilon\\
2\varepsilon&0&-4&-2\varepsilon\\
0&2\varepsilon&2\varepsilon&-2
\end{pmatrix}.
\]
For the classical boundary generator,
 $
\mathcal L_{0,\mathrm c}(A)
=
\mathcal L_{0,\mathrm q}(A)+(Z_1AZ_1-A).
 $ 
The additional term acts as
 $
(B_0,B_1,B_2,B_{12})
\mapsto
(0,-2B_1,0,-2B_{12}),
 $ 
so
\[
M_{\mathrm c}
=
M_{\mathrm q}-2\operatorname{diag}(0,1,0,1)
=
\begin{pmatrix}
0&-2\varepsilon&-2\varepsilon&0\\
2\varepsilon&-4&0&-2\varepsilon\\
2\varepsilon&0&-4&-2\varepsilon\\
0&2\varepsilon&2\varepsilon&-4
\end{pmatrix}.
\]
We note that each local rule is independent of the spectator coordinates of the Pauli words, so the
action of the full generator on the space $\mathcal V_x$ 
 has the corresponding Kronecker-factorized decomposition:
\begin{equation}\label{eq:spin-coefficient-kronecker-sum}
\left.\mathcal L_{\hashtag}\right|_{\mathcal V_x}
\simeq
M_{\hashtag}\otimes I_2^{\otimes n_{\mathrm b}}
+
\sum_{\ell=1}^{n_{\mathrm b}}
I_4\otimes
I_2^{\otimes(\ell-1)}
\otimes M_{\mathrm s}\otimes
I_2^{\otimes(n_{\mathrm b}-\ell)},
\quad
\hashtag\in\{\mathrm q,\mathrm c\}.
\end{equation}

\hypertarget{from-the-blocks-to-the-reduced-transition-maps}{\subsubsection{The reduced transition maps}\label{from-the-blocks-to-the-reduced-transition-maps}}

For \(\hashtag\in\{\mathrm q,\mathrm c\}\), let \(T_t^{(\hashtag)}=e^{t\mathcal L_\hashtag}\) be the Heisenberg semigroup. The Kronecker sum in \cref{eq:spin-coefficient-kronecker-sum} exponentiates factorwise:
\[
\left.T_t^{(\hashtag)}\right|_{\mathcal V_x}
\simeq e^{tM_\hashtag}\otimes(e^{tM_{\mathrm s}})^{\otimes n_{\mathrm b}}.
\]
In the ordered local bases of $\mathcal V_x$ above, \(Q_x = Z^{\otimes n }\) corresponds to the zeroth basis vector in every tensor component, and the conditional expectation \(E_\hashtag\) (see \cref{eq:condexpect_classical,eq:condexpect_qubit}) keeps the coefficient. Thus
\[
E_\hashtag T_t^{(\hashtag)}(Q_x)=u_\hashtag(t)Q_x,\quad
u_\hashtag(t)=f_\hashtag(t)f_{\mathrm s}(t)^{n_{\mathrm b}},
\]
where \(f_\hashtag(t)=(e^{tM_\hashtag})_{00}\) and \(f_{\mathrm s}(t)=(e^{tM_{\mathrm s}})_{00}\). The single-site factor is
\[
f_{\mathrm s}(t)=e^{-2t}\left[\cosh\!\left(2t\sqrt{1-\varepsilon^2}\right)
+\frac{\sinh\!\left(2t\sqrt{1-\varepsilon^2}\right)}{\sqrt{1-\varepsilon^2}}\right].
\]
The remaining factor \(f_\hashtag(t)\) is a \(4\times4\) matrix exponential, so \(M_\hashtag\) and \(M_{\mathrm s}\) determine the exact compressed transition coefficient.

For the qubit model, \(Q_z=X^{\otimes n}\) commutes with the Hamiltonian and every jump operator. Therefore
\[
T_t^{(\mathrm q)}(I)=I,\quad T_t^{(\mathrm q)}(Q_z)=Q_z.
\]
The fixed unitary \(Q_z\) lies in the multiplicative domain of the UCP map \(T_t^{(\mathrm q)}\) \cite[Theorem~3.18]{Paulsen2003}. Since \(Q_y=-\mathrm{i}Q_zQ_x\) and \(E_{\mathrm q}\) is an \(\mathcal N_{\mathrm q}\)-bimodule map,
\[
T_t^{(\mathrm q)}(Q_y)=-\mathrm{i}Q_zT_t^{(\mathrm q)}(Q_x),\quad
E_{\mathrm q}T_t^{(\mathrm q)}(Q_y)
=-\mathrm{i}Q_zE_{\mathrm q}T_t^{(\mathrm q)}(Q_x)
=u_{\mathrm q}(t)Q_y.
\]
Hence the compressed evolution $\mathcal K_{\rm q}  = E_{\mathrm q}T_t^{(\mathrm q)}\iota_{\rm q}$  
fixes \(I,\sigma_z\) and multiplies \(\sigma_x,\sigma_y\) by \(u_{\mathrm q}(t)\). At \(t=t_0\), its Schr\"odinger-picture qubit channel is
\[
\mathcal M_{\mathrm q}(\rho)
=\frac{1+u_{\mathrm q}(t_0)}2\rho
+\frac{1-u_{\mathrm q}(t_0)}2\sigma_z\rho\sigma_z.
\]

Similarly, for the classical model, let \(P_\pm=(I\pm Q_x)/2\). Since \(I\) is fixed and \(Q_x\) is multiplied by \(u_{\mathrm c}(t_0)\), the transition matrix in the ordered basis \((P_+,P_-)\) and its powers are
\[
K_{\mathrm c}=\frac12
\begin{pmatrix}
1+u_{\mathrm c}(t_0)&1-u_{\mathrm c}(t_0)\\
1-u_{\mathrm c}(t_0)&1+u_{\mathrm c}(t_0)
\end{pmatrix},\quad
K_{\mathrm c}^k=\frac12
\begin{pmatrix}
1+u_{\mathrm c}(t_0)^k&1-u_{\mathrm c}(t_0)^k\\
1-u_{\mathrm c}(t_0)^k&1+u_{\mathrm c}(t_0)^k
\end{pmatrix}.
\]

\hypertarget{slow-eigenoperators-and-non-positivity}{\subsubsection{Slow eigenoperators}\label{slow-eigenoperators-and-non-positivity}}

\paragraph{\textbf{Slow eigenvalues and eigenoperators.}} 

We now compute the slow eigenoperators of the full generator. Before that, we need the local eigenvalues and eigenvectors. 
Let $-\gamma_{0,\hashtag}$ be the simple eigenvalue of $M_\hashtag$ near zero,
$\hashtag\in\{\mathrm q,\mathrm c\}$, and let $-\gamma_{\mathrm s}$ be the slow
eigenvalue of $M_{\mathrm s}$.  Substituting $ -\gamma$ into the
characteristic equations gives
\[
\gamma_{\mathrm s}^2-4\gamma_{\mathrm s}+4\varepsilon^2=0,
\]
and
\[\gamma_{0,\mathrm q}^4-8\gamma_{0,\mathrm q}^3
 +(20+16\varepsilon^2)\gamma_{0,\mathrm q}^2
 -(16+64\varepsilon^2)\gamma_{0,\mathrm q}
 +48\varepsilon^2=0,\quad 
\gamma_{0,\mathrm c}^3-8\gamma_{0,\mathrm c}^2
 +(16+16\varepsilon^2)\gamma_{0,\mathrm c}
 -32\varepsilon^2=0.
\]
Expanding the roots vanishing at $\varepsilon=0$ yields
\[
\gamma_{\mathrm s}
 =\varepsilon^2+\frac14\varepsilon^4+\mathcal{O}(\varepsilon^6),
\quad
\gamma_{0,\mathrm q}
 =3\varepsilon^2-\frac34\varepsilon^4+\mathcal{O}(\varepsilon^6),
\quad 
\gamma_{0,\mathrm c}
 =2\varepsilon^2-\frac12\varepsilon^6+\mathcal{O}(\varepsilon^8).
\]
Let $a=(a_0,a_1,a_2,a_{12})$ and $b=(b_0,b_1)$ be real unit slow eigenvectors satisfying $M_\hashtag a=-\gamma_{0,\hashtag}a$ and $M_{\mathrm s}b=-\gamma_{\mathrm s}b$, chosen with $a_0,b_0>0$. By the Kronecker-sum representation \cref{eq:spin-coefficient-kronecker-sum}, the tensor mode $a\otimes b^{\otimes n_{\mathrm b}}$ has decay rate
\[
\gamma_\hashtag=\gamma_{0,\hashtag}+n_{\mathrm b}\gamma_{\mathrm s}\sim n\varepsilon^2.
\]
These are actually all the relevant slow eigenvalues of the full generator, as will be clearly seen later in the fast-mode estimate in Appendix~\ref{app:spin-spectral}.

  Set
\[
x=\frac{a_1}{a_0},
\quad
y=\frac{a_2}{a_0},
\quad
z=\frac{a_{12}}{a_0},
\quad
\beta=\frac{b_1}{b_0}.
\]
Writing \(M_\hashtag a=-\gamma_{0,\hashtag}a\) and \(M_{\mathrm s}b=-\gamma_{\mathrm s} b\) coordinatewise gives
\begin{equation}\label{eq:spin-small-eigenvector-equations}
\begin{aligned}
\gamma_{0,\mathrm q}&=2\varepsilon(x+y),&
(2-\gamma_{0,\mathrm q})x&=2\varepsilon(1-z),&
(4-\gamma_{0,\mathrm q})y&=2\varepsilon(1-z),&
(2-\gamma_{0,\mathrm q})z&=2\varepsilon(x+y),\\
\gamma_{0,\mathrm c}&=2\varepsilon(x+y),&
(4-\gamma_{0,\mathrm c})x&=2\varepsilon(1-z),&
(4-\gamma_{0,\mathrm c})y&=2\varepsilon(1-z),&
(4-\gamma_{0,\mathrm c})z&=2\varepsilon(x+y),\\
\gamma_{\mathrm s}&=2\varepsilon\beta,&
(4-\gamma_{\mathrm s})\beta&=2\varepsilon.
\end{aligned}
\end{equation}
The eigenvalue equations give
\[
\begin{array}{c c c c}
\toprule
 & x & y & z \\
\midrule
\mathrm q
 & \varepsilon+\mathcal{O}(\varepsilon^3)
 & \frac12\varepsilon+\mathcal{O}(\varepsilon^3)
 & \frac32\varepsilon^2+\mathcal{O}(\varepsilon^4)
\\[1mm]
\mathrm c
 & \frac12\varepsilon+\mathcal{O}(\varepsilon^3)
 & \frac12\varepsilon+\mathcal{O}(\varepsilon^3)
 & \frac12\varepsilon^2+\mathcal{O}(\varepsilon^4)\\
\bottomrule
\end{array}
\]
and
 $
\beta=\frac12\varepsilon+\mathcal{O}(\varepsilon^3).
 $ 
After normalization,
\begin{equation}\label{eq:spin-small-eigenvector-expansions}
\begin{array}{c c c c c c c}
\toprule
\hashtag & a_0 & a_1 & a_2 & a_{12} & b_0& b_1\\
\midrule
\mathrm q
& 1-\frac58\varepsilon^2+\mathcal{O}(\varepsilon^4)
& \varepsilon+\mathcal{O}(\varepsilon^3)
& \frac12\varepsilon+\mathcal{O}(\varepsilon^3)
& \frac32\varepsilon^2+\mathcal{O}(\varepsilon^4)
&  1-\frac18\varepsilon^2+\mathcal{O}(\varepsilon^4)
& \frac12\varepsilon+\mathcal{O}(\varepsilon^3),
\\[1mm]
\mathrm c
& 1-\frac14\varepsilon^2+\mathcal{O}(\varepsilon^4)
& \frac12\varepsilon+\mathcal{O}(\varepsilon^3)
& \frac12\varepsilon+\mathcal{O}(\varepsilon^3)
& \frac12\varepsilon^2+\mathcal{O}(\varepsilon^4)
&  1-\frac18\varepsilon^2+\mathcal{O}(\varepsilon^4)
& \frac12\varepsilon+\mathcal{O}(\varepsilon^3),
\\
\bottomrule
\end{array}\\[4pt]
\end{equation}
Denote the boundary factor 
 $ 
A_\hashtag
=
a_0Z_1Z_2+a_1Y_1Z_2+a_2Z_1Y_2+a_{12}Y_1Y_2.
 $ 
Then the normalized many-body slow eigenoperator is
\begin{equation}\label{eq:spin-slow-observable}
\Psi_\hashtag
=
A_\hashtag
\prod_{\ell=1}^{n_{\mathrm b}}
\bigl(b_0Z_{j_\ell}+b_1Y_{j_\ell}\bigr)
\prod_{k\in\mathcal I^c}Z_k .
\end{equation}

\paragraph{\textbf{Failure of algebraic closure and positivity.}}
Because all factors away from the boundary square to the identity,
\[
\Psi_\hashtag^2
=
I+d_\hashtag X_1X_2,
\quad
d_\hashtag
=
2(a_1a_2-a_0a_{12}),
\]
with
\[
d_{\mathrm q}
=
-2\varepsilon^2+\mathcal{O}(\varepsilon^4),
\quad
d_{\mathrm c}
=
-\frac12\varepsilon^2+\mathcal{O}(\varepsilon^4).
\]
Thus $d_\hashtag\neq0$ for small $\varepsilon$.

By direct computation, $\mathcal L_{\rm q}(X_1X_2) = -2X_1X_2$ and $\mathcal L_{\rm c}(X_1X_2)=  -4X_1X_2$, so  
the Pauli word $X_1X_2$ lies in the fast sector.  Hence
$\Psi_\hashtag^2$ does not belong to the corresponding slow Riesz range $\mathrm{Ran}(\mathcal P_\hashtag)$.  In particular,
the classical slow space 
 $ \mathrm{Ran}(\mathcal P_{\mathrm c})= 
\operatorname{span}\{I,\Psi_{\mathrm c}\}
 $ 
and the qubit slow space
 $ \mathrm{Ran}(\mathcal P_{\mathrm q})=
\operatorname{span}\{I,Q_z,\Psi_{\mathrm q},
-\mathrm{i}Q_z\Psi_{\mathrm q}\}
 $ 
are therefore not algebras. 
Moreover, $  
\mathcal P_\hashtag(\Psi_\hashtag)=\Psi_\hashtag,
\mathcal P_\hashtag(\Psi_\hashtag^2)=I,
 $ 
and consequently
 $
\mathcal P_\hashtag(\Psi_\hashtag^2)
-
\mathcal P_\hashtag(\Psi_\hashtag)^2
=
-d_\hashtag X_1X_2.
 $ 
Since $X_1X_2$ has eigenvalues $\pm1$ and $d_\hashtag\neq0$, the right-hand
side is indefinite.  This violates the Kadison inequality for a unital
positive map, so $\mathcal P_\hashtag$ is not positive and hence is not a
UCP map.

\hypertarget{app:spin-spectral}{\subsection{Error bounds for the conditional expectation and reduced dynamics}\label{app:spin-spectral}}

In this section, we provide technical estimates that justify the choice of the conditional expectation $E_\hashtag$ as the nearby idempotent of the metastable propagator and the reduced dynamics as a good approximation of the full semigroup over the metastable lifetime. Specifically, we view the slow spectral projection as an $\varepsilon$-perturbation of the physical conditional expectation
 using the perturbation theory for linear operators \cite{Kato1995}, and provide uniform estimates on the reduced dynamics.

We first establish uniform local spectral control. 
For \(1\leq\ell\leq n_{\mathrm b}\), set \(j_\ell=2\ell+2\) and \(\mathcal L_{\ell,\hashtag}:=\mathcal L_{j_\ell}\), with \(\mathcal L_{0,\hashtag}\) as in Part I. Then
$
\mathcal L_\hashtag=\mathcal D_X+\sum_{\ell=0}^{n_{\mathrm b}}\mathcal L_{\ell,\hashtag}.
 $ 
We note that
 $
[\mathcal L_{\ell,\hashtag},\mathcal L_{k,\hashtag}]=0,
 [\mathcal D_X,\mathcal L_{\ell,\hashtag}]=0,
 $ 
and therefore the global semigroup also has the exact factorization
\begin{equation}\label{eq:semigroup-factorization}
T_t^{(\hashtag)}=e^{t\mathcal D_X}
\prod_{\ell=0}^{n_{\mathrm b}}e^{t\mathcal L_{\ell,\hashtag}}.
\end{equation}
When fixing $\rm q$ or $\rm c$, we write \(\mathcal L_\ell=\mathcal L_{\ell,\hashtag}\) and \(T_{\ell,t}=e^{t\mathcal L_\ell}\).
\hypertarget{estimate-one-local-piece}{\subsubsection{Local estimates}\label{estimate-one-local-piece}}

Define the zero-field local generators by setting \(\varepsilon=0\): \(\mathcal L_{0,\mathrm q}^{(0)}=\mathcal D_{Z_1Z_2}+\mathcal D_{Z_2Z_3}\), \(\mathcal L_{0,\mathrm c}^{(0)}=\mathcal L_{0,\mathrm q}^{(0)}+\mathcal D_{Z_1}\), and \(\mathcal L_\ell^{(0)}=\mathcal D_{Z_{j_\ell-1}Z_{j_\ell}}+\mathcal D_{Z_{j_\ell}Z_{j_\ell+1}}\) for \(1\leq\ell\leq n_{\mathrm b}\). The local semigroups defined by these generators fix the Pauli words that commute with the corresponding ZZ bonds. The projections onto these locally-fixed subspaces can be explicitly written as
\begin{equation}\label{eq:spin-local-projections}
\begin{aligned}
E_{0,\mathrm q} =\mathcal E_{Z_1Z_2}\mathcal E_{Z_2Z_3}, \quad E_{0,\mathrm c}=\mathcal E_{Z_1}E_{0,\mathrm q}, \quad 
E_\ell =\mathcal E_{Z_{j_\ell-1}Z_{j_\ell}}\mathcal E_{Z_{j_\ell}Z_{j_\ell+1}}
 (\ell\geq1).
\end{aligned} 
\end{equation}
Here we denote $ \mathcal E_P(A):= \frac12(A+ P A P)$ which can in fact be viewed as the 
  conditional expectation onto the commutant of a Pauli word \(P\). 
Define the local slow spectral projection by
\begin{equation}\label{eq:spin-local-riesz}
\mathcal P_\ell=\frac1{2\pi\mathrm{i}}
\oint_{|z|=1}(z\,\mathrm{id}-\mathcal L_\ell)^{-1}\,\dd z.
\end{equation}
In the following estimate, we show that those zero-field projections are close to the Riesz projections \cref{eq:spin-local-riesz}. 
\begin{lem} \label{lem:spin-local-spectral}
For \(0<\varepsilon\leq1/16\), uniformly in \(\ell\), \(n\), and any reference-system dimension,
\[
\begin{aligned}
\|\mathcal P_\ell-E_\ell\|_{\mathrm{cb}}
&\leq C_0\varepsilon,\quad
\|T_{\ell,t}(\mathrm{id}-\mathcal P_\ell)\|_{\mathrm{cb}}
\leq C_1e^{-t/2},\\
\|\mathcal L_\ell\mathcal P_\ell\|_{\mathrm{cb}}
&\leq C_2\varepsilon^2,\quad
\|\mathcal P_\ell\|_{\mathrm{cb}}
\leq1+C_0\varepsilon,
\end{aligned}
\]
where the constants are independent of $n$.
\end{lem}
\begin{proof}
 Each of the local maps $E_\ell, \mathcal L_\ell$ and $\mathcal P_\ell$ etc. can be viewed as a linear map on a $3$-spin algebra which can be identified with $M_8$. We define the normalized Hilbert--Schmidt norm by \(\|A\|_{2,\tau_3}:=\tau_3(A^*A)^{1/2}\), with \(\tau_3=2^{-3}\operatorname{Tr}\), and \(\|S\|_{2\to2}:=\sup_{A\neq0}\|S(A)\|_{2,\tau_3}/\|A\|_{2,\tau_3}\) is its induced operator norm.

The zero-field generator $\mathcal L_\ell^{(0)}=\mathcal D_{Z_{j_\ell-1}Z_{j_\ell}}+\mathcal D_{Z_{j_\ell}Z_{j_\ell+1}}$ is a sum of commuting Pauli dephasings and is therefore self-adjoint for the local Hilbert--Schmidt inner product. Its Pauli-basis spectrum is contained in \(\{0,-2,-4,-6\}\), so zero is semisimple and is separated from the fast spectrum by at least $2$. By the spectral theorem for normal operators on Hilbert spaces~\cite[Chapter~I, \S6.9]{Kato1995}, \(\|(z\,\mathrm{id}-\mathcal L_\ell^{(0)})^{-1}\|_{2\to2}=\operatorname{dist}(z,\operatorname{spec}(\mathcal L_\ell^{(0)}))^{-1}\leq1\) for \(|z|=1\). Moreover, we write \(V_0=X_1+X_2\) and \(V_\ell=X_{j_\ell}\). Then \(\mathcal L_\ell-\mathcal L_\ell^{(0)}=\mathrm{i}\varepsilon[V_\ell,\,\cdot\,]\), and \(\|[V_\ell,A]\|_2\leq2\|V_\ell\|\|A\|_2\). We have \(\|V_\ell\|=1\) and \(\|V_0\|=\|X_1+X_2\|=2\). Consequently, uniformly over all local blocks,
 $
\|\mathcal L_\ell-\mathcal L_\ell^{(0)}\|_{2\to2}\leq4\varepsilon.
 $ 
Thus on the unit circle,
\[
\|(z\,\mathrm{id}-\mathcal L_\ell^{(0)})^{-1}  (\mathcal L_\ell-\mathcal L_\ell^{(0)})\|_{2\to2}
\leq\|(z\,\mathrm{id}-\mathcal L_\ell^{(0)})^{-1}\|_{2\to2}\| (\mathcal L_\ell-\mathcal L_\ell^{(0)})\|_{2\to2}
\leq4\varepsilon\leq\frac14.
\]
By the factorization $z\,\mathrm{id}-\mathcal L_\ell
=(z\,\mathrm{id}-\mathcal L_\ell^{(0)})
[\mathrm{id}-(z\,\mathrm{id}-\mathcal L_\ell^{(0)})^{-1} (\mathcal L_\ell-\mathcal L_\ell^{(0)})],$ we have
\begin{equation}\label{eq:spin-neumann}
\| (z\,\mathrm{id}-\mathcal L_\ell)^{-1}\|_{2\to2}
\leq
\frac{\| (z\,\mathrm{id}-\mathcal L_\ell^{(0)})^{-1}\|_{2\to2}}
{1-\| (z\,\mathrm{id}-\mathcal L_\ell^{(0)})^{-1} (\mathcal L_\ell-\mathcal L_\ell^{(0)})\|_{2\to2}}
\leq\frac1{1-4\varepsilon}\leq2.
\end{equation}
The resolvent identity gives
\[
\begin{aligned}
\|\mathcal P_\ell-E_\ell\|_{2\to2}
 \leq\frac1{2\pi}\oint_{|z|=1}
\bigl\|(z\,\mathrm{id}-\mathcal L_\ell)^{-1} -(z\,\mathrm{id}-\mathcal L_\ell^{(0)})^{-1}
\bigr\|_{2\to2}\,|\dd z| \leq\frac1{2\pi}\,2\cdot4\varepsilon\cdot1\cdot2\pi
=8\varepsilon.
\end{aligned}
\]
For maps on $M_8$ there exists a universal constant $\widetilde C_0$ such that $\norm{\cdot}_{\mathrm{cb}}\le \widetilde C_0 \norm{\cdot}_{2\to 2}$, where $C_0$ is independent of $n$. Recall that each of $\mathcal P_\ell$ and $E_\ell$ can be viewed as a linear map on $M_8$. Therefore, we have
\begin{equation}
\|\mathcal P_\ell-E_\ell\|_{\mathrm{cb}}\leq \widetilde C_0 \cdot 8 \varepsilon =: C_0 \varepsilon.
\end{equation}

For the fast estimate, let \(\mathcal C_{\mathrm f}\) be e.g. the positively oriented boundary of the rectangle \(\{-7\leq\operatorname{Re}z\leq-1/2,\ |\operatorname{Im}z|\leq1\}\). On this contour, the distance to the zero-field spectrum is at least \(1/2\), so \(\|(z\,\mathrm{id}-\mathcal L_\ell^{(0)})^{-1}\|_{2\to2}\leq2\). Since \(\|\mathcal L_\ell-\mathcal L_\ell^{(0)}\|_{2\to2}\leq1/4\), the same estimate as in \cref{eq:spin-neumann} gives \(\|(z\,\mathrm{id}-\mathcal L_\ell)^{-1}\|_{2\to2}\leq4\), and the preceding cb comparison gives \(\|(z\,\mathrm{id}-\mathcal L_\ell)^{-1}\|_{\mathrm{cb}}\leq 4\widetilde C_0  \). The Dunford functional calculus \cite[Chapter~I, \S5.6]{Kato1995} together with the fact that the contour has length \(17\) and satisfies \(\operatorname{Re}z\leq-1/2\), gives
\[
\begin{aligned}
\|T_{\ell,t}(\mathrm{id}-\mathcal P_\ell)\|_{\mathrm{cb}}\le\norm{\frac1{2\pi\mathrm{i}}\oint_{\mathcal C_{\mathrm f}}
e^{tz}(z\,\mathrm{id}-\mathcal L_\ell)^{-1}\,\dd z }
 \leq\frac{17}{2\pi}\cdot e^{-t/2}\cdot 4 \widetilde C_0 =:C_1e^{-t/2},
\end{aligned}
\]
The constants are independent of \(n\).

It remains to show that the slow local rates start at second order. Write \(\mathcal L_\ell=\mathcal L_\ell^{(0)}+\varepsilon\mathcal V_\ell\), with \(\mathcal V_\ell=\mathrm{i}[V_\ell,\,\cdot\,]\). Each X-field anticommutes with at least one of the local ZZ bonds, so \(E_\ell(V_\ell)=0\). Since \(E_\ell(A)\in\operatorname{Ran}E_\ell\), the bimodule property of the conditional expectation gives
\[
\begin{aligned}
E_\ell\mathcal V_\ell E_\ell(A) =\mathrm{i}E_\ell\!\left(V_\ell E_\ell(A)-E_\ell(A)V_\ell\right) =\mathrm{i}\left(E_\ell(V_\ell)E_\ell(A)-E_\ell(A)E_\ell(V_\ell)\right)=0.
\end{aligned}
\]
The Neumann expansion of the resolvent on the fixed contour is analytic in \(\varepsilon\), and contour integration therefore gives
 $
\mathcal P_\ell=E_\ell+\varepsilon\mathcal P_\ell^{(1)}+\mathcal{O}(\varepsilon^2)
 $ 
in the local \(2\to2\) norm or equivalently in cb norm because the local space is fixed dimensional. Expanding the generator gives
\[
\begin{aligned}
\mathcal P_\ell\mathcal L_\ell\mathcal P_\ell =E_\ell\mathcal L_\ell^{(0)}E_\ell +\varepsilon\left(
\mathcal P_\ell^{(1)}\mathcal L_\ell^{(0)}E_\ell
+E_\ell\mathcal L_\ell^{(0)}\mathcal P_\ell^{(1)}
+E_\ell\mathcal V_\ell E_\ell\right)
+\mathcal{O}(\varepsilon^2).
\end{aligned}
\]
All displayed zeroth- and first-order terms vanish because \(\mathcal L_\ell^{(0)}E_\ell=E_\ell\mathcal L_\ell^{(0)}=0\) and \(E_\ell\mathcal V_\ell E_\ell=0\). Since a Riesz projection commutes with its generator,
 $
\mathcal L_\ell\mathcal P_\ell
=\mathcal P_\ell\mathcal L_\ell\mathcal P_\ell
=\mathcal{O}(\varepsilon^2).
  $ Finally, complete contractivity of \(E_\ell\) and the first estimate of the lemma give
\[
\|\mathcal L_\ell\mathcal P_\ell\|_{\mathrm{cb}}=\mathcal{O}(\varepsilon^2),
\quad
\|\mathcal P_\ell\|_{\mathrm{cb}}
\leq\|E_\ell\|_{\mathrm{cb}}+\|\mathcal P_\ell-E_\ell\|_{\mathrm{cb}}
\leq1+C_0\varepsilon.
\] 
\end{proof}

\hypertarget{count-the-global-slow-modes}{\subsubsection{Dimension counting of global slow manifolds}\label{count-the-global-slow-modes}}

We construct the global slow projection using the local projections 
\begin{equation}\label{eq:spin-global-projection}
E_X=\prod_{j=1}^{n-1}\mathcal E_{X_jX_{j+1}},
\quad
\mathcal P_\hashtag
=E_X\mathcal P_{0,\hashtag}\prod_{\ell=1}^{n_{\mathrm b}}\mathcal P_\ell,
\end{equation}
Here \(\mathcal P_{0,\hashtag}\) is the slow Riesz projection of the model-dependent boundary generator \(\mathcal L_{0,\hashtag}\). The bulk projections \(\mathcal P_\ell\) for \(\ell\geq1\) are the Riesz projections of the same local generator \(\mathcal L_\ell\) and are 
the same in the two models. In the product estimates below, we abbreviate \(\mathcal P_0:=\mathcal P_{0,\hashtag}\) and \(E_0:=E_{0,\hashtag}\).

We recall that the conditional expectation $E_{\hashtag}$ ($\hashtag =\mathrm{c},\mathrm{q}$) onto the slow algebra $\mathcal N_{\hashtag}$ is exactly the zero-field projection and has the factorization $E_{\hashtag} = E_X E_{0,\hashtag} \prod_{\ell =1}^{n_{\mathrm b}} E_\ell$ (see \cref{eq:qubit-expectation,eq:spin-classical-generator}) which formally aligns with the factorization of $\mathcal P_\hashtag$ defined in \cref{eq:spin-global-projection}. In \cref{prop:spin-global-projection} we will show that $\mathcal P_\hashtag$ is $\mathcal{O}(n\varepsilon)$-close to $E_{\hashtag}$ in cb norm, and therefore has the same rank as $E_{\hashtag}$. Later we will see that $\mathcal P_\hashtag$ is indeed the Riesz projection of the metastable propagator $\Phi_\hashtag := e^{t_0 \mathcal L_{\hashtag}}$ when choosing $t_0$ to be a proper metastable time scale (see \cref{cor:spin-global-hypotheses}). 
\begin{prop}\label{prop:spin-global-projection}
There is a constant \(c_*>0\), independent of \(n\), such that \(n\varepsilon\leq c_*\) implies that\[
\|\mathcal P_\hashtag-E_\hashtag\|_{\mathrm{cb}}\lesssim n\varepsilon,\quad
\operatorname{rank}\mathcal P_\hashtag
=\begin{cases}2,&\hashtag=\mathrm c,\\4,&\hashtag=\mathrm q.\end{cases}
\]
Moreover, 
\(\operatorname{Ran}\mathcal P_{\mathrm c}=\operatorname{span}\{I,\Psi_{\mathrm c}\}\) and \(\operatorname{Ran}\mathcal P_{\mathrm q}=\operatorname{span}\{I,Q_z,\Psi_{\mathrm q},-\mathrm{i}Q_z\Psi_{\mathrm q}\}\).
\end{prop}
\begin{proof}

The product difference telescopes as
 \[
\prod_{\ell=0}^{n_{\mathrm b}}\mathcal P_\ell-\prod_{\ell=0}^{n_{\mathrm b}}E_\ell
=\sum_{k=0}^{n_{\mathrm b}}
\left(\prod_{\ell=0}^{k-1}\mathcal P_\ell\right)
(\mathcal P_k-E_k)
\left(\prod_{\ell=k+1}^{n_{\mathrm b}}E_\ell\right).
 \] 
Using \(\|E_X\|_{\mathrm{cb}}=\|E_\ell\|_{\mathrm{cb}}=1\) and the bounds from \cref{lem:spin-local-spectral} gives
\[
\|\mathcal P_\hashtag-E_\hashtag\|_{\mathrm{cb}}
\leq C_0\varepsilon\sum_{j=0}^{n_{\mathrm b}}(1+C_0\varepsilon)^j
=(1+C_0\varepsilon)^{n_{\mathrm b}+1}-1.
\]
Choose a fixed \(c_*>0\) with \(e^{C_0c_*}-1<1\), and assume \(n\varepsilon\leq c_*\). Then this difference is less than one and is \(\mathcal{O}(n\varepsilon)\). Two idempotents at operator-norm distance less than one have equal rank \cite[Chapter~I, \S4.6]{Kato1995}. Hence $
\operatorname{rank}\mathcal P_\hashtag
=\operatorname{rank}E_\hashtag
=\begin{cases}2,&\hashtag=\mathrm c,\\4,&\hashtag=\mathrm q.\end{cases}
 $ 
It remains to identify the ranges. Under the tensor-coordinate identification in \cref{eq:spin-coefficient-kronecker-sum}, the coefficient vector of \(\Psi_\hashtag\) is \(a\otimes b^{\otimes n_{\mathrm b}}\). Since \(M_\hashtag a=-\gamma_{0,\hashtag} a\) and \(M_{\mathrm s}b=-\gamma_{\mathrm s}b\), the summands in \cref{eq:spin-coefficient-kronecker-sum} give
\begin{equation}\label{eq:spin-dressed-local-eigenrelations}
\mathcal L_{0,\hashtag}(\Psi_\hashtag)=-\gamma_{0,\hashtag}\Psi_\hashtag,\quad
\mathcal L_{\ell,\hashtag}(\Psi_\hashtag)=-\gamma_{\mathrm s}\Psi_\hashtag,\quad 1\leq\ell\leq n_{\mathrm b}.
\end{equation}
Both eigenvalues lie inside the contour \(|z|=1\). Hence the Riesz formulas and \cref{eq:spin-dressed-local-eigenrelations} give
\[
\mathcal P_{0,\hashtag}(\Psi_\hashtag)
=\frac1{2\pi\mathrm{i}}\oint_{|z|=1}\frac{\dd z}{z+\gamma_{0,\hashtag}}\Psi_\hashtag
=\Psi_\hashtag,\quad
\mathcal P_\ell(\Psi_\hashtag)
=\frac1{2\pi\mathrm{i}}\oint_{|z|=1}\frac{\dd z}{z+\gamma_{\mathrm s}}\Psi_\hashtag
=\Psi_\hashtag.
\]
Every Pauli word in \(\Psi_\hashtag\) carries \(Z\) or \(Y\) at both ends of each XX bond, so it commutes with \(X_jX_{j+1}\). Therefore \(E_X(\Psi_\hashtag)=\Psi_\hashtag\), and the product definition of the global projection gives
\begin{equation}\label{eq:spin-global-projection-fixes-dressed}
\mathcal P_\hashtag(\Psi_\hashtag)=\Psi_\hashtag.
\end{equation}
For the fixed eigenoperator ($I$ in the classical case, $I$ and $Q_z$ in the qubit case), we have
\[
\mathcal L_{\ell,\hashtag}(I)=0,\quad E_X(I)=I,
\quad
\mathcal L_{\ell,\mathrm q}(Q_z)=0,\quad E_X(Q_z)=Q_z,
\quad 0\leq\ell\leq n_{\mathrm b},
\]
and hence
 $
\mathcal P_{\mathrm c}(I)=I, 
\mathcal P_{\mathrm q}(I)=I, 
\mathcal P_{\mathrm q}(Q_z)=Q_z.
 $ 
For each \(\ell\), the local semigroup \(e^{t\mathcal L_{\ell,\mathrm q}}\) is UCP and fixes the unitary \(Q_z\), and the UCP map \(E_X\) also fixes \(Q_z\). Hence \(Q_z\) belongs to their multiplicative domains \cite[Theorem~3.18]{Paulsen2003}, so
\begin{equation}\label{eq:spin-local-Qz-symmetry}
e^{t\mathcal L_{\ell,\mathrm q}}(Q_zA)=Q_ze^{t\mathcal L_{\ell,\mathrm q}}(A),\quad
E_X(Q_zA)=Q_zE_X(A),
\quad 0\leq\ell\leq n_{\mathrm b}.
\end{equation}
Differentiating \cref{eq:spin-local-Qz-symmetry} at \(t=0\) gives \(\mathcal L_{\ell,\mathrm q}(Q_zA)=Q_z\mathcal L_{\ell,\mathrm q}(A)\). Applying it to \cref{eq:spin-dressed-local-eigenrelations} gives
\[
\mathcal L_{0,\mathrm q}(-\mathrm{i}Q_z\Psi_{\mathrm q})
=-\gamma_{0,\mathrm q}(-\mathrm{i}Q_z\Psi_{\mathrm q}),\quad
\mathcal L_{\ell,\mathrm q}(-\mathrm{i}Q_z\Psi_{\mathrm q})
=-\gamma_{\mathrm s}(-\mathrm{i}Q_z\Psi_{\mathrm q}),
\quad 1\leq\ell\leq n_{\mathrm b},
\]
and \(E_X(-\mathrm{i}Q_z\Psi_{\mathrm q})=-\mathrm{i}Q_z\Psi_{\mathrm q}\). Therefore
\[
\operatorname{span}\{I,\Psi_{\mathrm c}\}\subseteq\operatorname{Ran}\mathcal P_{\mathrm c},\quad
\operatorname{span}\{I,Q_z,\Psi_{\mathrm q},-\mathrm{i}Q_z\Psi_{\mathrm q}\}\subseteq\operatorname{Ran}\mathcal P_{\mathrm q}.
\]
The two left-hand spaces have dimensions \(2\) and \(4\), respectively, equal to the ranks computed above, so both inclusions are equalities.
\end{proof}

\hypertarget{bound-the-remaining-evolution}{\subsubsection{Evolution in the fast sector}\label{bound-the-remaining-evolution}}
We next bound the remaining evolution in the fast sector. 
\begin{lem} \label{lem:spin-fast-remainder}
Under the assumptions of \cref{prop:spin-global-projection},
\[
\|T_t^{(\hashtag)}(\mathrm{id}-\mathcal P_\hashtag)\|_{\mathrm{cb}}
\leq Cn e^{-t/2},\quad \forall t\geq0,
\]
with a constant independent of $n$.
\end{lem}
\begin{proof}
 For a single bond, write \(\mathcal E_j=\mathcal E_{X_jX_{j+1}}\). Then
 $
e^{t\mathcal D_P} = 
e^{t(\operatorname{Ad}_P-\mathrm{id})}
=\mathcal E_P+e^{-2t}(\mathrm{id}-\mathcal E_P)
 $, thus 
 $
e^{t\mathcal D_X}
=\prod_{j=1}^{n-1}\left[\mathcal E_j+e^{-2t}(\mathrm{id}-\mathcal E_j)\right]. 
 $ 
Using \(\|\mathcal E_j\|_{\mathrm{cb}}=1\) and \(\|\mathrm{id}-\mathcal E_j\|_{\mathrm{cb}}\leq1\), the product difference telescopes as
\[ 
\begin{aligned}
\|e^{t\mathcal D_X}-E_X\|_{\mathrm{cb}} =
\norm{e^{-2t}\sum_{k=1}^{n-1}
\left(\prod_{j=1}^{k-1}
\left[\mathcal E_j+e^{-2t}(\mathrm{id}-\mathcal E_j)\right]\right)
(\mathrm{id}-\mathcal E_k)
\left(\prod_{j=k+1}^{n-1}\mathcal E_j\right)} \le (n-1)e^{-2t}.
\end{aligned}
\]
Using the semigroup factorization \cref{eq:semigroup-factorization},
\[
\begin{aligned}
T_t^{(\hashtag)}\mathcal P_\hashtag
=\left(e^{t\mathcal D_X}\prod_{\ell=0}^{n_{\mathrm b}}T_{\ell,t}\right)
\left(E_X\prod_{\ell=0}^{n_{\mathrm b}}\mathcal P_\ell\right) =e^{t\mathcal D_X}E_X\prod_{\ell=0}^{n_{\mathrm b}}(T_{\ell,t}\mathcal P_\ell)
=E_X\prod_{\ell=0}^{n_{\mathrm b}}(T_{\ell,t}\mathcal P_\ell),
\end{aligned}
\]
where \(e^{t\mathcal D_X}E_X=E_X\). Adding and subtracting \(E_X\prod_{\ell=0}^{n_{\mathrm b}}T_{\ell,t}\) therefore gives
\[
\begin{aligned}
T_t^{(\hashtag)}(\mathrm{id}-\mathcal P_\hashtag)
 =(e^{t\mathcal D_X}-E_X)\prod_{\ell=0}^{n_{\mathrm b}}T_{\ell,t} +E_X\left[\prod_{\ell=0}^{n_{\mathrm b}}T_{\ell,t}
-\prod_{\ell=0}^{n_{\mathrm b}}(T_{\ell,t}\mathcal P_\ell)\right].
\end{aligned}
\]
The second product difference telescopes as
\[
\prod_{\ell=0}^{n_{\mathrm b}}T_{\ell,t}
-\prod_{\ell=0}^{n_{\mathrm b}}(T_{\ell,t}\mathcal P_\ell)
=\sum_{k=0}^{n_{\mathrm b}}
\left(\prod_{\ell=0}^{k-1}T_{\ell,t}\mathcal P_\ell\right)
T_{k,t}(\mathrm{id}-\mathcal P_k)
\left(\prod_{\ell=k+1}^{n_{\mathrm b}}T_{\ell,t}\right).
\]
Each \(T_{\ell,t}\) is UCP and has cb norm one, while \(\|T_{\ell,t}\mathcal P_\ell\|_{\mathrm{cb}}\leq1+C_0\varepsilon\). The local fast bound from \cref{lem:spin-local-spectral} therefore gives
\[
\begin{aligned}
\|T_t^{(\hashtag)}(\mathrm{id}-\mathcal P_\hashtag)\|_{\mathrm{cb}}
\leq(n-1)e^{-2t}
+C_1(n_{\mathrm b}+1)(1+C_0\varepsilon)^{n_{\mathrm b}} e^{-t/2}\leq Cn e^{-t/2},
\end{aligned}
\]
with a constant independent of \(n\), 
because \((1+C_0\varepsilon)^{n_{\mathrm b}}\leq e^{C_0n\varepsilon}\leq e^{C_0c_*}\).  
\end{proof}

\hypertarget{state-the-resulting-spectral-decomposition}{\subsubsection{Decomposition of the spectral structure}\label{state-the-resulting-spectral-decomposition}}
We now combine the range identification from \cref{prop:spin-global-projection} with the fast bound from \cref{lem:spin-fast-remainder}.
The projections onto the fixed eigenoperators ($\{I\}$ in the classical model and $\{I,Q_z\}$ in the qubit model) are
\[
\mathcal F_{\mathrm c}(A)=\tau(A)I,
\quad
\mathcal F_{\mathrm q}(A)=\tau(A)I+\tau(Q_zA)Q_z.
\]
Define the slow-decaying and the fast-remainder components of the semigroup by
\[
\mathcal{S}_\hashtag=\mathcal P_\hashtag-\mathcal F_\hashtag,
\quad
\mathcal Z_{\hashtag,t}=T_t^{(\hashtag)}(\mathrm{id}-\mathcal P_\hashtag).
\]
With those notations above, we have the following slow--fast decomposition of the semigroup.
\begin{prop} \label{prop:spin-spectral-decomposition}
Under \(n\varepsilon\leq c_*\),
\[
T_t^{(\hashtag)}=\mathcal F_\hashtag
+e^{-\gamma_\hashtag t}\mathcal{S}_\hashtag
+\mathcal Z_{\hashtag,t},\quad
\|\mathcal Z_{\hashtag,t}\|_{\mathrm{cb}}\leq Cn e^{-t/2}.
\]
Moreover, the cb norms of $\mathcal P_\hashtag$ and $\mathcal{S}_\hashtag$ are bounded uniformly in \(n\).
\end{prop}
\begin{proof}
By the identification of the eigenoperators in $\mathrm{Ran} \mathcal P_\hashtag$ from \cref{prop:spin-global-projection}, since the slow eigenvalue is \(-\gamma_\hashtag\) and the fixed eigenvalue is \(0\), we have
\[
T_t^{(\hashtag)}\mathcal F_\hashtag=\mathcal F_\hashtag,\quad
T_t^{(\hashtag)}\mathcal{S}_\hashtag
=e^{-\gamma_\hashtag t}\mathcal{S}_\hashtag.
\]
Therefore the semigroup can be decomposed as $
T_t^{(\hashtag)}
 =T_t^{(\hashtag)}\mathcal P_\hashtag
+T_t^{(\hashtag)}(\mathrm{id}-\mathcal P_\hashtag) =:\mathcal F_\hashtag
+e^{-\gamma_\hashtag t}\mathcal{S}_\hashtag
+\mathcal Z_{\hashtag,t}.
 $ 
The estimate for \(\mathcal Z_{\hashtag,t}\) is exactly \cref{lem:spin-fast-remainder}. For the slow projections, \cref{lem:spin-local-spectral} gives 
 $
\|\mathcal P_\hashtag\|_{\mathrm{cb}}
\leq\prod_{\ell=0}^{n_{\mathrm b}}\|\mathcal P_\ell\|_{\mathrm{cb}}
\leq(1+C_0\varepsilon)^{n_{\mathrm b}+1}
\leq e^{C_0c_*}.
 $ 
Since \(\|\mathcal F_\hashtag\|_{\mathrm{cb}}=1\), the identity \(\mathcal{S}_\hashtag=\mathcal P_\hashtag-\mathcal F_\hashtag\) also bounds \(\|\mathcal{S}_\hashtag\|_{\mathrm{cb}}\) uniformly.  
\end{proof}
 
As a corollary, the assumptions in \cref{thm:main} are satisfied for this spin chain model, and we have a bound on the error between the metastable propagator $\Phi_\hashtag$ and the conditional expectation $E_\hashtag$.
\begin{cor}[Approximate idempotency of $\Phi_\hashtag$ and the error bound of the conditional expectation]\label{cor:spin-global-hypotheses} For $n$ sufficiently large, 
at the metastable timescale 
\(t_0=C_t\log n\) and \(\varepsilon=n^{-\alpha}\), where \(\alpha\geq2\) and \(C_t\geq4\alpha\), we define the metastable propagator \(\Phi_\hashtag=T_{t_0}^{(\hashtag)}\). Then, $\mathcal P_\hashtag$ is the slow Riesz projection of $\Phi_\hashtag$ in the sense of \cref{prop:riesz}, and we have the following estimates with constants independent of \(n\):
\[
\begin{aligned}
\delta_\hashtag
 =\|E_\hashtag-\Phi_\hashtag\|_{\mathrm{cb}}
& \lesssim n\varepsilon+\gamma_\hashtag t_0
+n e^{-t_0/2}  \lesssim n\varepsilon + n\varepsilon^2 t_0 + n e^{-t_0/2} \sim \frac{1}{n^{\alpha - 1}} 
,\\
\eta_\hashtag
 =\|\Phi_\hashtag^2-\Phi_\hashtag\|_{\mathrm{cb}}
&\lesssim\gamma_\hashtag t_0
+n e^{-t_0/2}  \lesssim  n\varepsilon^2 t_0 + n e^{-t_0/2} \sim \frac{\log n}{n^{2\alpha-1}}.
\end{aligned}
\]
In particular, when $\alpha =2$ as in \cref{sec:spin-8}, we have
 $
\delta_\hashtag=\mathcal{O}(n^{-1}),
\eta_\hashtag \sim \frac{\log n}{n^3}.
 $
\end{cor}
\begin{proof}

Since \(\mathcal P_\hashtag\) commutes with \(\Phi_\hashtag\), the spectrum of \(\Phi_\hashtag\) on its range is \(\{1,e^{-\gamma_\hashtag t_0}\}\to\{1\}\) ($n\to\infty$) because $\gamma_\hashtag t_0  \sim  n\varepsilon^2\cdot \log n\sim \log n/n^{2\alpha-1}\to 0~(n\to\infty) $ , whereas \(\|\Phi_\hashtag(\mathrm{id}-\mathcal P_\hashtag)\|_{\mathrm{cb}}\leq Cn e^{-t_0/2}\). Whenever $t_0 = C_t \log n$ with $C_t > 2$, we have $\norm{\Phi_\hashtag(\mathrm{id}-\mathcal P_\hashtag)}_{\mathrm{cb}} \leq \frac{C}{n^{(C_t-2)/2}} \to 0~(n\to\infty)$. Therefore, the
 spectral separation shows that \(\mathcal P_\hashtag\) is exactly the Riesz projection of \(\Phi_\hashtag\) onto the cluster near \(1\). 
 At the metastable time, \cref{prop:spin-spectral-decomposition} gives
\[
E_\hashtag-\Phi_\hashtag
=(E_\hashtag-\mathcal P_\hashtag)
+(1-e^{-\gamma_\hashtag t_0})\mathcal{S}_\hashtag
-\mathcal Z_{\hashtag,t_0}.
\]
The projection comparison from \cref{prop:spin-global-projection}, the uniform cb-norm bound on \(\mathcal{S}_\hashtag\) in \cref{prop:spin-spectral-decomposition}, the inequality \(1-e^{-x}\leq x\), and the fast estimate in \cref{lem:spin-fast-remainder} give 
 $
\delta_\hashtag
\leq
\|E_\hashtag-\mathcal P_\hashtag\|_{\mathrm{cb}}
+\left(1-e^{-\gamma_\hashtag t_0}\right)\|\mathcal{S}_\hashtag\|_{\mathrm{cb}}
+\|\mathcal Z_{\hashtag,t_0}\|_{\mathrm{cb}}
 \lesssim\left(n\varepsilon+\gamma_\hashtag t_0+n e^{-t_0/2}\right).
 $ Here the $\lesssim$ is up to a constant independent of $n$ because all the used cb-norm bounds are uniform in $n$.  

The semigroup property gives \(\Phi_\hashtag^2=T_{2t_0}^{(\hashtag)}\). Applying the same decomposition at the two times yields
\[
\Phi_\hashtag^2-\Phi_\hashtag
=\left(e^{-2\gamma_\hashtag t_0}
-e^{-\gamma_\hashtag t_0}\right)\mathcal{S}_\hashtag
+\mathcal Z_{\hashtag,2t_0}
-\mathcal Z_{\hashtag,t_0}.
\]
Since \(|e^{-2x}-e^{-x}|=e^{-x}(1-e^{-x})\leq x\) for \(x\geq0\), and the two fast remainders are bounded by \(Cn e^{-t_0/2}\), we prove the estimate for \(\eta_\hashtag\).

For \(C_t\geq4\alpha\),
 $
n\varepsilon=n^{1-\alpha},
\gamma_\hashtag t_0\sim\frac{\log n}{n^{2\alpha-1}},
n e^{-t_0/2}=n^{1-C_t/2}\leq n^{1-2\alpha}.
 $ 
Substitution gives \(\delta_\hashtag=\mathcal{O}(n^{1-\alpha})\) and \(\eta_\hashtag=\mathcal{O}((\log n)/n^{2\alpha-1})\). Testing the defect on any decaying slow eigenoperator in $\mathrm{Ran}\mathcal{S}_\hashtag$ 
gives
 $
\eta_\hashtag
\geq
e^{-\gamma_\hashtag t_0}
\left(1-e^{-\gamma_\hashtag t_0}\right)
\sim\gamma_\hashtag t_0
\sim\frac{\log n}{n^{2\alpha-1}}.
 $  
\end{proof}
 
\hypertarget{app:spin-weighted}{\subsubsection{Errors in metastable seminorms}\label{app:spin-weighted}}

We next prove the metastable-seminorm estimates. For this, we first prove \cref{lem:spin-prepared-modes} below, which shows that the slow eigenoperators are close to the logical operators in the subalgebra $\mathcal N_{\hashtag}$. 

  Fix \(\hashtag\) and \(n\), with \(0<\varepsilon\leq\varepsilon_*\) and \(n\varepsilon\leq c_*<1\). 
We denote $Q \in \mathscr Q_{\rm q }: = \{Q_x,Q_y\}$ and $Q \in \mathscr Q_{\rm c}:= \{Q_x\}$.  
Set \(\Psi_{Q_x}:=\Psi_\hashtag\), with \(\Psi_\hashtag\) given by \cref{eq:spin-slow-observable}. Since \(M_\hashtag(\varepsilon)^T=M_\hashtag(-\varepsilon)\) and \(M_{\mathrm s}(\varepsilon)^T=M_{\mathrm s}(-\varepsilon)\), put
\[
\widetilde A_\hashtag=a_0Z_1Z_2-a_1Y_1Z_2-a_2Z_1Y_2+a_{12}Y_1Y_2,
\quad
\widetilde\Psi_{Q_x}:=\widetilde A_\hashtag
\prod_{\ell=1}^{n_{\mathrm b}}(b_0Z_{j_\ell}-b_1Y_{j_\ell})
\prod_{k\in\mathcal I^c}Z_k.
\]
The transpose of \cref{eq:spin-small-eigenvector-equations}, together with the decomposition \cref{eq:spin-coefficient-kronecker-sum}, gives
\begin{equation}\label{eq:spin-Qx-left-right-pair}
\mathcal L_\hashtag(\Psi_{Q_x})=-\gamma_\hashtag\Psi_{Q_x},\quad
\mathcal L_{\hashtag,\ast}(\widetilde\Psi_{Q_x})=-\gamma_\hashtag\widetilde\Psi_{Q_x},\quad
\|\Psi_{Q_x}\|_{2,\tau}=\|\widetilde\Psi_{Q_x}\|_{2,\tau}=1.
\end{equation}
For the qubit model, define
\begin{equation}\label{eq:spin-Qy-left-right-pair}
\Psi_{Q_y}:=-\mathrm{i}Q_z\Psi_{Q_x}=-\mathrm{i}Q_z\Psi_{\mathrm q},\quad
\widetilde\Psi_{Q_y}:=-\mathrm{i}Q_z\widetilde\Psi_{Q_x}.
\end{equation}
Because \(Q_z\) is a fixed unitary of \(T_t^{(\mathrm q)}\), the multiplicative-domain argument \cite[Theorem~3.18]{Paulsen2003} again gives \(T_t^{(\mathrm q)}(Q_zA)=Q_zT_t^{(\mathrm q)}(A)\). Differentiating at \(t=0\), and then taking Hilbert--Schmidt adjoints using that left multiplication by \(Q_z\) is self-adjoint, gives
\[
\mathcal L_{\mathrm q}(Q_zA)=Q_z\mathcal L_{\mathrm q}(A),\quad
\mathcal L_{\mathrm q,\ast}(Q_zA)=Q_z\mathcal L_{\mathrm q,\ast}(A).
\]
Together with \cref{eq:spin-Qx-left-right-pair}, these identities verify
\begin{equation}\label{eq:spin-Qy-eigenrelations}
\mathcal L_{\mathrm q}(\Psi_{Q_y})=-\gamma_{\mathrm q}\Psi_{Q_y},\quad
\mathcal L_{\mathrm q,\ast}(\widetilde\Psi_{Q_y})=-\gamma_{\mathrm q}\widetilde\Psi_{Q_y},\quad
\|\Psi_{Q_y}\|_{2,\tau}=\|\widetilde\Psi_{Q_y}\|_{2,\tau}=1.
\end{equation} 
\begin{lem} \label{lem:spin-prepared-modes}Assume \(0<\varepsilon\leq\varepsilon_*\) and \(n\varepsilon\leq c_*<1\). 
For each \(Q \in \mathscr Q_\hashtag\), let \(\omega_Q=\tau(\widetilde\Psi_Q\Psi_Q)\). Uniformly in $n$ and \(\varepsilon\), we have
\[
\|\Psi_Q-Q\|_{2,\tau}+\|\widetilde\Psi_Q-Q\|_{2,\tau}
\lesssim\sqrt n\,\varepsilon,\quad
|\omega_Q-1|\lesssim n\varepsilon^2,
\quad 
\|\Psi_Q-Q\|\lesssim n\varepsilon,\quad
\|\Psi_Q\|\leq1+\mathcal{O}(\varepsilon^2),
\]
and, for every \(t\geq0\),
\[
\|T_t^{(\hashtag)}((Q-\Psi_Q)^2)\|
\lesssim n\varepsilon^2+e^{-2t}.
\] 
\end{lem}
\begin{proof}
We first prove the estimates for \(Q_x\). From  \cref{eq:spin-small-eigenvector-expansions}, we note that
\begin{equation}\label{eq:QxPsi}
\begin{aligned}
Q_x\Psi_{Q_x}  =(Z_1Z_2A_\hashtag)
\prod_{\ell=1}^{n_{\mathrm b}}Z_{j_\ell}(b_0Z_{j_\ell}+b_1Y_{j_\ell})
 =\left(a_0-\mathrm{i}a_1X_1-\mathrm{i}a_2X_2-a_{12}X_1X_2\right)
\prod_{\ell=1}^{n_{\mathrm b}}\left(b_0-\mathrm{i}b_1X_{j_\ell}\right).
\end{aligned}
\end{equation}
Normalized Pauli words are orthonormal for the inner product defined by the normalized trace 
\(\tau\), so the constant coefficient of \cref{eq:QxPsi} gives \(\tau(Q_x\Psi_{Q_x})=a_0b_0^{n_{\mathrm b}}\), since every nontrivial Pauli word is traceless. Note that \(Q_x\) and \(\Psi_{Q_x}\) have normalized Hilbert--Schmidt norm one,
\[
\|\Psi_{Q_x}-Q_x\|_{2,\tau}^2
=2(1-a_0b_0^{n_{\mathrm b}})\leq2[(1-a_0)+n_{\mathrm b}(1-b_0) ]
\lesssim n\varepsilon^2,
\]
where the last inequality follows from \cref{eq:spin-small-eigenvector-expansions}.

Transposing the real coefficient matrices replaces \(\varepsilon\) by \(-\varepsilon\), hence the left local eigenvectors are
 $
\widetilde a=(a_0,-a_1,-a_2,a_{12}), 
\widetilde b=(b_0,-b_1),
 $ 
and the same calculation gives \(\|\widetilde\Psi_{Q_x}-Q_x\|_{2,\tau}\lesssim\sqrt n\,\varepsilon\). 
The pairing of the left and right eigenoperators is given by
\[
\omega_{Q_x}=\tau(\widetilde\Psi_{Q_x}\Psi_{Q_x})=\langle \widetilde a,a\rangle\langle \widetilde b,b\rangle^{n_{\mathrm b}}
=(a_0^2-a_1^2-a_2^2+a_{12}^2)(b_0^2-b_1^2)^{n_{\mathrm b}},
\]
so \cref{eq:spin-small-eigenvector-expansions} gives \(|\omega_{Q_x}-1|\lesssim n\varepsilon^2\).  

For each bulk factor,
 $
(b_0Z_{j_\ell}+b_1Y_{j_\ell})^2
=(b_0^2+b_1^2)I+b_0b_1(Z_{j_\ell}Y_{j_\ell}+Y_{j_\ell}Z_{j_\ell})
=I.
 $ 
For the boundary factor,
 $ 
A_\hashtag^2=I+d_\hashtag X_1X_2$ where $ 
d_\hashtag=2(a_1a_2-a_0a_{12}),
 $ 
and \cref{eq:spin-small-eigenvector-expansions} gives \(|d_\hashtag|\lesssim\varepsilon^2\). Therefore
\[
\|\Psi_{Q_x}\|=\|A_\hashtag\|
\leq\sqrt{1+|d_\hashtag|}
\leq1+\mathcal{O}(\varepsilon^2).
\]
We next bound the operator norm of \(\Psi_{Q_x}-Q_x\) with the following telescoping expansion
\[
\begin{aligned}
\Psi_{Q_x}-Q_x  
&=(A_\hashtag-Z_1Z_2)\prod_{\ell=1}^{n_{\mathrm b}}(b_0Z_{j_\ell}+b_1Y_{j_\ell})\\
&\quad+Z_1Z_2\sum_{k=1}^{n_{\mathrm b}}
\left(\prod_{\ell=1}^{k-1}(b_0Z_{j_\ell}+b_1Y_{j_\ell})\right)
\left[(b_0Z_{j_k}+b_1Y_{j_k})-Z_{j_k}\right]
\left(\prod_{\ell=k+1}^{n_{\mathrm b}}Z_{j_\ell}\right).
\end{aligned}
\]
 \cref{eq:spin-small-eigenvector-expansions} gives \(\|A_\hashtag-Z_1Z_2\|\lesssim\varepsilon\) and \(\|(b_0Z_{j_\ell}+b_1Y_{j_\ell})-Z_{j_\ell}\|\lesssim\varepsilon\). Since all other factors have norm one,
\[
\|\Psi_{Q_x}-Q_x\|  
\lesssim\varepsilon+n_{\mathrm b}\varepsilon
\lesssim n\varepsilon.
\]
Moreover, by the triangle inequality
\[
\begin{aligned}
\|Q_x-\widetilde\Psi_{Q_x}/\omega_{Q_x}\|_{2,\tau}
 \leq\|Q_x-\widetilde\Psi_{Q_x}\|_{2,\tau}
+|1-\omega_{Q_x}^{-1}|\|\widetilde\Psi_{Q_x}\|_{2,\tau} \lesssim\sqrt n\,\varepsilon+n\varepsilon^2
\lesssim\sqrt n\,\varepsilon,
\end{aligned}
\]
where the last inequality uses \(\omega_{Q_x}\geq1-n\varepsilon^2\geq1-c_*\varepsilon_*\) and \(\sqrt n\,\varepsilon\leq n\varepsilon\leq c_*<1\), for \(\varepsilon_*\) sufficiently small.

For \(S\subseteq\mathcal I\), put \(X_S=\prod_{j\in S}X_j\), with \(X_\varnothing=I\). \cref{eq:QxPsi} shows that both \(Q_x\Psi_{Q_x}\) and \(\Psi_{Q_x}Q_x=(Q_x\Psi_{Q_x})^*\) are linear combinations of \(X_S\), supported on the active sites \(\mathcal I\). Together with \(\Psi_{Q_x}^2=I+d_\hashtag X_1X_2\), the identity
 $
(Q_x-\Psi_{Q_x})^2=I-Q_x\Psi_{Q_x}-\Psi_{Q_x}Q_x+\Psi_{Q_x}^2
 $ 
therefore gives
\[
(Q_x-\Psi_{Q_x})^2=\sum_{S\subseteq\mathcal I}c_SX_S,\quad
c_\varnothing=\tau((Q_x-\Psi_{Q_x})^2)
=\|Q_x-\Psi_{Q_x}\|_{2,\tau}^2
\lesssim n\varepsilon^2.
\]
Set \(L_a=|a_0|+|a_1|+|a_2|+|a_{12}|\) and \(L_b=|b_0|+|b_1|\). The \(\ell^1\)-norm of the Pauli coefficients is submultiplicative, so
\[
\sum_{S\subseteq\mathcal I}|c_S|
\leq(1+L_aL_b^{n_{\mathrm b}})^2 \leq\left[1+(1+\mathcal{O}(\varepsilon))e^{\mathcal{O}(n_{\mathrm b}\varepsilon)}\right]^2
\lesssim1,
\]
 because \cref{eq:spin-small-eigenvector-expansions} gives \(L_a=1+\mathcal{O}(\varepsilon)\) and \(L_b=1+\mathcal{O}(\varepsilon)\) and we set \(n\varepsilon\leq c_*\).
Let \(\partial S=\{j<n:\mathbf1_S(j)\neq\mathbf1_S(j+1)\}\). For each bond \(j,j+1\),
\[
\begin{aligned}
\mathcal D_{Z_jZ_{j+1}}(X_S) =Z_jZ_{j+1}X_SZ_jZ_{j+1}-X_S =\left((-1)^{\mathbf1_S(j)+\mathbf1_S(j+1)}-1\right)X_S
=-2\mathbf1_{\{j\in\partial S\}}X_S.
\end{aligned}
\]
The remaining generator terms act as
 $
\mathrm{i}[H,X_S]=0, 
\mathcal D_{X_jX_{j+1}}(X_S)=0, 
\mathcal D_{Z_1}(X_S)=-2\mathbf1_{\{1\in S\}}X_S.
 $ 
  Summing these contributions gives
\[
\mathcal L_\hashtag(X_S)
=-2\bigl(|\partial S|+\mathbf1_{\{\hashtag=\mathrm c,\,1\in S\}}\bigr)X_S,
\quad 
T_t^{(\hashtag)}(X_S)
=e^{-2t(|\partial S|+\mathbf1_{\{\hashtag=\mathrm c,\,1\in S\}})}X_S.
\]
Every nonempty \(S\subseteq\mathcal I\) is a proper subset of the chain, so \(|\partial S|\geq1\). Consequently,
\[
\begin{aligned}
\|T_t^{(\hashtag)}((Q_x-\Psi_{Q_x})^2)\|
 \leq |c_\varnothing|
+\sum_{\varnothing\neq S\subseteq\mathcal I}
|c_S|e^{-2t(|\partial S|+\mathbf1_{\{\hashtag=\mathrm c,\,1\in S\}})} 
 \lesssim n\varepsilon^2
+e^{-2t}\sum_{\varnothing\neq S}|c_S| \lesssim n\varepsilon^2+e^{-2t}.
\end{aligned}
\]
For the qubit \(Q_y\) pair, \cref{eq:spin-Qy-left-right-pair} and \(Q_y=-\mathrm{i}Q_zQ_x\) give
\[
\Psi_{Q_y}-Q_y=-\mathrm{i}Q_z(\Psi_{Q_x}-Q_x),\quad
\widetilde\Psi_{Q_y}-Q_y=-\mathrm{i}Q_z(\widetilde\Psi_{Q_x}-Q_x).
\]
Since \(Q_z\) anticommutes with \(Q_x\), \(\Psi_{Q_x}\), and \(\widetilde\Psi_{Q_x}\), multiplication by \(-\mathrm{i}Q_z\) preserves all the operator and Hilbert--Schmidt norms above and gives
\[
(Q_y-\Psi_{Q_y})^2=(Q_x-\Psi_{Q_x})^2,\quad
\omega_{Q_y}=\tau(\widetilde\Psi_{Q_y}\Psi_{Q_y})
=\tau(\widetilde\Psi_{Q_x}\Psi_{Q_x})=\omega_{Q_x}.
\]
The \(Q_y\) estimates therefore follow from the \(Q_x\) estimates.
\end{proof}
 
\begin{prop}[The error bound in the metastable seminorm]\label{prop:spin-weighted-error} Under the assumptions of \cref{prop:spin-global-projection}, 
the explicit expectation satisfies
\[
\|E_\hashtag-\mathcal P_\hashtag\|_{\Phi_\hashtag,2}^{\mathrm{cb}}
\lesssim\sqrt n\varepsilon+e^{-t_0},
\]
and consequently
\[
\epsilon_\hashtag
:=\|E_\hashtag-\Phi_\hashtag\|_{\Phi_\hashtag,2}^{\mathrm{cb}}
\lesssim\sqrt n\varepsilon+\gamma_\hashtag t_0
+n e^{-t_0/2} \lesssim \sqrt{n} \varepsilon+ n\varepsilon^2 t_0 + n e^{-t_0/2} \sim \frac1{n^{\alpha-1/2}}.
\]
The constants are independent of \(n\), \(\varepsilon\).
\end{prop}
\begin{proof}
Write \(\Phi=\Phi_\hashtag=T_{t_0}^{(\hashtag)}\) and set
 $
\mathscr Q_{\mathrm c}=\{Q_x\}, 
\mathscr Q_{\mathrm q}=\{Q_x,Q_y\}.
 $ 
Recall from the decomposition \cref{prop:spin-spectral-decomposition} that 
\[
\mathcal F_{\mathrm c}(A)=\tau(A)I,\quad
\mathcal F_{\mathrm q}(A)=\tau(A)I+\tau(Q_zA)Q_z.
\]
The explicit conditional expectations therefore have the form
\begin{equation}\label{eq:spin-expectation-mode-expansion}
E_\hashtag(A)=\mathcal F_\hashtag(A)
+\sum_{Q\in\mathscr Q_\hashtag}\tau(QA)Q.
\end{equation}
For each \(Q\in\mathscr Q_\hashtag\), let \(\Psi_Q,\widetilde\Psi_Q\) be the left--right pair defined in \cref{eq:spin-Qx-left-right-pair,eq:spin-Qy-eigenrelations}, and let \(\omega_Q=\tau(\widetilde\Psi_Q\Psi_Q)\). The biorthogonal expansion of the Riesz projection onto the semisimple slow spectral subspace \cite[Chapter~I, \S5.5]{Kato1995} gives
\begin{equation}\label{eq:spin-riesz-mode-expansion}
\mathcal P_\hashtag(A)=\mathcal F_\hashtag(A)
+\sum_{Q\in\mathscr Q_\hashtag}
\frac{\tau(\widetilde\Psi_QA)}{\omega_Q}\Psi_Q.
\end{equation}
 Subtracting \cref{eq:spin-riesz-mode-expansion} from \cref{eq:spin-expectation-mode-expansion} gives
\begin{equation}\label{eq:spin-rank-one-difference}
\begin{aligned}
(E_\hashtag-\mathcal P_\hashtag)(A)
=\sum_{Q\in\mathscr Q_\hashtag}\bigl[
\tau(QA)(Q-\Psi_Q)
+\tau\bigl((Q-\widetilde\Psi_Q/\omega_Q)A\bigr)\Psi_Q
\bigr].
\end{aligned}
\end{equation}
 For \(D\in M_{2^n}\), denote the linear functional \(\varphi_D(A)=\tau(DA)\). Schatten \(1\)--\(\infty\) duality \cite[Equation~(1.173)]{Watrous2018} gives \(\|\varphi_D\|=\tau(|D|)\). The same norm holds completely: for unit vectors \(u,v\in\mathbb C^m\) and \(X\in M_m\otimes M_{2^n}\),
\[
\left|\left\langle u,(\mathrm{id}_{M_m}\otimes\varphi_D)(X)v\right\rangle\right|
=\left|\tau\bigl(D(\langle u|\otimes I)X(|v\rangle\otimes I)\bigr)\right|
\leq\tau(|D|)\|X\|,
\]
and equality already occurs at \(m=1\). Therefore \(\|\varphi_D\|_{\mathrm{cb}}=\tau(|D|)\). In particular, \(\|\varphi_Q\|_{\mathrm{cb}}=1\). For the second functional in \cref{eq:spin-rank-one-difference}, Schatten \(2\)--\(2\) H\"older inequality \cite[Equation~(1.174)]{Watrous2018} and \cref{lem:spin-prepared-modes} give
\begin{equation}\label{eq:spin-left-functional}
\begin{aligned}
\|\varphi_{Q-\widetilde\Psi_Q/\omega_Q}\|_{\mathrm{cb}}
 &=\tau(|Q-\widetilde\Psi_Q/\omega_Q|)
\leq\|Q-\widetilde\Psi_Q/\omega_Q\|_{2,\tau}\\
&\leq\|Q-\widetilde\Psi_Q\|_{2,\tau}
+|1-\omega_Q^{-1}|\|\widetilde\Psi_Q\|_{2,\tau}
\lesssim\sqrt n\varepsilon.
\end{aligned}
\end{equation}
We now apply the metastable-seminorm definition \cref{eq:weighted-seminorm} at an arbitrary amplification. Fix \(m\geq1\) and \(\|X\|\leq1\), and set
\[
\Lambda_Q=(\mathrm{id}_{M_m}\otimes\varphi_Q)(X),\quad
\widetilde\Lambda_Q=(\mathrm{id}_{M_m}\otimes\varphi_{Q-\widetilde\Psi_Q/\omega_Q})(X).
\]
The functional estimates above give \(\|\Lambda_Q\|\leq1\) and, by \cref{eq:spin-left-functional}, \(\|\widetilde\Lambda_Q\|\lesssim\sqrt n\varepsilon\). For the first term in \cref{eq:spin-rank-one-difference}, the final estimate of \cref{lem:spin-prepared-modes}, evaluated at \(t=t_0\), gives
\begin{equation}\label{eq:spin-first-weighted-term}
\begin{aligned}
\|\Lambda_Q\otimes(Q-\Psi_Q)\|_{2,\mathscr D_{\star,\Phi_m}}^2
&=\|(\mathrm{id}_{M_m}\otimes \Phi)(\Lambda_Q^*\Lambda_Q\otimes(Q-\Psi_Q)^2)\|\\
&=\|\Lambda_Q^*\Lambda_Q\|\,\|T_{t_0}^{(\hashtag)}((Q-\Psi_Q)^2)\|
\lesssim n\varepsilon^2+e^{-2t_0}.
\end{aligned}
\end{equation}
For the second term, positivity and unitality of \(\Phi\) imply \(0\leq\Phi(\Psi_Q^2)\leq\|\Psi_Q\|^2I\). Consequently, using \cref{eq:spin-left-functional} and the bound \(\|\Psi_Q\|\leq1+\mathcal{O}(\varepsilon^2)\) from \cref{lem:spin-prepared-modes},
\begin{equation}\label{eq:spin-second-weighted-term}
\begin{aligned}
\|\widetilde\Lambda_Q\otimes\Psi_Q\|_{2,\mathscr D_{\star,\Phi_m}}^2
&=\|\widetilde\Lambda_Q^*\widetilde\Lambda_Q\otimes\Phi(\Psi_Q^2)\|
\leq\|\widetilde\Lambda_Q\|^2\|\Psi_Q\|^2
\lesssim n\varepsilon^2.
\end{aligned}
\end{equation}Note that $\abs{\mathscr Q_\hashtag}\leq2$. By 
applying the triangle inequality to \cref{eq:spin-rank-one-difference} at each amplification and then taking the suprema in \cref{eq:weighted-seminorm}, \cref{eq:spin-first-weighted-term,eq:spin-second-weighted-term} give
\begin{equation}\label{eq:spin-weighted-error}
\begin{aligned}
\|E_\hashtag-\mathcal P_\hashtag\|_{\Phi_\hashtag,2}^{\mathrm{cb}}
&\leq\sum_{Q\in\mathscr Q_\hashtag}\sup_{m\geq1,\,\|X\|\leq1}
\left(
\|\Lambda_Q\otimes(Q-\Psi_Q)\|_{2,\mathscr D_{\star,\Phi_m}}
+\|\widetilde\Lambda_Q\otimes\Psi_Q\|_{2,\mathscr D_{\star,\Phi_m}}
\right)\\
& 
\lesssim\sqrt n\varepsilon+e^{-t_0},
\end{aligned}
\end{equation}
Finally, \cref{prop:spin-spectral-decomposition} gives
\[
E_\hashtag-\Phi_\hashtag
=(E_\hashtag-\mathcal P_\hashtag)
+(1-e^{-\gamma_\hashtag t_0})\mathcal{S}_\hashtag
-\mathcal Z_{\hashtag,t_0}.
\]
By \cref{prop:weighted-hierarchy,eq:spin-weighted-error}, \cref{prop:spin-spectral-decomposition}, and \(1-e^{-x}\leq x\) for \(x\geq0\),
\[
\begin{aligned}
\|E_\hashtag-\Phi_\hashtag\|_{\Phi_\hashtag,2}^{\mathrm{cb}}
&\leq\|E_\hashtag-\mathcal P_\hashtag\|_{\Phi_\hashtag,2}^{\mathrm{cb}}
+(1-e^{-\gamma_\hashtag t_0})\|\mathcal{S}_\hashtag\|_{\mathrm{cb}}
+\|\mathcal Z_{\hashtag,t_0}\|_{\mathrm{cb}}\\
&\lesssim\sqrt n\varepsilon+e^{-t_0}+\gamma_\hashtag t_0+n e^{-t_0/2}\\
&\lesssim\sqrt n\varepsilon+\gamma_\hashtag t_0+n e^{-t_0/2}.
\end{aligned}
\] 
\end{proof}
\begin{cor}[Reduced-iteration errors]\label{cor:spin-reduced-iterations}
Let \(\mathcal K_\hashtag=E_\hashtag\Phi_\hashtag\iota_\hashtag\). Then
\[
\|\iota_\hashtag\mathcal K_\hashtag^kE_\hashtag-\Phi_\hashtag^k\|_{\mathrm{cb}}
\leq2k\delta_\hashtag=\mathcal{O}(k/n),\quad k\geq1,
\]
and, after one microscopic preparation step,
\[
\|\Phi_\hashtag\iota_\hashtag\mathcal K_\hashtag^kE_\hashtag-\Phi_\hashtag^{k+1}\|_{\mathrm{cb}}
\leq(k+1)(\epsilon_\hashtag+\eta_\hashtag)
=\mathcal{O}((k+1)n^{-3/2}),\quad k\geq0.
\]
\end{cor}
\begin{proof}
Apply \cref{cor:reduced-transition} with the bounds from \cref{cor:spin-global-hypotheses,prop:spin-weighted-error}.
\end{proof}

\hypertarget{app:spin-lifetime}{\subsubsection{Direct iteration error over the slow lifetime}\label{app:spin-lifetime}}
In fact, we can improve the error bound of the reduced iteration, over some number of metastable timescales. 
\begin{prop} \label{prop:spin-slow-lifetime}
Uniformly for \(t\geq0\),
\[
|u_\hashtag(t)-e^{-\gamma_\hashtag t}|\lesssim n\varepsilon^2.
\]
Consequently,
\[
\|\mathcal K_\hashtag^k-E_\hashtag T_{kt_0}^{(\hashtag)}\iota_\hashtag\|_{\mathrm{cb}}
\lesssim (k+1)n\varepsilon^2.
\]
If \(\varepsilon=n^{-\alpha}\) and \(t_0=C_t\log n\), where \(\alpha\geq2\) and \(C_t\geq4\alpha\), then for every fixed \(S\),
\[
\sup_{1\leq k,\,k\gamma_\hashtag t_0\leq S}
\|\iota_\hashtag\mathcal K_\hashtag^kE_\hashtag-T_{kt_0}^{(\hashtag)}\|_{\mathrm{cb}}
=\mathcal{O}_S(n^{1-\alpha}+(\log n)^{-1}).
\]
\end{prop}
\begin{proof}
Recall that 
 $
u_\hashtag(t)=f_\hashtag(t)f_{\mathrm s}(t)^{n_{\mathrm b}}, 
\gamma_\hashtag=\gamma_{0,\hashtag}+n_{\mathrm b}\gamma_{\mathrm s}.
 $ 
The spectral decompositions of \(M_\hashtag\) and \(M_{\mathrm s}\) give, uniformly for \(t\geq0\),
\[
|f_\hashtag(t)-e^{-\gamma_{0,\hashtag} t}|+|f_{\mathrm s}(t)-e^{-\gamma_{\mathrm s}t}|\lesssim\varepsilon^2,
\quad
|f_\hashtag(t)|,|f_{\mathrm s}(t)|\leq1.
\]
Therefore
\[
\begin{aligned}
|u_\hashtag(t)-e^{-\gamma_\hashtag t}|
&\leq|f_\hashtag(t)-e^{-\gamma_{0,\hashtag} t}|\,|f_{\mathrm s}(t)|^{n_{\mathrm b}}
+e^{-\gamma_{0,\hashtag} t}|f_{\mathrm s}(t)^{n_{\mathrm b}}-e^{-n_{\mathrm b}\gamma_{\mathrm s}t}|\\
&\leq\varepsilon^2+\sum_{a=0}^{n_{\mathrm b}-1}|f_{\mathrm s}(t)|^{n_{\mathrm b}-1-a}
|f_{\mathrm s}(t)-e^{-\gamma_{\mathrm s}t}|e^{-a\gamma_{\mathrm s}t}
\lesssim(n_{\mathrm b}+1)\varepsilon^2\lesssim n\varepsilon^2.
\end{aligned}
\]
Define \(\mathcal K_\hashtag(t)=E_\hashtag T_t^{(\hashtag)}\iota_\hashtag\), so \(\mathcal K_\hashtag=\mathcal K_\hashtag(t_0)\). The fixed components of $\mathcal N_\hashtag$ form \(\operatorname{Ran}\mathcal F_{\mathrm c}=\operatorname{span}\{I\}\) and \(\operatorname{Ran}\mathcal F_{\mathrm q}=\operatorname{span}\{I,Q_z\}\), while the non-fixed ones are precisely \(Q\in\mathscr Q_\hashtag\). Thus,
\[
\left.\mathcal K_\hashtag(t)\right|_{\operatorname{Ran}\mathcal F_\hashtag}=\mathrm{id},\quad
\mathcal K_\hashtag(t)(Q)=u_\hashtag(t)Q,\quad Q\in\mathscr Q_\hashtag.
\]
Thus \(\mathcal K_\hashtag^k\) and \(\mathcal K_\hashtag(kt_0)\) agree on \(\operatorname{Ran}\mathcal F_\hashtag\), while their multipliers on each \(Q\in\mathscr Q_\hashtag\) are \(u_\hashtag(t_0)^k\) and \(u_\hashtag(kt_0)\), respectively. Hence
\[
\begin{aligned}
\|\mathcal K_\hashtag^k-\mathcal K_\hashtag(kt_0)\|_{\mathrm{cb}}
&\lesssim|u_\hashtag(t_0)^k-u_\hashtag(kt_0)|  \leq k|u_\hashtag(t_0)-e^{-\gamma_\hashtag t_0}|
+|e^{-k\gamma_\hashtag t_0}-u_\hashtag(kt_0)|\\
&\lesssim(k+1)n\varepsilon^2.
\end{aligned}
\]
We drop 
 the model index for brevity, 
  and write \(\mathcal Z_t=T_t(\mathrm{id}-\mathcal P)\). Since \(T_t=T_t\mathcal P+\mathcal Z_t\),
\[
\begin{aligned}
T_t-ET_tE
&=(\mathrm{id}-E)T_t\mathcal P+ET_t\mathcal P(\mathrm{id}-E)
+\mathcal Z_t-E\mathcal Z_tE,\\
(\mathrm{id}-E)\mathcal P&=(\mathcal P-E)\mathcal P,
\quad
\mathcal P(\mathrm{id}-E)=\mathcal P(\mathcal P-E).
\end{aligned}
\]
Using \(\|T_t\mathcal P\|_{\mathrm{cb}}\lesssim1\), \(\|\mathcal P\|_{\mathrm{cb}}\lesssim1\), \(\|E\|_{\mathrm{cb}}=1\), and \cref{prop:spin-global-projection,prop:spin-spectral-decomposition},
\[
\|T_t-ET_tE\|_{\mathrm{cb}}
\lesssim\|E-\mathcal P\|_{\mathrm{cb}}+\|\mathcal Z_t\|_{\mathrm{cb}}
\lesssim n\varepsilon+n e^{-t/2}.
\]
For \(k\geq1\), complete contractivity of \(E\) and \(\iota\), together with \(\iota\mathcal K_\hashtag(kt_0)E=ET_{kt_0}E\), now gives
\[
\begin{aligned}
\|\iota\mathcal K^kE-T_{kt_0}\|_{\mathrm{cb}}
 \leq\|\mathcal K^k-\mathcal K(kt_0)\|_{\mathrm{cb}}
+\|ET_{kt_0}E-T_{kt_0}\|_{\mathrm{cb}}  \lesssim(k+1)n\varepsilon^2+n\varepsilon+n e^{-t_0/2}.
\end{aligned}
\]
Finally, \(\varepsilon=n^{-\alpha}\), \(t_0=C_t\log n\), and \(\gamma_\hashtag\sim n\varepsilon^2\) imply
\[
k\gamma_\hashtag t_0\leq S
\quad\Rightarrow\quad
(k+1)n\varepsilon^2\lesssim n\varepsilon^2+\frac{S}{t_0}
=\mathcal{O}_S((\log n)^{-1}),
\]
while \(n\varepsilon=n^{1-\alpha}\) and \(n e^{-t_0/2}=n^{1-C_t/2}\leq n^{1-2\alpha}\). This proves the last estimate. Left composition by \(\Phi=T_{t_0}\) and duality give the corresponding prepared-state diamond bound. Its \(k=0\) case follows from \cref{cor:reduced-transition}.
\end{proof}

\section{Sharpness and obstructions}\label{app:sharpness}

This appendix explains what the two approximation theorems can and cannot guarantee. We distinguish four issues:
\begin{itemize}
    \item Global cb approximation by ambient conditional expectations can fail even at slow rank $2$ (\cref{prop:sharp-ambient-obstruction}).
    \item The square-root exponent in the metastable seminorm is optimal for ambient expectations with prescribed range dimension (\cref{prop:sharp-weighted}).
    \item The cube-root cb exponent is optimal for global cb approximation using UCP idempotents (\cref{prop:sharp-cube-root}).
    \item Without a bound on slow rank, no dimension-independent approximation modulus tending to zero exists using UCP idempotents (\cref{prop:sharp-dimension-obstruction}).
\end{itemize}

\subsection{A rank-$2$ example: ambient obstruction and the square-root bound}\label{sec:globalambient}
 
For \(d=2m+1\), \(m\geq1\), let \(h=\operatorname{diag}(I_m,-I_m,0)\), and write \(P_+,P_-,P_0\) for its spectral projections. Define the trace-preserving UCP map which is self-adjoint with respect to the Hilbert--Schmidt inner product, by
\begin{equation}\label{eq:sharp-rank-two-channel}
\Phi(X)=\frac{\Tr X}{d}I_d+\frac{\Tr(hX)}{d}h.
\end{equation} 
The key feature is that \(h\) is almost fixed by $\Phi$, but its square is not.

\begin{prop} \label{prop:sharp-ambient-obstruction}
The slow Riesz projection of \(\Phi\) has rank two, and
\begin{equation}\label{eq:etam}
\eta:=\|\Phi^2-\Phi\|_{\mathrm{cb}}=\frac{d-1}{d^2}\longrightarrow0.
\end{equation}
Nevertheless, every UCP conditional expectation \(E:M_d\to M_d\) onto an ambient subalgebra satisfies
\begin{equation}\label{eq:sharp-ambient-lower}
\|E-\Phi\|_{\mathrm{cb}}\geq\frac15-\frac1d.
\end{equation}
\end{prop}
\begin{proof}
Set \(\lambda=(d-1)/d\). Direct substitution gives \(\Phi(I)=I\), \(\Phi(h)=\lambda h\), and \(\Phi(X)=0\) when \(X\) is Hilbert--Schmidt orthogonal to both \(I\) and \(h\). Hence the slow Riesz range is \(\operatorname{span}\{I,h\}\), and direct computation shows that
\[
\Phi^2-\Phi=\left(\frac1d\Tr(\,\cdot)I_d+\frac{\lambda}{d}\Tr(h\,\cdot)h\right)-\left(\frac1d\Tr(\,\cdot)I_d+\frac1d\Tr(h\,\cdot)h\right)=-d^{-2}\Tr(h\,\cdot)h.
\]
The functional \(\Tr(h\,\cdot)=\Tr(P_+\,\cdot)-\Tr(P_-\,\cdot)\) has cb norm \(d-1\): the triangle inequality gives the upper bound, and evaluation at \(h\) gives equality. This proves the formula \cref{eq:etam} for $\eta$. Complete contractivity of conditional expectations gives \(\|E(h)\|\leq1\), while \(\|E(h)-h\|\leq \|E-\Phi\|_{\mathrm{cb}}+d^{-1}\). Because \(\operatorname{Ran}E\) is an algebra, it contains \(E(h)^2\), so \(E(E(h)^2)=E(h)^2\). Also \(\|E(h)^2-h^2\|\leq2\|E(h)-h\|\), by writing \(E(h)^2-h^2=E(h)(E(h)-h)+(E(h)-h)h\). Since \(\Phi(h^2)=\lambda I\), we obtain
\[
\begin{aligned}
1-\frac1d
&=\|\Phi(h^2)-h^2\|\\
&\leq\|\Phi(h^2)-E(h^2)\|
+\|E(h^2-E(h)^2)\|+\|E(h)^2-h^2\|\\
&\leq \|E-\Phi\|_{\mathrm{cb}}+2\|E(h)^2-h^2\|
\leq5\|E-\Phi\|_{\mathrm{cb}}+\frac4d,
\end{aligned}
\]
which proves \cref{eq:sharp-ambient-lower}.
\end{proof}

The theorem's ambient completion constructed in \cref{thm:weighted-ambient} has dimension at most \(\operatorname{rank}\mathcal P_{\Phi}+1=3\). In this class, the metastable square-root exponent is also optimal. Write \(\tau_d=d^{-1}\Tr\) and \(\|X\|_{2,\tau_d}=\tau_d(X^*X)^{1/2}\).

\begin{prop} \label{prop:sharp-weighted}
There are constants \(c,C>0\) such that, for all sufficiently large odd \(d\),
\begin{equation}\label{eq:sharp-weighted-optimal}
c\sqrt{\eta}\leq
\inf_{\substack{E\text{ conditional expectation}\\\dim\operatorname{Ran}E\leq3}}
\|E-\Phi\|_{\Phi,2}^{\mathrm{cb}}
\leq C\sqrt{\eta}.
\end{equation} 
\end{prop}
\begin{proof}
The square-root upper bound is a consequence of \cref{thm:weighted-ambient}. It remains to prove that this order cannot be improved.

Let \(E\) be a conditional expectation in \cref{eq:sharp-weighted-optimal}. The maximally mixed state \(I/d\) is stationary for \(\Phi_\ast\), so it is among the metastable states $\mathscr D_{\star,\Phi}$ defined in \cref{def:metastable-seminorm}. Testing this state and the contraction \(h\) gives
\[
\begin{aligned}
\|E(h)-h\|&\leq\|E(h)-\lambda h\|+\|(\lambda-1)h\| \leq\sqrt d\,\|(E-\Phi)h\|_{2,\tau_d}+\frac1d  \\
&\note{\cref{def:metastable-seminorm}}\leq\sqrt d\,\|E-\Phi\|_{\Phi,2}^{\mathrm{cb}}+\frac1d.
\end{aligned}
\]
Here we used \(\|Y\|\leq\sqrt d\,\|Y\|_{2,\tau_d}\).

We claim by contraction that $\norm{E-\Phi}_{\Phi,2}\ge 1/(8\sqrt{d})$. 
Suppose $\|E-\Phi\|_{\Phi,2}^{\mathrm{cb}}<1/(8\sqrt d)$. For large $d$, the last bound gives $\|E(h)-h\|<1/4$. Weyl's eigenvalue perturbation inequality \cite[Theorem 11.21]{carlen2025inequalities} gives
\[
|\mu_k(E(h))-\mu_k(h)|\leq\|E(h)-h\|<\frac14,\quad k=1,\ldots,d,
\]
where $\mu_k$ denotes the $k$-th eigenvalue in decreasing order. As $h$ has eigenvalues $1,0,-1$ with multiplicities $m,1,m$, respectively, $E(h)$ has exactly one eigenvalue in $(-1/4,1/4)$. Its spectral projection $R_0$ has rank one and is a polynomial in $E(h)$, so $R_0\in\Ran E$ and $E(R_0)=R_0$.

On the other hand, \(\Tr R_0=1\) and \(|\Tr(hR_0)|\leq1\), so \cref{eq:sharp-rank-two-channel} gives \(\|\Phi(R_0)\|\leq2/d\). Test the same metastable state \(I/d\), now with the contraction \(R_0\):
\[
\|E-\Phi\|_{\Phi,2}^{\mathrm{cb}}\geq\|R_0-\Phi(R_0)\|_{2,\tau_d}
\geq\frac1{\sqrt d}-\frac2d
>\frac1{8\sqrt d}
\]
for large \(d\), a contradiction. Hence $\|E-\Phi\|_{\Phi,2}^{\mathrm{cb}}\geq1/(8\sqrt d)\gtrsim\sqrt\eta$.\end{proof}

The obstruction in \cref{prop:sharp-ambient-obstruction} concerns \emph{ambient conditional expectations}. It does not rule out a nearby UCP idempotent with a Choi--Effros range, as supplied by \cref{thm:main}. The obstruction for arbitrary UCP idempotents uses a different example with growing slow rank in \cref{sec:dimension-obstruction}.

\subsection{A rank-$9$ example: the sharpness of the global exponent}

The construction has two orthogonal copies of a $3$-dimensional logical space. One copy can be decoded exactly; the other has a small coupling to one extra physical direction. We give the second copy weight \(t\). Its decoding error is of order \(t^2\), so the idempotency defect is of order \(t^3\), but an exact idempotent must still change the UCP map by order \(t\). 
Let \(L=\mathbb C^3\), with basis \(\{u_0,u_1,u_2\}\), and set \(\mathsf E_{ij}=\dyad{u_i}{u_j}\) and \(x_{ij}=\langle u_i,xu_j\rangle\). Put
\[
H=G\oplus B\oplus\mathbb Ck,\quad
G=\operatorname{span}\{g_0,g_1,g_2\},\quad
B=\operatorname{span}\{b_0,b_1,b_2\},
\]
All displayed basis vectors are orthonormal. Let \(P_G,P_B\) be the projections onto \(G,B\), and let \(J_G,J_B:L\to H\) send \(u_i\) to \(g_i,b_i\), respectively. For \(0<t<1\), write \(c=\sqrt{1-t^2}\) and define an isometry \(W:H\to L\otimes\mathbb C^3\), with environment basis \(\{f_g,f_0,f_1\}\), by
\begin{equation}\label{eq:sharp-rank-nine-isometry}
\begin{alignedat}{3}
Wg_i&=u_i\otimes f_g, &\quad &&\quad &(i=0,1,2),\\
Wb_0&=u_0\otimes(cf_0+tf_1), &\quad Wb_j&=u_j\otimes f_0, &\quad &(j=1,2),\\
Wk&=u_1\otimes f_1. &&&&
\end{alignedat}
\end{equation} 
Define UCP maps \(\Delta:M_3\to M_7\), \(\Upsilon:M_7\to M_3\), and \(\Phi:M_7\to M_7\) by
\begin{equation}\label{eq:sharp-rank-nine-maps}
\Delta(x)=W^*(x\otimes I)W,\quad
\Upsilon(X)=(1-t)J_G^*XJ_G+tJ_B^*XJ_B,\quad
\Phi=\Delta\Upsilon.
\end{equation}

\begin{prop} \label{prop:sharp-cube-root}
For sufficiently small \(t>0\), \(\operatorname{rank}\mathcal P_{\Phi}=9\),
\begin{equation}\label{eq:sharp-cube-defect}
\eta:=\|\Phi^2-\Phi\|_{\mathrm{cb}}
=t(1-\sqrt{1-t^2})\sim t^3,
\end{equation}
and there is an absolute constant \(a>0\) such that
\begin{equation}\label{eq:sharp-cube-distance}
at\leq\inf_{\substack{F:M_7\to M_7\text{ UCP}\\F^2=F}}
\|F-\Phi\|_{\mathrm{cb}}\leq2t.
\end{equation}
Thus the exponent \(1/3\) in \cref{thm:main} cannot be increased for any rank bound \(r\geq9\).
\end{prop}
\begin{proof}
\proofstep{Step 1: defect and slow rank.}
We first compute $\eta$ in \cref{eq:sharp-cube-defect} and identify the slow range. Write $R=\Delta_\ast$, $N=\Upsilon_\ast$, and $N_G(\rho)=J_G\rho J_G^*$. Relative to $L=\CC u_0\oplus\operatorname{span}\{u_1,u_2\}$, define the channel $D$ by
\[
D\begin{pmatrix}\rho_{00}&\rho_{0b}\\\rho_{b0}&\rho_{bb}\end{pmatrix}=\begin{pmatrix}\rho_{00}&c\rho_{0b}\\c\rho_{b0}&\rho_{bb}\end{pmatrix}.
\]
Thus both diagonal blocks are unchanged, while only the coherences between the two sectors are multiplied by $c$. Indeed, $R(X)=\Tr_{\CC^3}(WXW^*)$ and $N(\rho)=(1-t)J_G\rho J_G^*+tJ_B\rho J_B^*$. Tracing out the environment in \cref{eq:sharp-rank-nine-isometry} gives
\[
\begin{aligned}
R(J_G\rho J_G^*)&=\Tr_{\CC^3}(\rho\otimes\dyad{f_g}{f_g})=\rho,\\
R(J_B\rho J_B^*)&=\Tr_{\CC^3}(WJ_B\rho J_B^*W^*)\\
&=\Tr_{\CC^3}\!\left((I_3+(c-1)\mathsf E_{00})\rho(I_3+(c-1)\mathsf E_{00})\otimes\dyad{f_0}{f_0}\right)+t^2\Tr_{\CC^3}\!\left(\mathsf E_{00}\rho\mathsf E_{00}\otimes\dyad{f_1}{f_1}\right)\\
&=(I_3+(c-1)\mathsf E_{00})\rho(I_3+(c-1)\mathsf E_{00})+t^2\mathsf E_{00}\rho\mathsf E_{00}=D(\rho),
\end{aligned}
\]
Here $WJ_B\xi=(I_3+(c-1)\mathsf E_{00})\xi\otimes f_0+t\mathsf E_{00}\xi\otimes f_1$, and orthogonality of $f_0,f_1$ removes the cross terms under the partial trace. Therefore, 
\begin{equation}\label{eq:sharp-logical-roundtrip}
RN=(1-t)\id+tD=:A,\quad RN_G=\id.
\end{equation}
For every $k\geq1$ and $X\in M_k\otimes M_3$, the two summands of $(\id_{M_k}\otimes N)(X)$ are supported on the orthogonal spaces $\CC^k\otimes G$ and $\CC^k\otimes B$. Isometric embeddings preserve nonzero singular values, thus
\[
\|(\id_{M_k}\otimes N)(X)\|_1=\|(1-t)X\oplus tX\oplus0\|_1=(1-t)\|X\|_1+t\|X\|_1=\|X\|_1.
\]
Consequently,
\begin{equation}\label{eq:sharp-defect-reduction}
\|\Phi^2-\Phi\|_{\mathrm{cb}}=\|N(A-\id)R\|_\diamond=\|A-\id\|_\diamond.
\end{equation}
Indeed, $R$ is a channel, so $\|N(A-\id)R\|_\diamond\leq\|A-\id\|_\diamond$. Conversely, using \cref{eq:sharp-logical-roundtrip} gives $\|N(A-\id)R\|_\diamond\geq\|N(A-\id)RN_G\|_\diamond=\|N(A-\id)\|_\diamond=\|A-\id\|_\diamond$, where the last equality uses the complete trace-norm isometry of $N$.

With $Z=\operatorname{diag}(1,-1,-1)$, conjugation by $Z$ fixes the diagonal blocks and negates the off-diagonal blocks. Therefore
\[
D\begin{pmatrix}\rho_{00}&\rho_{0b}\\\rho_{b0}&\rho_{bb}\end{pmatrix}=\frac{1+c}{2}\begin{pmatrix}\rho_{00}&\rho_{0b}\\\rho_{b0}&\rho_{bb}\end{pmatrix}+\frac{1-c}{2}\begin{pmatrix}\rho_{00}&-\rho_{0b}\\-\rho_{b0}&\rho_{bb}\end{pmatrix},\quad D=\frac{1+c}{2}\id+\frac{1-c}{2}\operatorname{Ad}_Z.
\]
It is easy to see that the two channels $\id$ and $\operatorname{Ad}_Z$ have diamond distance $2$. Together with \cref{eq:sharp-defect-reduction}, this gives
\[
\eta=t\|D-\id\|_\diamond=t(1-c)=\frac{t^3}{1+\sqrt{1-t^2}},\quad \frac{t^3}{2}\leq\eta\leq t^3,
\]
proving \cref{eq:sharp-cube-defect}.

Since $Z=Z^*$, $\operatorname{Ad}_Z$ is self-adjoint for the Hilbert--Schmidt inner product. Therefore, $A$ is also Hilbert--Schmidt self-adjoint. The eigenspaces of $A$ are explicitly
\[
\begin{aligned}
\ker(A-\id_{M_3})&=\operatorname{span}_{\CC}\{\mathsf E_{00},\mathsf E_{11},\mathsf E_{12},\mathsf E_{21},\mathsf E_{22}\},\\
\ker(A-(1-\eta)\id_{M_3})&=\operatorname{span}_{\CC}\{\mathsf E_{01},\mathsf E_{02},\mathsf E_{10},\mathsf E_{20}\}.
\end{aligned}
\]
These are the block-diagonal and off-diagonal matrices, respectively, with dimensions $1^2+2^2=5$ and $2+2=4$. Since $A=\Upsilon\Delta$ is invertible, every $X\in M_7$ has the unique decomposition $X=\Delta(a)+Y$, with $a=A^{-1}\Upsilon(X)\in M_3$ and $Y\in\ker\Upsilon$. Thus
 $
\Phi(\Delta(a)+Y)=\Delta(\Upsilon\Delta(a))=\Delta(Aa).
 $ 
The eigenspaces of $\Phi$ at $1$ and $1-\eta$ are the respective images under $\Delta$ of the two spans above, while its zero eigenspace is $\ker\Upsilon$, of dimension $7^2-9=40$. In a basis adapted to $(\Ran\Delta,\ker\Upsilon)$, the matrix representation of $\Phi$ is $\mqty(A&0\\0&0)$. For small $t$, the Riesz contour $|z-1|=1/2$ therefore encloses exactly the $5$ eigenvalues at $1$ and $4$ at $1-\eta$, proving $\rank\mathcal P_\Phi=9$.

For the upper bound, by direct computation we know that \(X\mapsto J_G^*XJ_G\) is a UCP left inverse of \(\Delta\). Thus
 $
F^{(0)}(X)=\Delta(J_G^*XJ_G)
$ 
is a UCP idempotent, by the same exact recovery argument used in \cref{sec:cb-proof}. Since \(\Upsilon\) differs from $J_G^\ast(\cdot)J_G$ by at most \(2t\) in cb norm, \(\|F^{(0)}-\Phi\|_{\mathrm{cb}}\leq2t\). In fact the order of the upper bound also follows from \cref{thm:main}, because \(\eta^{1/3}\sim   (t^3)^{1/3} = t\).

\proofstep{Step 2: a nearby idempotent has range algebra $M_3$.}
For the lower bound, let \(F\) be UCP and idempotent, and set \(\epsilon=\|F-\Phi\|_{\mathrm{cb}}\). We may assume \(\epsilon<1/8\). The Riesz estimate \cref{prop:riesz} gives
\(\|\mathcal P_{\Phi}-\Phi\|_{\mathrm{cb}}=\mathcal{O}(\eta)\), so \(\|F-\mathcal P_{\Phi}\|_{\mathrm{cb}}<1\) for small \(t\). Two projections at distance less than one have the same rank \cite[Chapter~I, \S4.6]{Kato1995}, as also used in the proof of \cref{prop:spin-global-projection}. Hence \(\rank F=9\). Its Choi--Effros range \(\mathcal A_F\) is therefore a $9$-dimensional \(C^*\)-algebra.

We claim \(\mathcal A_F\cong M_3\). By the structure theorem \cite[Chapter 4]{carlen2025inequalities} (see also \cref{sec:stinespring-picture}), the only alternative is a direct sum of matrix blocks of size at most $2$. To exclude this alternative, we
prove by contraction. 
Consider \(K=\Upsilon F\Delta\).  
\begin{equation}\label{eq:sharp-compressed-error}
\|K-\id\|_{\mathrm{cb}} = \norm{\Upsilon(F-\Phi)\Delta +\Upsilon\Delta \Upsilon\Delta  -\id }_{\rm cb} 
\leq \epsilon+\|A^2-\id\|_{\mathrm{cb}}
\leq \epsilon+2\eta<\frac12.
\end{equation}
If $\mathcal A_F=\bigoplus_jM_{n_j}$ with $n_j\leq2$, the adjoint channel $K_\ast$ factors as a channel from $M_3$ to $\bigoplus_jM_{n_j}$ followed by a channel back to $M_3$. By the blockwise Kraus representation in \cref{sec:stinespring-picture}, choose $B_{j\beta}:\CC^3\to\CC^{n_j}$ and $C_{j\gamma}:\CC^{n_j}\to\CC^3$ so that the two channels act as
$
\rho\mapsto\bigoplus_j\sum_\beta B_{j\beta}\rho B_{j\beta}^*, \bigoplus_j\sigma_j\mapsto\sum_{j,\gamma}C_{j\gamma}\sigma_jC_{j\gamma}^*.
$
Trace preservation gives $\sum_{j,\beta}B_{j\beta}^*B_{j\beta}=I_3$ and $\sum_\gamma C_{j\gamma}^*C_{j\gamma}=I_{n_j}$ for each $j$. Let a single index $\alpha$ enumerate the triples $(j,\gamma,\beta)$ and define $
L_\alpha:=C_{j\gamma}B_{j\beta},  K_\ast(\rho)=\sum_\alpha L_\alpha\rho L_\alpha^*,\rank L_\alpha\leq n_j\leq2.$ 
 Let $\Omega=3^{-1/2}\sum_{i=0}^2u_i\otimes u_i$ be the qutrit maximally entangled state, and $\sigma=(\id\otimes K_\ast)(\dyad{\Omega}{\Omega})$. For each $L_\alpha$, its at most two nonzero singular values give $|\Tr L_\alpha|^2\leq\|L_\alpha\|_1^2\leq2\Tr(L_\alpha^*L_\alpha)$ by Cauchy--Schwarz. Expanding $\Omega $ gives
 $
\langle\Omega,(I\otimes L_\alpha)\Omega\rangle
 =\frac13\sum_{i,j=0}^2\langle u_i,u_j\rangle\langle u_i,L_\alpha u_j\rangle 
 =\frac13\sum_{i,j=0}^2\delta_{ij}\langle u_i,L_\alpha u_j\rangle
=\frac13\sum_{i=0}^2\langle u_i,L_\alpha u_i\rangle
=\frac13\Tr L_\alpha.
 $
Together with $\sum_\alpha L_\alpha^*L_\alpha=I_3$, this yields
\[
\langle\Omega,\sigma\Omega\rangle=\frac19\sum_\alpha|\Tr L_\alpha|^2\leq\frac29\sum_\alpha\Tr(L_\alpha^*L_\alpha)=\frac23.
\]
The trace-norm duality bound, tested on the contraction $2\dyad{\Omega}{\Omega}-I$, now gives
\[
\begin{aligned}
\|K-\id\|_{\mathrm{cb}}&=\|K_\ast-\id\|_\diamond\geq\|\sigma-\dyad{\Omega}{\Omega}\|_1\\
&\geq\left|\Tr\bigl((2\dyad{\Omega}{\Omega}-I)(\sigma-\dyad{\Omega}{\Omega})\bigr)\right|=2(1-\langle\Omega,\sigma\Omega\rangle)\geq\frac23,
\end{aligned}
\]
contradicting \cref{eq:sharp-compressed-error}. Therefore $\mathcal A_F\cong M_3$.

\proofstep{Step 3: exact recovery forces orthogonal logical copies.}
Choose a unital $*$-isomorphism $\theta:\mathcal A_F\to M_3$, corestrict $F:M_7\to\mathcal A_F$, and let $\iota:\mathcal A_F\hookrightarrow M_7$ be inclusion, as in the proof of \cref{thm:qmsm-master}. Then the UCP maps $\Delta'=\iota\theta^{-1}$ and $\Upsilon'=\theta F$ satisfy
\[
\Delta'\Upsilon'=\iota F=F,\quad \Upsilon'\Delta'=\theta F\iota\theta^{-1}=\id_{M_3}.
\]
Taking adjoints gives $F_\ast=N'R'$ and $R'N'=\id$, where $N'=\Upsilon'_\ast$ and $R'=\Delta'_\ast$. For a map $T:M_3\to M_3$, define its Choi matrix by
\begin{equation}\label{eq:sharp-choi-matrix}
\operatorname{Ch}(T):=\sum_{i,j=0}^2\mathsf E_{ij}\otimes T(\mathsf E_{ij}).
\end{equation}
If $T(\rho)=\sum_\alpha L_\alpha\rho L_\alpha^*$, then $\operatorname{Ch}(T)=\sum_\alpha\dyad{\mathrm{vec}(L_\alpha)}{\mathrm{vec}(L_\alpha)}$, where $\mathrm{vec}(L)=\sum_i u_i\otimes Lu_i$ \cite[Theorem~2.22]{Watrous2018}. 
Note that for any two positive operators such that $0\le X\le Y$ we have $\Ran X\subset \Ran Y$, 
since $\operatorname{Ch}(\id)=\dyad{\mathrm{vec}(I_3)}{\mathrm{vec}(I_3)}$, every Kraus operator of the identity channel is a scalar multiple of $I_3$. Choose the Kraus representations for $N'$ and $R'$ as follows
\[
\begin{aligned}
N'(\rho)=\sum_\beta B_\beta\rho B_\beta^*, \quad B_\beta:L\to H,\quad R'(X)=\sum_\alpha C_\alpha XC_\alpha^*, \quad C_\alpha:H\to L.
\end{aligned}
\]
Then $R'N'(\rho)=\sum_{\alpha,\beta}(C_\alpha B_\beta)\rho(C_\alpha B_\beta)^*=\rho$, so the identity channel has Kraus operators $C_\alpha B_\beta$. The argument above gives $C_\alpha B_\beta=c_{\alpha\beta}I_3$. Trace preservation of $R'$ gives $\sum_\alpha C_\alpha^*C_\alpha=I_7$, hence
\[
B_\beta^*B_\gamma=B_\beta^*\left(\sum_\alpha C_\alpha^*C_\alpha\right)B_\gamma=\sum_\alpha(C_\alpha B_\beta)^*(C_\alpha B_\gamma)=\left(\sum_\alpha\overline{c_{\alpha\beta}}c_{\alpha\gamma}\right)I_3.
\]
Diagonalizing this positive coefficient matrix, and discarding zero Kraus operators, gives
\begin{equation}\label{eq:sharp-reversible-encoding}
N'(\rho)=\sum_{j=1}^s p_jV_j\rho V_j^*,\quad
V_i^*V_j=\delta_{ij}I_3,\quad
p_j>0,\quad\sum_jp_j=1,\quad s\leq2.
\end{equation}
The last inequality follows from \(3s\leq\dim H=7\). Thus an exact encoding consists of at most $2$ \emph{orthogonal} logical copies.

\proofstep{Step 4: encodings into $1$ or $2$ orthogonal copies give the lower bound.}
If $s=1$, $N'(\rho ) = V_1 \rho V_1^\ast$ with $V_1^\ast V_1 = I_3$. Every output of $F_\ast$ is supported on the three-dimensional space $V_1L$. On the other hand,
\[
\Phi_\ast\bigl(J_G(I_3/3)J_G^*\bigr)=\frac{1-t}{3}P_G+\frac t3P_B.
\]
For $t<1/2$, the three largest eigenvalues of this state sum to $1-t$. By the Ky Fan maximum principle \cite[Theorem 11.13]{carlen2025inequalities}, every rank-$3$ projection $P$ satisfies $\Tr((I-P)\Phi_\ast(J_G(I_3/3)J_G^*))\geq t$. Set $P=V_1V_1^*$ and $\rho_0=J_G(I_3/3)J_G^*$. Since $PF_\ast(\rho_0)P=F_\ast(\rho_0)$, we have $\Tr(PF_\ast(\rho_0))=1$, and both output states have trace one. Using the contraction $2P-I$ gives
\[
\begin{aligned}
\epsilon&=\|F_\ast-\Phi_\ast\|_\diamond\geq\|F_\ast(\rho_0)-\Phi_\ast(\rho_0)\|_1 \geq\Tr\!\left((2P-I)(F_\ast(\rho_0)-\Phi_\ast(\rho_0))\right)\\
&=2\Tr(PF_\ast(\rho_0))-2\Tr(P\Phi_\ast(\rho_0)) =2\left(1-\Tr(P\Phi_\ast(\rho_0))\right)=2\Tr((I-P)\Phi_\ast(\rho_0))\geq2t.
\end{aligned}
\]

It remains to consider $s=2$. The space $C=V_1L\oplus V_2L$ in \cref{eq:sharp-reversible-encoding} has codimension one. Write $C^\perp=\CC h$, with $\|h\|=1$. Put $\Lambda_j(x)=V_j^*\Delta'(x)V_j$. These are UCP maps and exact recovery gives $\sum_jp_j\Lambda_j=\Upsilon'\Delta'=\id$. 
We next use the same Choi-matrix argument as in \textbf{Step 3}. 
By \cref{eq:sharp-choi-matrix}, the Choi matrices of their adjoints are positive and their weighted sum is the rank-one Choi matrix of $\id$. Each must therefore be proportional to $\mathrm{Ch}(\id)$. Trace preservation forces $\Lambda_{j,\ast}=\id$, hence $V_j^*\Delta'(x)V_j=x$ for every $j$ and $x$.

For a unitary $u\in M_3$, $T=\Delta'(u)$ is a contraction and its compression to $V_jL$ is the unitary $V_juV_j^*$ on that space. With $P=V_jV_j^*$, this means $(PTP)^*(PTP)=P=(PTP)(PTP)^*$. Consequently,
\[
PT^*(I-P)TP=PT^*TP-P\leq0,\quad PT(I-P)T^*P=PTT^*P-P\leq0.
\]
Both left-hand sides are positive, so $(I-P)TP=0=PT(I-P)$: $V_jL$ reduces $T$. Since matrices are spanned by unitaries, $C$ reduces every $\Delta'(x)$. Trace duality then gives
\[
R'(|h\rangle\langle v|)=0,\quad F_\ast(|h\rangle\langle v|)=0
\quad(v\perp h).
\]
Using the trace-norm isometry of \(N\) established before \cref{eq:sharp-defect-reduction}, for unit \(v\perp h\) and \(\|x\|\leq1\) we obtain
\[
|\langle v,\Delta(x)h\rangle|
\leq\|R(|h\rangle\langle v|)\|_1
=\|\Phi_\ast(|h\rangle\langle v|)\|_1\leq \epsilon.
\]
Equivalently, writing \(P_{h^\perp}=I-|h\rangle\langle h|\),
\begin{equation}\label{eq:sharp-almost-reducing-line}
\|P_{h^\perp}\Delta(x)h\|\leq \epsilon\quad(\|x\|\leq1).
\end{equation}
We finish by testing this inequality on a few matrix units.

The operators \(P_i:=\Delta(\mathsf E_{ii})\) are the orthogonal projections onto
\[
H_0=\operatorname{span}\{g_0,b_0\},\quad
H_1=\operatorname{span}\{g_1,b_1,k\},\quad
H_2=\operatorname{span}\{g_2,b_2\}.
\]
The off-diagonal matrix units used below are
\begin{equation}\label{eq:sharp-matrix-unit-tests}
\begin{aligned}
\Delta(\mathsf E_{21}) =|g_2\rangle\langle g_1|+|b_2\rangle\langle b_1|,\quad
\Delta(\mathsf E_{01}) =|g_0\rangle\langle g_1|+c|b_0\rangle\langle b_1|
+t|b_0\rangle\langle k|.
\end{aligned}
\end{equation}
Put \(p_i=\|P_i h\|^2\). \cref{eq:sharp-almost-reducing-line}, applied to \(\mathsf E_{ii}\), gives \(p_i(1-p_i)\leq \epsilon^2\). Since \(\sum_i p_i=1\) and \(\epsilon<1/8\), exactly one \(p_i\) is at least \(1-2\epsilon^2\), and the other two are at most \(2\epsilon^2\). This dominant index $i$ cannot be \(0\): by \cref{eq:sharp-matrix-unit-tests}, \(\Delta(\mathsf E_{10})\) is an isometry from \(H_0\) into \(H_1\), so \cref{eq:sharp-almost-reducing-line} would give
\(\epsilon\geq\sqrt{p_0(1-p_1)}\geq1-2\epsilon^2\).
The same argument with \(\Delta(\mathsf E_{12})\) excludes index \(2\). Therefore \(p_1\geq1-2\epsilon^2\).

The first identity in \cref{eq:sharp-matrix-unit-tests} sends the \(g_1,b_1\) components of \(h\) into \(H_2\). By \cref{eq:sharp-almost-reducing-line}, their squared norm is at most \(\epsilon^2/(1-p_2)\leq2\epsilon^2\). Together with \(p_0+p_2\leq2\epsilon^2\), this gives \(1-|\langle k,h\rangle|^2\leq4\epsilon^2\). Choose the phase of \(h\) so that \(\langle k,h\rangle\geq0\). Then \(\|h-k\|\leq3\epsilon\).

Finally, the second identity in \cref{eq:sharp-matrix-unit-tests} gives \(\Delta(\mathsf E_{01})k=tb_0\). Since \(\|\Delta(\mathsf E_{01})\|\leq1\) and \(\|P_{h^\perp}-P_{k^\perp}\|\leq2\|h-k\|\), \cref{eq:sharp-almost-reducing-line} yields
\[
t=\|P_{k^\perp}\Delta(\mathsf E_{01})k\|
\leq\|P_{h^\perp}\Delta(\mathsf E_{01})h\|+3\|h-k\|
\leq10\epsilon.
\]
Thus \(\epsilon\geq t/10\) in this case as well. Together with the \(s=1\) case and the trivial case \(\epsilon\geq1/8\), this proves \cref{eq:sharp-cube-distance} for all sufficiently small \(t\).
\end{proof}

\subsection{A dimensional obstruction without bounded slow rank}\label{sec:dimension-obstruction}

The previous examples keep the slow rank fixed. We now let it grow and show that no cb approximation modulus tending to zero can hold uniformly over all UCP maps. The construction encodes $m$ classical labels into $m-1$ physical dimensions.

Fix $m\geq9$, set $d=m-1$, and identify $M_d$ with $\mathcal B(\mathcal H)$, where $\mathcal H=\mathrm{span}\{\mathbf e\}^\perp\subset\CC^m$ and $\mathbf e=\sum_{j=1}^m e_j$. Define
\begin{equation}\label{eq:sharp-simplex-frame}
v_j=\sqrt{\frac{m}{m-1}}\left(e_j-\frac{\mathbf e}{m}\right),\quad
P_j=|v_j\rangle\langle v_j|,\quad
\langle v_i,v_j\rangle=-\frac1{m-1}\ (i\ne j),\quad
\sum_jP_j=\frac{m}{m-1}I_d.
\end{equation}
Each $v_j$ is a unit vector. The associated UCP maps and physical UCP propagator are
\begin{equation}\label{eq:sharp-simplex-maps}
\begin{aligned}
\Upsilon:M_d\to\CC^m,\quad&\Upsilon(X)_j=\Tr(P_jX),\\
\Delta:\CC^m\to M_d,\quad&\Delta(a)=\frac{m-1}{m}\sum_ja_jP_j,\quad
\Phi=\Delta\Upsilon.
\end{aligned}
\end{equation}

\begin{prop}\label{prop:sharp-dimension-obstruction}
The maps in \cref{eq:sharp-simplex-maps} satisfy
\begin{equation}\label{eq:sharp-dimension-defect}
\eta:=\|\Phi^2-\Phi\|_{\mathrm{cb}}\leq\frac2m,\quad
\rank\mathcal P_{\Phi}=m,
\end{equation}
and
\begin{equation}\label{eq:sharp-dimension-distance}
\inf_{\substack{E:M_{m-1}\to M_{m-1}\text{ UCP}\\E^2=E}}
\|E-\Phi\|_{\mathrm{cb}}
\geq1-\frac{2(2m-3)}{m(m-1)}.
\end{equation}
Consequently, there is no dimension-independent function $f(\eta)\to0$ as $\eta\to0$ that bounds the distance of every almost-idempotent UCP map to a UCP idempotent.  
\end{prop}
\begin{proof}
\proofstep{Step 1: defect and slow rank.}
We first show that $\eta \le \frac2m$ and the slow dimension $\rank \mathcal P_\Phi = m$. 
As in the proof of \cref{prop:sharp-cube-root}, we first compute the logical round trip $A:=\Upsilon\Delta$. Write $\Pi(a)=m^{-1}(\sum_j a_j)I_{\CC^m}$ and $\lambda=(m-2)/(m-1)$. By \cref{eq:sharp-simplex-frame},
\begin{equation}\label{eq:sharp-simplex-roundtrip}
(A)_{ij}=\begin{cases}
1-1/m,&i=j,\\
1/[m(m-1)],&i\ne j,
\end{cases}
\quad
A=\Pi+\lambda(\id_{\CC^m}-\Pi).
\end{equation}
Since $\CC^m$ is commutative, the cb norm equals the ordinary operator norm \cite[\S1.2.6]{BlecherLeMerdy2004}, which for a scalar matrix is the maximum absolute row sum. Therefore
\begin{equation}\label{eq:sharp-simplex-roundtrip-error}
\|A-\id_{\CC^m}\|_{\mathrm{cb}}=\frac2m,\quad
\delta:=\|A^2-\id_{\CC^m}\|_{\mathrm{cb}}
=2(1-\lambda^2)\left(1-\frac1m\right)
=\frac{2(2m-3)}{m(m-1)}.
\end{equation}
Since $\Phi^2-\Phi=\Delta(A-\id_{\CC^m})\Upsilon$, complete contractivity proves the defect bound. Since $A$ is invertible, every $X\in M_d$ has the unique decomposition $X=\Delta(a)+Y$, with $a=A^{-1}\Upsilon(X)$ and $Y\in\ker\Upsilon$. In these coordinates,
\[
\Phi(\Delta(a)+Y)=\Delta(\Upsilon\Delta(a))=\Delta(Aa),\quad \Phi\longleftrightarrow\begin{pmatrix}A&0\\0&0\end{pmatrix}.
\]
Thus $\mqty(A&0\\0&0)$ is the matrix representation of $\Phi$ on the basis adapted to $(\Ran\Delta,\ker\Upsilon)$, using $\Delta$ to identify the first summand with $\CC^m$. Its nonzero eigenvalues are $1$ and $\lambda>1/2$, with multiplicities $1$ and $m-1$. Since $\eta<1/4$, the Riesz projection in \cref{prop:riesz} selects exactly these modes and has rank $m$.

\proofstep{Step 2: a nearby idempotent has an $m$-dimensional range algebra.}  
To prove \cref{eq:sharp-dimension-distance}, suppose for contradiction that a UCP idempotent $E$ satisfies $\|E-\Phi\|_{\mathrm{cb}}<1-\delta$. Equip $\mathcal A_E=\Ran E$ with the Choi--Effros product, as in \cref{cor:reduced-transition}, and denote its inclusion by $\iota$. The maps
\[
\Delta_E:=E\Delta:\CC^m\to\mathcal A_E,\quad
\Upsilon_E:=\Upsilon\iota:\mathcal A_E\to\CC^m
\]
are UCP and satisfy
\begin{equation}\label{eq:sharp-simplex-comparison}
\begin{aligned}
\Upsilon_E\Delta_E-\id_{\CC^m}
&=\Upsilon(E-\Phi)\Delta+(A^2-\id_{\CC^m}),\\
\Delta_E\Upsilon_E-\id_{\mathcal A_E}
&=E(\Phi-E)\iota.
\end{aligned}
\end{equation}
Here $E$ denotes either the ambient idempotent or its corestriction, as before. Thus the two cb norms are bounded by $\|E-\Phi\|_{\mathrm{cb}}+\delta<1$ and $\|E-\Phi\|_{\mathrm{cb}}<1$, respectively. Both round trips are invertible since the right hand sides of \cref{eq:sharp-simplex-comparison} both have cb norm less than one. The first makes $\Delta_E$ injective and the second makes it surjective, so $\dim\mathcal A_E=m$.

\proofstep{Step 3: the range algebra must be commutative.}
We argue by contradiction using the similar argument as in \textbf{Step 2} of the proof of \cref{prop:sharp-cube-root}. 
Suppose $\mathcal A_E$ has a matrix block $M_{n_j}$ with $n_j\geq2$. Using the block decomposition in \cref{sec:stinespring-picture}, choose an isometry $V:\CC^2\to\CC^{n_j}$ and define the corner inclusion $\jmath(X)=(0,\ldots,VXV^*,\ldots,0)$ and the projection onto this corner, represented in $M_2$, by $\kappa(a)=V^*a_jV$. Both maps are completely positive contractions and $\kappa\jmath=\id_{M_2}$. Set $\Lambda=\kappa\Delta_E\Upsilon_E\jmath$. By \cref{eq:sharp-simplex-comparison}, $\|\Lambda-\id_{M_2}\|_{\mathrm{cb}}\leq\|E-\Phi\|_{\mathrm{cb}}<1$.

The factorization through the classical algebra is explicitly
 $
M_2\xrightarrow{\Upsilon_E\jmath}\CC^m\xrightarrow{\kappa\Delta_E}M_2.
  $ Each coordinate of the first map is a positive functional, so $(\Upsilon_E\jmath)(X)_i=\Tr(B_iX)$ for some $B_i\geq0$. If $e_i$ is the $i$-th coordinate vector of $\CC^m$, the second map sends it to $C_i:=\kappa\Delta_E(e_i)\geq0$. Hence $\Lambda(X)=\sum_i\Tr(B_iX)C_i$. For $\Omega=(e_1\otimes e_1+e_2\otimes e_2)/\sqrt2$ in $\CC^2\otimes\CC^2$ and $\omega=\dyad{\Omega}{\Omega}$,
 $
\sigma:=(\id_{M_2}\otimes\Lambda_\ast)(\omega)=\frac12\sum_i C_i^T\otimes B_i.
 $ 
Thus $\sigma$ is a sum of positive product operators. For any $a,b\geq0$,
\[
\Tr\bigl(\omega(a\otimes b)\bigr)=\tfrac12\Tr(a^Tb)
\leq\tfrac12\Tr(a)\Tr(b),\quad
\Tr(\omega\sigma)\leq\tfrac12\Tr\sigma.
\]
The same maximally entangled-state test as in \cref{prop:sharp-cube-root} gives
\[
\begin{aligned}
\|E-\Phi\|_{\mathrm{cb}}&\geq\|\Lambda-\id_{M_2}\|_{\mathrm{cb}}
=\|\Lambda_\ast-\id_{M_2}\|_\diamond
\geq\|\omega-\sigma\|_1\\
&\geq\Tr\bigl((2\omega-I)(\omega-\sigma)\bigr)
=1+\Tr\sigma-2\Tr(\omega\sigma)\geq1,
\end{aligned}
\]
a contradiction. Hence $\mathcal A_E\cong\CC^m$.

\proofstep{Step 4: $m$ perfectly distinguishable states cannot fit in $m-1$ dimensions.} 
We now pass to states to turn exact preservation of the $m$ classical labels into a lower bound on the physical dimension. Let $e_1,\ldots,e_m$ be the minimal projections of $\mathcal A_E\cong\CC^m$. The reconstruction and compression maps from \cref{cor:reduced-transition} give states $\rho_i=E_\ast(e_i)\in\mathscr D_d$ and physical effects $p_j=\iota(e_j)\in\mathcal A_E\subset M_d$ satisfying
\[
0\leq p_j\leq I_d,\quad \sum_jp_j=I_d,\quad
\Tr(\rho_i p_j)=\Tr_{\mathcal A_E}(e_iE\iota(e_j))=\delta_{ij}.
\]
Consequently, $\supp\rho_i\subseteq\ker(I_d-p_i)$ and $\supp\rho_i\subseteq\ker p_j$ for $j\ne i$. These supports are pairwise orthogonal: for $u\in\supp\rho_i$ and $v\in\supp\rho_j$, $i\ne j$, we have $\langle u,v\rangle=\langle p_i u,v\rangle=\langle u,p_i v\rangle=0$. Each state has nonzero support, so $m\leq d=m-1$, a contradiction. This proves \cref{eq:sharp-dimension-distance}.
\end{proof}
 
\bibliographystyle{alpha}
\bibliography{ref}

\end{document}